\documentclass{article}
\usepackage{fullpage}
\usepackage{amssymb}
\usepackage{amsmath, amsthm}
\usepackage{amsfonts}
\usepackage{epsfig}
\usepackage{verbatim}
\usepackage[hyphens]{url}
\usepackage{graphics}
\usepackage{color}
\usepackage{marginnote}

\usepackage{enumerate}

\makeatletter
\@ifpackageloaded{fullpage}
  {\setlength\marginparwidth{2cm}}
  {}
\makeatother

\def\rado#1{\textcolor{red}{Rado:#1}}
\def\marcus#1{\textcolor{blue}{Marcus:#1}}
\def\cro{\mathop{\rm cr}\nolimits}
\def\pcr{\mathop{\rm pcr}\nolimits}

\def\iacr{\mathop{\rm iacr}\nolimits}
\def\iocr{\mathop{\rm iocr}\nolimits}
\def\iocro{\mbox{iocr-}$0$}

\newcommand{\ZN}{\mathbb{Z}}

\def\sct{\mathop{\rm sc}\nolimits}

\def\itemi{\item [$(i)$]}
\def\itemii{\item [$(ii)$]}

\def\itemiii{\item [$(iii)$]}

\newtheorem{theorem}{Theorem}[section]
\newtheorem{lemma}[theorem]{Lemma}

\newtheorem{corollary}[theorem]{Corollary}
\newtheorem{conjecture}[theorem]{Conjecture}

\theoremstyle{definition}
\newtheorem{remark}[theorem]{Remark}
\newtheorem{claim}[theorem]{Claim}

\newtheorem{question}[theorem]{Question}

\def\marrow{{\marginpar[\hfill$\longrightarrow$]{$\longleftarrow$}}}

\if 11 
\def\rado#1{{\color{red}\sc  Rado says: }{\marrow\sf #1}}
\else
\def\rado#1{}
\fi

\if 11 
\def\marcus#1{{\color{blue}\sc  Marcus says: }{\marrow\sf #1}}
\else
\def\marcus#1{}
\fi

\begin{document}

\title{Strong Hanani-Tutte for the Torus\thanks{An extended abstract of this paper appeared in SoCG 2021.}}

\author{
{Radoslav Fulek
} \\
{\small Department of Mathematics} \\[-0.13cm]
{\small UC San Diego} \\[-0.13cm]
{\small La Jolla, California 92093, USA} \\[-0.13cm]
{\small \tt radoslav.fulek@gmail.com}\\[-0.13cm]
\and
{Michael J.~Pelsmajer
} \\
{\small Department of Applied Mathematics} \\[-0.13cm]
{\small Illinois Institute of Technology} \\[-0.13cm]
{\small Chicago, Illinois 60616, USA} \\[-0.13cm]
{\small \tt pelsmajer@iit.edu}\\[-0.13cm]
\and
{Marcus Schaefer\thanks{This work was supported by a DePaul/CDM Faculty Summer Research Stipend Program.}
} \\
{\small Department of Computer Science} \\[-0.13cm]
{\small DePaul University} \\[-0.13cm]
{\small Chicago, Illinois 60604, USA} \\[-0.13cm]
{\small \tt mschaefer@cs.depaul.edu}\\[-0.13cm]
}

\maketitle

\begin{abstract}
If a graph can be drawn on the torus so that
every two independent edges cross an even number of times, then the graph can
be embedded on the torus.
\end{abstract}

\section{Introduction}

Two edges in a graph are \emph{independent} if they do not share a vertex.
A drawing of a graph on a surface is \emph{independently even}, or \emph{\iocro} for short, if every two independent edges cross\footnote{Crossings for us are proper crossings, not shared endpoints or touching points.} an even number of times in the drawing. Independently even drawings of graphs on surfaces are important relaxations of graph embeddings with a wide array of applications, which we will discuss in detail in Section~\ref{sec:A}.

In the plane, there is a beautiful characterization of planar graphs known as the Hanani-Tutte
theorem which says that a graph is planar, that is, it has an embedding in the plane,
if and only if it has an independently even drawing in the plane.
Equivalently, any drawing of a non-planar graph in the plane must contain
two independent edges that cross oddly.

There are several proofs of the Hanani-Tutte theorem, including the original $1934$ proof by Hanani and
the $1970$ proof by Tutte, see~\cite{PSS07} for more references. We also know that the result remains true for the projective plane\footnote{A sphere with a crosscap.
We assume that the reader is familiar with the basic terminology of drawings and embeddings in surfaces.
For background see~\cite{MT01, D05}.}~\cite{PSS09,CKP+17}.
On the other hand, counterexamples were found recently which show that the Hanani-Tutte theorem does not
extend to orientable surfaces of genus $4$ and higher~\cite{fulek2019counterexample}. An approximate version of the Hanani-Tutte theorem is true in any surface~\cite{FK19_sym,FK18_kura}; that is, for every (orientable) surface $S$ there is an (orientable) surface $S'$ so that if $G$ can be drawn on $S$ so that every two independent edges cross evenly, then $G$ can be embedded in $S'$.

We  complement these results by proving that the Hanani-Tutte theorem does extend to the torus.

\begin{theorem}\label{thm:HTtorus}
Let $G$ be a  graph. Suppose that $G$ can be drawn on the torus so that every two independent edges cross evenly. Then $G$ can be embedded on the torus.
\end{theorem}

Among orientable surfaces, this leaves only the double and triple torus, for which we do not know whether the Hanani-Tutte theorem holds. For non-orientable surfaces, all cases starting with the Klein bottle are open.

Our approach extends and refines techniques developed in~\cite{CKP+17,PSS09} and other papers, in particular in Sections~\ref{sec:RIOT} and~\ref{sec:KM}. After that, in Sections~\ref{sec:MC} and~\ref{sec:HTT} our proof requires new tools.

The proof of Theorem~\ref{thm:HTtorus} is inductive, and for the induction to work we need to strengthen the result; this strengthened version is Theorem~\ref{thm:WUHT} in Section~\ref{sec:RIOT}. For the induction, we carefully define a partial order on  drawings of graphs on the torus. We then show that no minimal counterexample to Theorem~\ref{thm:WUHT} with respect to our partial order exists. Our reductions are mostly, but not always, minor preserving.

In the base case of our proof, we work with drawings of subdivisions of $K_{3,t}$, $t\le 6$ and $K_5$, possibly with some added paths in the case of $K_{3,t}$. Reducing to the base case is a relatively smooth procedure, using natural redrawing tools up to the point when the graph is
formed by a subdivision of $K_{3,t}$ or  $K_5$ with additional simple bridges. In this case an extensive case analysis seems to us the only way to proceed. The bulk of the full version of the paper deals with these cases. The main difficulty in this part of the  proof lies in performing the reduction steps so that the case analysis becomes manageable.

\subsection{Known Results}

The Hanani-Tutte theorem~\cite{C34, T70} has been known with many proofs for the plane for a while, for a survey of known results see~\cite{S14}. It was shown to be true for the projective plane by Pelsmajer et al.~\cite{PSS09} using the excluded minors for the projective plane, and later directly, without recourse to excluded minors by \'E. Colin de Verdi\`ere et al.~\cite{CKP+17}. Excluded minors stop being useful at that exact point, since we do not know the complete list of excluded minors for the torus or any higher-order surfaces.

If we strengthen the assumption of the Hanani-Tutte theorem to also require adjacent edges to cross each other evenly, then embeddability follows, for any surfaces. This is known as the weak Hanani-Tutte theorem.\footnote{The naming of this variant as ``weak'' is somewhat misleading, in that it allows a stronger conclusion.}

\begin{theorem}[Weak Hanani-Tutte for Surfaces~\cite{CN00,PSS09b}]\label{thm:WHTS}
    If a graph can be drawn in a surface $S$ (orientable or not) so that every two edges cross an even number of times, then the graph can be embedded in $S$ with the same rotation system (if $S$ is orientable), or the same embedding scheme (if $S$ is non-orientable).
\end{theorem}

Unfortunately, the available proofs of Theorem~\ref{thm:WHTS} do not give any insight on how to establish the strong version on a surface whenever this is possible. It does not even settle the seemingly easy question whether a graph which can be drawn in a surface so that the only crossings are between adjacent edges, can be embedded in that surface.

For work on applying Hanani-Tutte to different planarity variants, see~\cite{DF21,FPSS12,FulekPS16,FulekPS17,GMS14,S13}. A variant of the (strong) Hanani-Tutte theorem in the context of approximating maps of graphs, which works on any surface, was  announced in a paper co-authored by the first author~\cite{FK18_approx}.

\section{Terminology, Definitions, and Basic Properties}
%

For the purposes of this paper, graphs are simple (no multiple edges or loops), and surfaces are compact $2$-manifolds (with or without boundary).

For a particular drawing $D$ of a graph $G$ on a surface $S$,
we make the following definitions:
The \emph{crossing parity} of a pair of edges is the number of times the two
edges cross modulo $2$. Two edges form an \emph{odd pair} if their crossing parity is $1$. A subgraph (or single edge) is \emph{even} if none of its edges belong to an
odd pair in $G$.


A closed curve $\gamma$ on a surface  $S$ is \emph{non-essential} if it forms the boundary of a (not necessarily connected) closed sub-manifold of~$S$, or equivalently, if its complement in $S$ can be two-colored so that path-connected components sharing a non-trivial part of $\gamma$ receive opposite colors. In homological terms, a closed curve $\gamma$ on  $S$ is \emph{essential} if  its homology class does not vanish over $\ZN_2$.
A closed curve {\em separates the surface} if removing the curve from the surface disconnects the surface.
A curve is {\em simple} if it is free of self-intersections.
A closed simple curve $\gamma$ separates the surface if and only if $\gamma$ is non-essential.
An essential simple closed curve is therefore {\em non-separating}.

We apply the same terminology to cycles in graph drawings, since a cycle determines a closed curve.
We say a subgraph in a drawing is {\em essential} if it contains an essential cycle.

Closed curves in the plane are non-essential.  Non-essential curves, on any surface, tend to be  easier to handle in proving results similar to ours, because they cross every other closed curve an even
number of times. (This is easy to see using the two-coloring of the complement of the curve from the above definition.) The difficulty in proving Hanani-Tutte type results lies in the presence of essential curves.

To make this more precise, suppose we are given a drawing $D$ of a graph $G$.  We associate with every closed walk $W$ of $G$ a $2$-dimensional vector $H(W)$ over $\ZN_2$ representing the 1-dimensional homology class of $W$ in $D$,
intuitively, we count how often $W$ crosses an equator and meridian of the torus, modulo $2$. The vectors $H(W)$'s satisfy the following properties, where $\oplus$ denote the symmetric difference applied to  edges.
\begin{enumerate}[(i)]
\item
$H(W)=(0,0)$ if and only if $W$ is non-essential;
\item
 $H(W_1 \oplus W_2)= H(W_1)+H(W_1)$, for every pair of closed walks $W_1$ and $W_2$; and
 \item  the number of crossings over $\ZN_2$ between closed curves that are drawings of closed walks $W_1$ and $W_2$ in $D$ (possibly after a small perturbation to achieve a generic position) is $H(W_1)^T \begin{pmatrix}
      0  & 1 \\ 1 &  0
    \end{pmatrix} H(W_2)$, that is, $W_1$ and $W_2$
    cross an odd number of times if and only if they belong to different non-vanishing homology classes.
\end{enumerate}

The following lemma rephrases a basic property of essential curves in a surface, see, for example,~\cite[Section 5.1]{FK18_kura} or~\cite[Proposition 4.3.1]{MR96}.

\begin{lemma}
\label{lem:3PC}
The family of essential cycles in a graph drawn in a surface
satisfies the $3$-path condition: given three internally disjoint paths with the same endpoints, if two of the cycles formed by the paths are non-essential then so is the third.
\end{lemma}

\begin{proof}
Let $\gamma_1,\gamma_2,\gamma_3$ be three curves from point $x$ to point $y$ in the
surface, and for each $i$ let $\gamma_i^{-1}$ be the the reversal of $\gamma_i$ (i.e., $\gamma_i$ taken
from $y$ to $x$).
If $\gamma_1\gamma_2^{-1}$ (the concatenation of $\gamma_1$ and $\gamma_2^{-1}$) and $\gamma_2\gamma_3^{-1}$
are non-essential then $\gamma_1\gamma_2^{-1}\gamma_2\gamma_3^{-1}$ is also non-essential.
The claim follows since $\gamma_1\gamma_2^{-1}\gamma_2\gamma_3^{-1}$ is in the same 1-dimensional homology class over $\ZN_2$ as $\gamma_1\gamma_3^{-1}$.
\end{proof}

\paragraph{Edge-vertex move.}
An \emph{edge-vertex} $(e,v)$-\emph{move} (also known as \emph{van Kampen's finger-move}, or \emph{edge-vertex switch})
is a generic deformation of the edge $e$ in a drawing of $G$ changing the crossing
parity between $e$ and all the edges incident to $v$, without changing any other crossing parities; see the left illustration in Figure~\ref{fig:drawmoves}.
An edge-vertex move is performed as follows. Connect an interior point of $e$
to $v$ via a curve $\gamma$ that does not pass through any vertices (and does not cross $e$). Then
reroute $e$ close to $\gamma$ and around $v$ (as shown in the illustration). Since $e$
traverses $\gamma$ twice, only the crossing parities of $e$ with edges incident to $v$ change.

\begin{figure}[htp]
\centering
\includegraphics[height=0.5in]{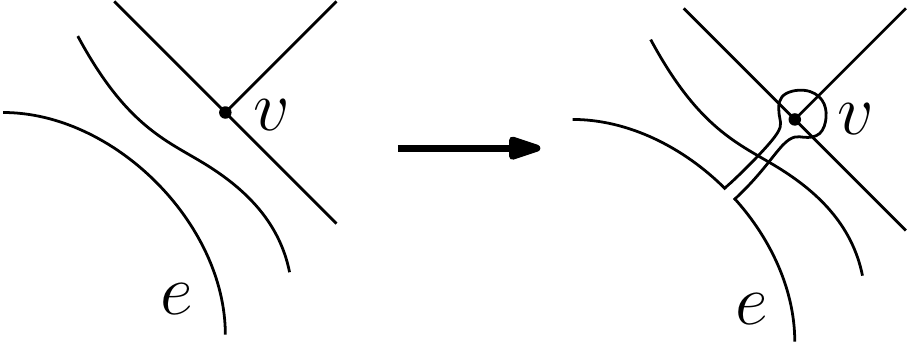} \hspace{15pt}
\includegraphics[height=0.5in]{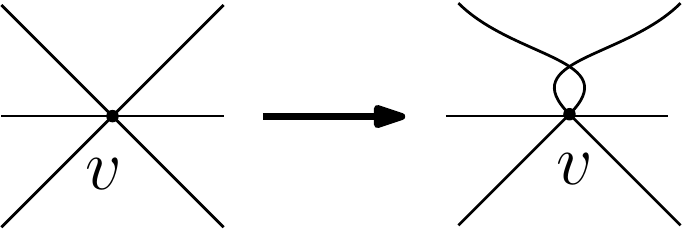} \hspace{15pt}
\includegraphics[height=0.5in]{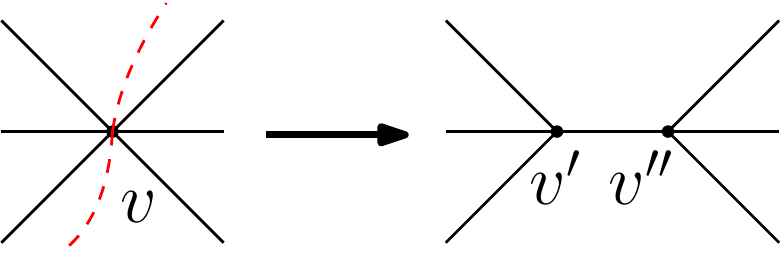}
 \caption{{\em Left:} An $(e,v)$-move. {\em Middle:} An edge-flip. {\em Right:} A vertex-split.}
\label{fig:drawmoves}
\end{figure}

\paragraph{Edge-flip.} An \emph{edge-flip}, or \emph{flip}, \emph{(at $v$)} in a drawing is a redrawing operation that happens near a vertex $v$, and which takes two consecutive edges in the rotation at $v$ and (locally) exchanges their position in the rotation at $v$; see the middle illustration in Figure~\ref{fig:drawmoves}. As a result, the crossing parity between the two flipped edges changes, and no other crossing parities are affected.

\paragraph{Edge contraction and vertex split.}
A \emph{contraction} of an  edge $e=uv$ in a drawing of a graph is an operation that turns
$e$ into a vertex
by moving $v$ along $e$ towards $u$ while dragging all the other edges incident to $v$ along $e$.
If $e$ is even, then contracting $e$ in this fashion does not change the crossing parity of any pair of edges.
Contraction may introduce multi-edges or loops at the vertices, which is why we typically only contract $uv$ partially, that is, we move $v$ close enough to $u$ so that $uv$ is free of crossings.

We will also often use the following operation which can be thought of as the inverse of contracting an edge in a drawing of a graph. To {\em split} a vertex $v$, we split its rotation into two contiguous parts,
and then cut through the vertex to separate those two parts. This results in two vertices $v'$ and $v''$ which we connect by a crossing-free edge $v'v''$ so that contracting $v'v''$ recovers
the original rotation at $v$; see the right illustration in Figure~\ref{fig:drawmoves}. Vertex-splits are not unique.

Edge-vertex moves and contractions may introduce self-crossings of edges. Such self-crossings
are easily resolved~\cite[Section 3.1]{FPSS12B}: remove the crossing, and reconnect the four severed ends so that the edge consists of a single curve. In this redrawing, essential cycles remain essential and non-essential cycles remain non-essential. (This is easy to see using the two-coloring of the complement of the curve corresponding to a non-essential cycle in the drawing.)

\section{Redrawing \iocro-Drawings on the Torus}\label{sec:RIOT}

In this section we establish some of the basic redrawing tools for \iocro-drawings in Section~\ref{sec:GRT}, some from earlier papers, show how to work with disjoint essential cycles, and give a first application of these tools to the $1$-spindle.

\subsection{Compatibility and Weak Compatibility}\label{sec:CaWC}

We say that a vertex $v$ in a drawing $D$ of a graph on a surface is \emph{even} if every two edges incident to $v$ cross each other an
even number of times; otherwise, $v$ is {\em odd}.\footnote{This will not conflict with the usual degree-based definition of even and odd
for vertices, since we will not be using that terminology.}
A drawing $D_2$ of a graph $G$  in an orientable surface $S$ is {\em compatible}  with a  drawing $D_1$ of $G$ in $S$ if every even vertex in $D_1$ is even in $D_2$, and the rotation at even vertices in $D_1$ is preserved in $D_2$. Note that the compatibility relation  is transitive, but not necessarily symmetric. We can define a notion of connectivity on even vertices: two even vertices $u$  and $v$ in a drawing of a graph are {\em evenly connected} if there exists a path connecting $u$ and $v$ consisting only of  even vertices. An {\em evenly connected component} in a drawing  is a maximal connected subgraph in the underlying abstract graph induced by a set of even vertices.

A drawing $D_2$ of a graph $G$  in an orientable surface $S$ is  {\em weakly compatible}
  with a  drawing $D_1$ of $G$ in $S$ if
   every even vertex in $D_1$ is even in $D_2$, and
  for every evenly connected component $K$ in $D_1$ either the rotation in $D_2$ at every vertex $v\in V(K)$ is the same as the rotation of $v$ in $D_1$, or the rotation in $D_2$ at every vertex $v\in V(K)$ is the reverse of the rotation of $v$ in $D_1$.
 Note that the weak compatibility relation is transitive, but not necessarily symmetric, just like compatibility.


We can now state a version of the Hanani-Tutte theorem on the plane that implies both the weak and
the strong version. While the result follows from the proof of the Hanani-Tutte theorem in~\cite{PSS07}, it was first explicitly stated, and given a new proof, in~\cite{FKP17}.

\begin{theorem}[The Unified Hanani-Tutte Theorem~\cite{FKP17}]\label{thm:HTS}
If $G$ has an \iocro-drawing in the plane, then $G$ has an embedding in the plane compatible with the \iocro-drawing.
\end{theorem}

Theorem~\ref{thm:HTS} does not hold on any surface other than the plane~\cite[Theorem 7]{fulek2019counterexample}. The counterexample requires compatibility, it fails for weak compatibility.
If we assume that the graph is $3$-connected, then weak compatibility can be achieved on the torus.


\begin{theorem}\label{thm:WUHT}
If a $3$-connected graph $G$ has an \iocro-drawing in the torus, then $G$ has an embedding in the torus that is weakly compatible with the \iocro-drawing.
\end{theorem}

Theorem~\ref{thm:WUHT} follows by Lemma~\ref{lem:Rotations} and the strengthened version of Theorem~\ref{thm:HTtorus} stated at the beginning of Section~\ref{sec:mainproof}, where we prove our main result.

\subsection{General Redrawing Tools}\label{sec:GRT}

In this section, we collect some redrawing tools which (mostly) work on all surfaces, and may be useful for establishing Hanani-Tutte type results for surfaces other than the torus. We focus on orientable surfaces.

Our first result shows that an even tree (all its edges are even) can always be cleaned of crossings. The result is true for all surfaces, and remains true for forests, but we will not need these stronger versions.

\begin{lemma}\label{lem:treeclear}
    If $G$ has an \iocro-drawing $D$ containing an even tree $T$, then there exists an \iocro-drawing of $G$ that is compatible with $D$ in
    which $T$ is free of crossings. Only edges incident to $T$ are redrawn.
\end{lemma}
\begin{proof}
 Fix a root of the tree and orient all edges towards the root. Process the edges in a depth-first traversal (any order in which ancestor edges are processed first is fine) as follows: partially contract each edge $uv$ by moving the child $v$ along $uv$ towards $u$ until $uv$ is free of crossings. This does not change the crossing parity between any pair of edges, and an edge, once contracted, remains crossing-free, since it cannot be incident to an edge contracted later on. This clears the tree of crossings.
\end{proof}

The next lemma is one of our main redrawing tools. It shows that we can always clear an essential cycle of crossings, in any surface. A precursor of this lemma, for the projective plane can be found in~\cite{PSS09}.

\begin{lemma}
\label{lem:G-C}
Let $D$ be an \iocro-drawing of a graph $G$ on a surface $S$.
Suppose that $C$ is an essential cycle in  $D$.
Then there exists an \iocro-drawing $D'$ of $G$ compatible with $D$ on $S$ in which $C$ is crossing free, and every cycle is essential in $D'$ if and only if it is essential in $D$.
\end{lemma}

\begin{proof}
We start by modifying the drawing so that every edge of $C$ is even.
This can be achieved by flipping edges incident to odd vertices of $C$: For any odd vertex on $C$
first ensure that the two $C$-edges incident to the vertex cross evenly, and then move the remaining edges at
the vertex so they cross both $C$-edges evenly; this can be done using edge-flips; we do not need to change the rotation
at even vertices.

Let $e^*$ be an edge of $E(C)$. We apply Lemma~\ref{lem:treeclear} to $E(C)-e^*$
to obtain an \iocro-drawing of $G$ in which $E(C)-e^*$ is free of crossings. Since the proof
of Lemma~\ref{lem:treeclear} is based on partial contractions, $C$ remains essential. Moreover, the underlying curve of $C$ is simple (contains no self-crossings).

We now use a $1$-dimensional variant of the Whitney trick from~\cite[Lemma 2]{FPSS12B} to clear $e^*$ of crossings (without introducing crossings along $E(C)-e^*$). Suppose some edge $f$
crosses $e^*$; by assumption, it must do so an even number of times. We cut every such $f$ at every crossing with $e^*$ and reconnect every pair of consecutive severed ends of $f$ along $e^*$ on both sides of $e^*$ and staying close to $e^*$. This operation separates $f$ into a set of curves consisting of finitely many closed curves and one curve connecting the endpoints of $f$. Since $C$ is essential (and free of self-crossings), it does not separate the surface. We can therefore reconnect the closed curves belonging to $f$ by pairs of parallel curves which do not intersect $C$, thereby turning $f$ into a closed curve joining its end vertices, making the drawing proper again. Repeating this for every $f$ crossing $e^*$, we obtain a drawing $D'$ of $G$ which is compatible with the original drawing and in $C$ is free of crossings.
\end{proof}

\begin{corollary}
\label{cor:esseven}
If $G$ has an \iocro-drawing on the torus containing an essential cycle $C$ consisting of even vertices only, then $G$ has an embedding on the torus
that is compatible with the \iocro-drawing. $C$ remains essential in the new drawing.
\end{corollary}

In Lemma~\ref{lem:HT1S} we will show that the result remains true even if at most one vertex of $C$ is odd.

\begin{proof}
 Use Lemma~\ref{lem:G-C} to free $C$ of all crossings; the embedding is compatible, $C$ remains essential, and all vertices of $C$ remain even.
 We can then cut the torus along $C$, creating two duplicate copies of $C$ in a sphere with two holes; let the new graph be $G^*$. By filling in the holes, we obtain an \iocro-drawing of $G^*$ on the sphere, in which both copies of $C$ are free of crossings (and each bounds an empty face). Apply the unified Hanani-Tutte theorem, to obtain a compatible embedding of $G^*$; we can then reidentify the two copies of $C$ by adding a handle to the sphere, obtaining a compatible embedding of $G$ on the torus.
\end{proof}

The following lemma shows, roughly speaking, that for an \iocro-drawing
it is not a single vertex that makes the difference between planarity and non-planarity.
(It appeared, with a slightly different proof, for the projective plane in~\cite{PSS09}.)

\begin{lemma}
\label{lem:G-x}
Let $x\in{V(G)}$ and $H=G-{x}$. Suppose there is an \iocro-drawing of $G$ in a surface $S$ such that $H$ is non-essential. Then $G$ is planar.
If $S$ is orientable, then the planar embedding of $H$ induced by $G$ is compatible with the original \iocro-drawing of $H$ inherited from the \iocro-drawing of $G$.
\end{lemma}

\begin{proof}
We consider the surface $S$ as a sphere with handles and crosscaps. For each crosscap we can choose a $1$-sided closed curve cutting through the crosscap; for each handle, we can pick two $2$-sided closed curves (sharing a single point) so that cutting the surface along the two curves results in a single (square) boundary hole. For the given surface $S$, we can choose a set ${\cal C}$ of $1$-sided and $2$-sided essential curves so that cutting the surface along these curves in ${\cal C}$ results in a planar surface (with a single hole for each crosscap and handle).\footnote{For non-orientable surfaces Mohar~\cite{M09} calls these a ``planarizing system of disjoint curves''. If we deformed the curves of $\cal C$ so that they all shared a single point, we'd get a set of generators of the fundamental group of the surface.}

Given an edge $uv$, we may contract the edge by pulling $u$ toward $v$ until $uv$ no longer crosses any curve of $\cal C$.  For any edge $e$ that crosses $uv$, when $u$ reaches $e$ then $e$ will be deformed so that instead of $u$ crossing $e$, $e$ will get pulled along to just stay in front of $u$, see Figure~\ref{fig:contraction}.
This may cause $e$ to add new crossings with curves of $\cal C$, two crossings at a time, whenever $u$ crosses a curve of $\cal C$.  Thus the crossing parity between each edge of $G$ not incident to $u$ and each curve of $\cal C$ will be unchanged by this operation.

\begin{figure}[t]
\centering
\includegraphics[scale=0.7]{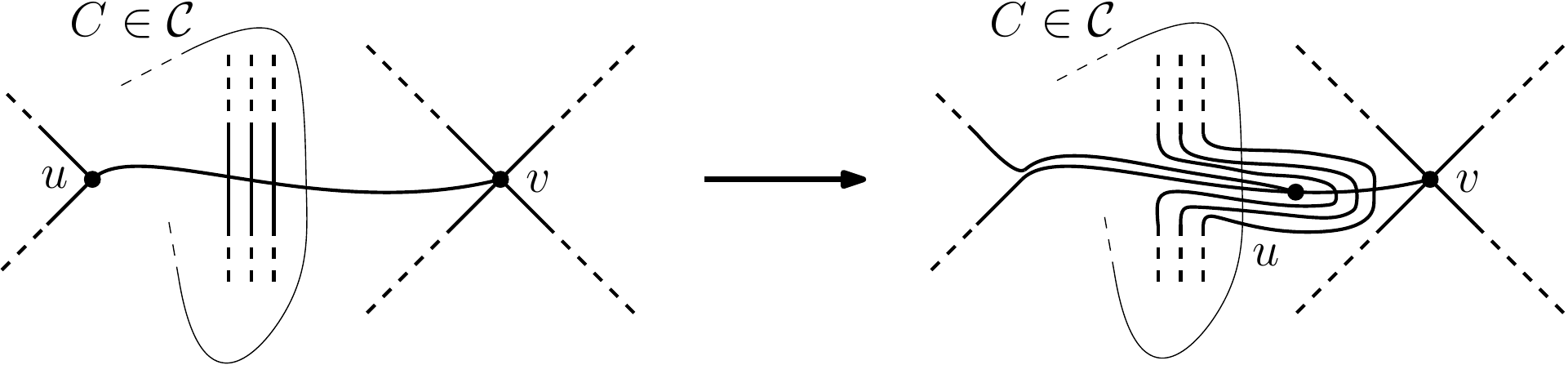}
 \caption{Contracting the edge $uv$ by pulling $u$ along $uv$ to $v$.}
\label{fig:contraction}
\end{figure}

Let $F$ be a rooted maximum spanning forest in $H := G-x$.
For each component of $F$, perform a breadth-first search transversal, contracting each edge as described towards the root of the component.  After an edge $e$ of $F$ is contracted it has zero crossings with every curve of $C$; later contractions may add crossings between $e$ and curves of ${\cal C}$ but without affecting the crossing parity. Thus, at the end of this process, every edge of $F$ crosses every curve of ${\cal C}$ evenly.

Consider any edge $e \in E(H)-E(F)$. Since we assumed that $H$ is non-essential (and that did not change during the redrawing we did), $e$ must cross each $C  \in {\cal C}$ evenly (if it did not, the (unique) cycle in $F \cup \{e\}$ is essential, a contradiction). In summary, every edge
of $H$ crosses each $C \in {\cal C}$ evenly.

Now cut the surface along the curves in ${\cal C}$, and fill in each boundary hole with a disk, resulting in a sphere. For each edge-crosscap intersection, reconnect the ends with a curve passing straight through the disk; for each edge broken at a handle,  reconnect it straight across the square-shaped disk. We remove any resulting self-crossings using the standard method~\cite[Section 3.1]{FPSS12B}. Since each edge of $H$ intersects each curve of ${\cal C}$ evenly, its crossing parity with other edges does not change.  In particular: if two edges did not cross oddly before the redrawing, but they do now, then both edges are incident to $x$, and even vertices, except possibly $x$, remain even.

We can then apply Theorem~\ref{thm:HTS} to obtain a plane drawing of $G$. If $S$ is orientable, then the subdrawing of $H$ is compatible to the original drawing of $H$ on $S$.
\end{proof}

The following is a useful consequence of Lemma~\ref{lem:G-x}.

\begin{corollary}
\label{cor:nonessH}
If $H$ has a non-essential \iocro-drawing in a surface $S$, then $H$ is planar, and, if $S$ is orientable, then the embedding of $H$ is compatible with the original drawing.
\end{corollary}
\begin{proof}
 Apply Lemma~\ref{lem:G-x} to $G = H \cup \{x\}$, where $x$ is a new vertex, not in $V(H)$.
\end{proof}

\subsection{Vertex-Disjoint and Nearly Disjoint Cycles}\label{sec:VDC}

A pair of edge-disjoint cycles $C_1$ and $C_2$ \emph{touch} at a vertex $v\in V(C_1) \cap V(C_2)$ if in the rotation at $v$ the edges
of $C_1$ and $C_2$ do not interleave; otherwise, we say the cycles {\em cross} at $v$.

A pair of \emph{nearly disjoint} cycles is a pair of edge-disjoint cycles for which any shared vertices are even and touching.

\begin{lemma}
\label{lem:2disjoint}
If $G$ has an \iocro-drawing on the torus containing two nearly disjoint essential cycles,
then $G$ has an embedding on the torus that is compatible with the \iocro-drawing, and the two nearly disjoint cycles remain essential.  
\end{lemma}

This lemma implies that we can assume that a counterexample to the Hanani-Tutte theorem on the torus does not contain two vertex-disjoint essential cycles; we extend this result in Section~\ref{sec:MC}.

\begin{proof}
Let $C_1$ and $C_2$ be two nearly disjoint essential cycles in an \iocro-drawing of $G$.

It is sufficient to prove the result for two vertex-disjoint essential cycles: If $C_1$ and $C_2$ touch at an
even vertex $v$, we can split $v$ so that $C_1$ and $C_2$ are vertex-disjoint. We can then apply the result for vertex-disjoint essential cycles,
and then merge the split vertices by contracting edges added in the splitting. Since each even vertex $v$ splits into two even vertices, whose rotations will be preserved, the final contraction recovers the original rotation of $v$, and we have compatibility for even vertices.

We can therefore assume that $C_1$ and $C_2$ are vertex-disjoint. Apply Lemma~\ref{lem:G-C} to make $C_1$ crossing-free.
Next, we redraw $C_2$ proceeding more or less as in the proof of Lemma~\ref{lem:G-C}
until the last step, which is different. Let $e^*$ be an edge on $C_2$. Following the proof of Lemma~\ref{lem:G-C} we can ensure that $E(C_2)-e^*$ is free of crossings, so we have
a drawing in which $C_1$ and $C_2$ are disjoint {\it simple} closed curves.
Every edge of $C_1\cup C_2$ is crossing-free except possibly $e^*$, which is even.
Here is where the proof differs from the proof of Lemma~\ref{lem:G-C}:

$C_1\cup C_2$ separates the surface into two {\it faces}, each homeomorphic to a cylinder.
For each edge $f$ that crosses $C_2$ (at $e^*$), cut $f$ at each crossing. Since
$f$ crosses $e^*$ an even number of times, we can reconnect the broken ends of $f$ in consecutive pairs by drawing curves alongside $e^*$ (and without crossing it).  

This does not change the crossing parity of any pair of edges, and $f$ no longer crosses $C_2$.
However, $f$ may consist of a number of components, exactly one of which is a simple curve between the endpoints of $f$, while the others are closed curves. Each of the closed curves (if there are any) lies completely in one of the two cylinder-faces.  The closed components in the same face as the simple curve can all be joined into a single curve, using pairs of ``parallel curves'', while the closed components in the other face will simply be removed.
We do this for all edges $f$ that cross $C_2$. Suppose this construction yielded two edges $f$ and $f'$ that cross oddly; then both $f$ and $f'$ would have to lie in the same cylinder-face. Then all their discarded components were closed curves in the other cylinder-face; since
any two closed curves in the plane cross evenly, removing those discards did not change the crossing parity, so $f$ and $f'$ must have crossed
oddly to start with. We have thus obtained a compatible \iocro-drawing of $G$ on the torus such that $C_1$ and $C_2$ are both crossing-free, and essential.

Let $G^*$ be a plane graph obtained by removing one of the two cylinder faces and replacing it with two disks $D_i$ each containing a vertex $v_i$ and
crossing-free edges to each vertex of $V(C_i)$, $i \in \{0,1\}$.
By Theorem~\ref{thm:HTS}, $G^*$ has a planar embedding in which the orientation of $C_1$ and $C_2$ does
not change. We can remove $v_1,v_2$ from the drawing and transfer the drawing of $G^*$ back to its face.  Doing this with the other face as well completes the proof.
\end{proof}

If the second cycle $C'$ is not essential, we cannot directly prove that $G$ has a (compatible) embedding, but the following lemma allows us to clear both cycles of crossings.

\begin{lemma}\label{lem:crossingFreeCycle}
If $G$ has an \iocro-drawing on the torus containing an essential cycle $C$ and a cycle $C'$ that is nearly disjoint from $C$, then $G$ has a compatible \iocro-drawing in which $C$ and $C'$ are crossing-free.
\end{lemma}

\begin{proof}
If $C$ and $C'$ are both essential cycles, then the result follows from Lemma~\ref{lem:2disjoint}.
So we can assume that every cycle that is nearly disjoint from $C$ is non-essential, including $C'$.

As in the proof of Lemma~\ref{lem:2disjoint}, we can first split vertices where $C$ and $C'$ touch to make them vertex-disjoint, and after proving it for that case we can recontract and recover the original graph, with rotations of even vertices preserved.  That is to say, we can assume that $C$ and $C'$ are vertex-disjoint.

Lemma~\ref{lem:G-C} allows us to redraw so that $C$ is free of crossings (without changing the rotation of any even vertex). Next, make all edges of $C'$ even by flips at its vertices, then contract all but one edge of $C'$ so that it becomes a loop. Remove self-intersections of the loop, and uncontract. We now have an \iocro-drawing of $G$ in which all crossings with $C'$ occur along a single edge (and that edge is even). Cut edges crossing that edge of $C'$ and reconnect their ends pairwise (as usual). Edges which crossed $C'$ consists of multiple components now, one of them connecting the endpoints of the edge, and the other components closed curves. Reconnect the closed curves to the edge components if they are on the same side of $C'$; otherwise drop them. We claim that the crossing parity between pairs of independent edges remains zero: if there were two edges that now cross oddly, they must be on the same side of $C'$ to do so. But two closed components belonging to these edges can only cross evenly: if they lie inside $C'$, because $C'$ is non-essential, if they lie outside $C'$, because they lie in a cylinder (as they cannot cross $C$).
\end{proof}

\subsection{The $1$-Spindle}

Before we move on to the torus, we show how to apply our tools so far in a much simpler context than the torus (or the projective plane).

We show that there is a strong Hanani-Tutte theorem for the $1$-spindle. The {\em $1$-spindle} is a pseudosurface obtained from the sphere by identifying two distinct points (the {\em pinchpoint}). Embeddings are defined as usual, for drawings we allow at most one edge to pass through the pinchpoint. In an embedding, we can always assume that there is a vertex at the pinchpoint (if not, and there is an edge passing through the pinchpoint, pull the edge until an end-vertex lies in the pinchpoint, otherwise,
the pinchpoint isn't necessary, and the graph is planar, in which case we can make an arbitrary vertex a pinchpoint).

\begin{theorem}[Hanani-Tutte for $1$-Spindle]\label{thm:HT1S}
  If a graph has an \iocro-drawing on the $1$-spindle, it can be embedded on the $1$-spindle. The embedding is compatible with the  \iocro-drawing, except for (possibly) the vertex at the pinchpoint.
\end{theorem}

It would seem that this result should follow from a Hanani-Tutte theorem for the projective plane or the torus, but that does not seem immediately obvious.

\begin{proof}
    Fix an \iocro-drawing of $G$ on the $1$-spindle. We can assume that there is a vertex $v$ at the pinchpoint. If not, there must be an edge $e$ passing through the pinchpoint (otherwise, we have an \iocro-drawing on the sphere, and we are done by the Hanani--Tutte theorem in the plane). Consider the subcurve $\gamma$ of $e$    between the pinchpoint and a vertex $v$. For any edge $f$ which crosses $\gamma$ oddly, we perform an $(f,u)$-move for every vertex $u$ in $G$. This does not change the crossing parity of any pair of edges, but it does ensure that $\gamma$ is an even curve. We can then partially contract $e$ by pulling $v$ along $\gamma$ onto the pinchpoint. Since $\gamma$ is even, the drawing remains \iocro.

 \begin{figure}[t]
\centering
\includegraphics[scale=0.8]{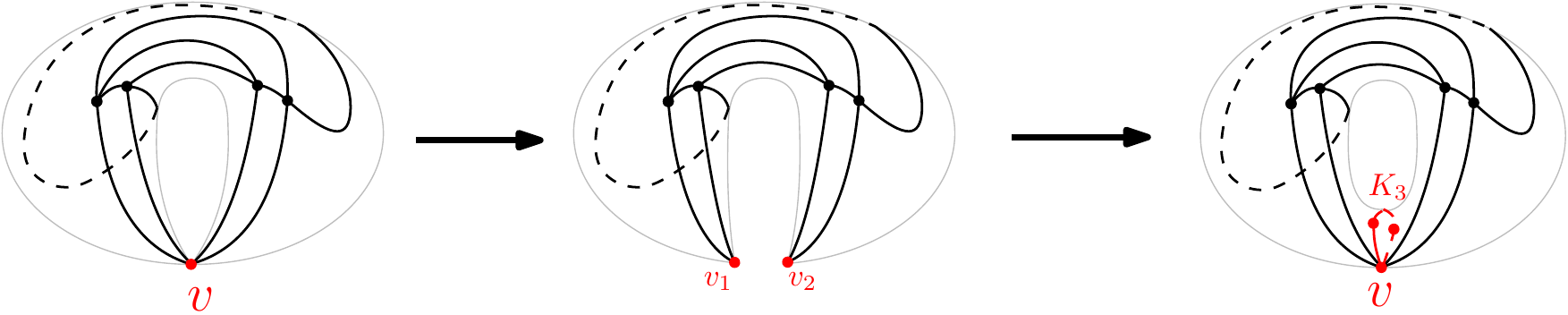}
 \caption{Surgery in the proof of Theorem~\ref{thm:HT1S}. Eliminating the pinched point splits the vertex $v$ into two vertices $v_1$ and $v_2$ that are subsequently merged using the handle. A $3$-cycle $K_3$ is introduced as a meridian of the handle passing through $v$.}
\label{fig:pinched}
\end{figure}

    Refer to Figure~\ref{fig:pinched}. Remove the pinchpoint (splitting $v$ into two copies), giving us a drawing on a sphere. Add a handle close to the two copies (close to where they were split), and move them back together, merging them into a single vertex.
    This gives us an \iocro-drawing of $G$ on the torus. Add an essential cycle, say a $3$-cycle $K_3$, through $v$ using the handle (and free of crossings), let $G^*$ denote the new graph. If the drawing of $G-v$ (as part of $G^*$) does not contain an essential cycle, then $G$ is planar, by Lemma~\ref{lem:G-x}, and we are done. So the drawing of $G-v$ contains an essential cycle, implying that the drawing of $G^*$ contains two disjoint essential cycles. By Lemma~\ref{lem:2disjoint} we can find a compatible embedding of $G^*$ in the torus. Note that in that embedding the essential cycle $K_3$ is still essential (and free of crossings), so we can contract it until it becomes a point and we have an embedding of $G$ on the $1$-spindle.
\end{proof}

The proof of Theorem~\ref{thm:HT1S} suggests the following result, which is useful by itself.

\begin{lemma}\label{lem:HT1S}
 If a graph has an \iocro-drawing on the torus containing an essential cycle $C$ such that exactly one vertex of $C$ is odd, then the graph has a compatible embedding on the torus.
\end{lemma}

Recall that Corollary~\ref{cor:esseven} deals with the case that all vertices on $C$ are even.

\begin{proof}
 Fix an \iocro-drawing of $G$ on the torus. Use Lemma~\ref{lem:G-C} to clear $C$ of crossings. By assumption there is a $v \in V(C)$ so that all vertices in $V(C)-v$ are even. Let $\gamma$ be the open, simple curve from $v$ to $v$ along $C$ (excluding $v$). Cut the torus along $\gamma$, creating two copies of each edge in $E(C)$ and all vertices in $V(C)-v$. Let $G^*$ be the resulting graph. After filling in the hole, we have an \iocro-drawing of $G^*$ in the torus. If all essential cycles in the drawing of $G^*$ pass through $v$, then $G^*$ has a compatible planar embedding, by Lemma~\ref{lem:G-x}. From that, we can reconstruct a compatible embedding of $G$ in the torus as follows: both copies of the original $E(C)$ lie on the boundary of faces in the plane embedding of $G^*$ (only the rotation at $v$ can have changed, since it is odd), so by adding a single handle, we can merge the two copies of $V(C)-v$, and drop the edges $E(C)$ of one of those two copies. This gives us a compatible embedding of $G$ in the torus.

 Otherwise, there is an essential cycle in the \iocro-drawing of $G^*$ on the torus which avoids $v$. In this case, we can add a $3$-cycle $K_3$ through $v$ which lies between the two copies of $C$, and thus us essential. This $K_3$ and the essential cycle avoiding $v$ are vertex-disjoint, and we can apply Lemma~\ref{lem:2disjoint} to find a compatible embedding of $G^* \cup K_3$ in the torus. By removing the $K_3$ (except for $v$) from this embedding, and merging the two copies of $C$, we obtain a compatible embedding of $G$ in the torus, which is what we had to show.
\end{proof}

\section{Hanani-Tutte for Some Kuratowski Minors}\label{sec:KM}

Fulek and Kyn{\v c}l~\cite{FK18_kura} showed that the Hanani-Tutte theorem is true for any surface
if we restrict ourselves to the graphs known as Kuratowski minors, which include $K_{3,t}$ for any $t\ge 3$. 

\begin{lemma}[\cite{FK18_kura}]
 \label{lem:No7Stars}
For $t\ge 7$, $K_{3,t}$ does not admit an \iocro-drawing on the torus.
\end{lemma}

Figure~\ref{fig:K6uC3} shows a toroidal embedding of $K_{3,6}$, so Lemma~\ref{lem:No7Stars} is sharp, even, as we can see in the figure, if we add a $3$-cycle on
the vertices of degree $6$. Therefore, Lemma~\ref{lem:No7Stars} implies that the Hanani-Tutte theorem on the torus is true for all $K_{3,t}$, $t \geq 7$ (and their subdivisions).

\begin{figure}[htp]
\centering
\includegraphics[scale=0.7]{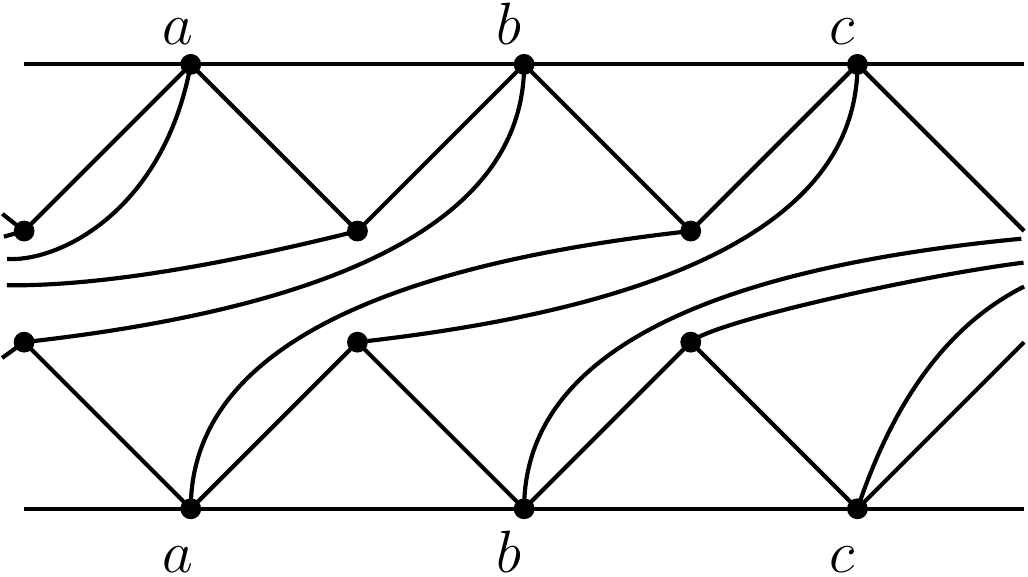}
 \caption{A torus embedding of $K_{3,6} \cup K_3$, with the union taken over the three vertices of degree $6$. The six $K_{1,3}$-claws incident to $c$ alternate connecting to $a$ and $b$ between $ab$ and $ba$.}
\label{fig:K6uC3}
\end{figure}

As a base case for our proof of Theorem~\ref{thm:HTtorus}, we need a unified Hanani-Tutte style result for Kuratowski minors on the torus.
Fulek and Kyn{\v c}l~\cite{fulek2019counterexample} showed that this is not possible, even for a $K_{3,4}$: there is an \iocro-drawing of $K_{3,4}$ on the torus which does not have a compatible embedding. We can show, however, that
there always is a {\em weakly} compatible embedding of $K_{3,n}$ up to $n = 6$ which is sharp, since $K_{3,7}$ has no embedding on the torus.

We establish a slightly stronger version. Let $a$, $b$, and $c$ be the three vertices of degree $n$ in $K_{3,n}$ (for $n=3$ pick the vertices of one of the two sides). We call a graph a {\em $K_{3,n}$ with bracers (at $a,b,c$)} if it is the union of the $K_{3,n}$ with (any number of)
interior-disjoint paths of length at most two between any two of $\{a,b,c\}$ (no multiple edges). These paths are the {\em bracers}.

\begin{lemma}\label{lem:K3nspokesHTWC} 
Let $G$ be a subgraph of a subdivision of a $K_{3,n}$ with bracers at $\{a,b,c\}$. For every \iocro-drawing of $G$ on the torus there exists a weakly compatible embedding of $G$.
\end{lemma}

The result remains true for bracers of arbitrary length, but we do not need this version.

\begin{proof}
    We can assume that $G$ contains no vertex $v$ of degree $0$ or $1$, since removing such a vertex cannot make another even vertex odd, and any embedding of $G-v$ can easily be extended to $G$. If $v$ is a vertex of degree~2 with neighbors$u,w$ and $uw\not\in E(G)$, then we can suppress $v$ by similar reasoning. Thus we may assume that $G$ is an induced subgraph of a $K_{3,n}$ with bracers.

    Let $S = \{a,b,c\}$. We want to replace bracers of length~1 with bracers of length~2. So suppose there is an edge $uw$ with $u,w \in S$. If necessary, we use edge-flips at $u$ to ensure that $uv$ crosses every edge incident to $u$ evenly. We can then subdivide $uw$ with a vertex $v$ placed very close to $u$; even vertices stay even, and after embedding we can suppress $v$. Thus we may assume that all bracers are paths of length~2. Finally, we may assume that all the vertices of degree at most $3$ are even, by performing edge-flips if necessary. So any odd vertices must belong to $S$.

    If $S$ contains no even vertices, let $K_S$ be a $K_3$ on $S$. Figure~\ref{fig:K6uC3} shows that $K_{3,6} \cup K_S$ has an embedding on the torus, which includes $G$. Each even vertex is its own evenly connected component, and any two rotations of a degree-$3$ vertex are weakly compatible, so this suffices. If $S$ contains only even vertices, we are done by Theorem~\ref{thm:WHTS}, the weak Hanani-Tutte theorem for surfaces.

    Suppose that $S$ contains exactly one odd vertex. If there is an essential cycle that avoids the odd vertex in $S$, we are done by Corollary~\ref{cor:esseven}. Otherwise, all essential cycles pass
    through the odd vertex in $S$, in which case we are done by Lemma~\ref{lem:G-x}.

    This leaves us with the case that $S$ contains two odd vertices and one even vertex. We will establish this case separately in Lemma~\ref{lem:K3nspokesHTWConeeven}.
\end{proof}

The following lemma covers the missing case in the proof of Lemma~\ref{lem:K3nspokesHTWC}.

\begin{lemma}\label{lem:K3nspokesHTWConeeven}
Let $G$ be a $K_{3,n}$ with bracers at $\{a,b,c\}$, where $n \leq 6$. If $G$ has an \iocro-drawing in which $c$ is even, and $a$ and $b$ are odd, then $G$ has a weakly compatible embedding on the torus.
\end{lemma}

\begin{proof}
    Using edge-flips, if necessary, we can ensure that all vertices of degree $3$ are even.
    Since $c$ is an even vertex, we can use Lemma~\ref{lem:treeclear} to clear all its incident edges of crossings. Every edge
    incident to $c$ belongs to either a bracer from $c$ to $a$, a bracer from $c$ to $b$, or a {\em claw} $K_{1,3}$ with legs $a$ and $b$ (apart from $c$).
    In the last case we distinguish we distinguish between {\em $ab$-claws} and {\em $ba$-claws}, depending on whether the rotation at the central claw vertex is $abc$ or $bac$ (the vertex is even). We can then describe the rotation at $c$ as a cyclic list containing
    elements $a$, $b$, $ab$, and $ba$. For example $ab,a,ba,b,a,b$ describes the rotation at $c$ for a $K_{3,2}$ with (four) bracers, see Figure~\ref{fig:spokes} (left) for an illustration. We assume for the moment that there are no bracers between $a$ and $b$, so the cyclic list completely determines the rotation at even vertices of $G$.

    \begin{figure}[htp]
    \centering
    \includegraphics[scale=0.7]{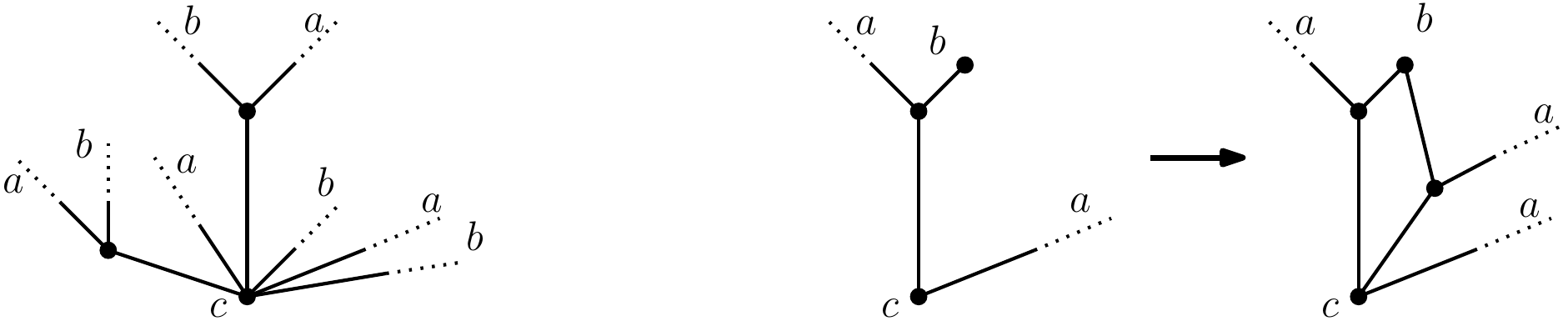}
    \caption{The cyclic permutation $\sigma$ contains the pattern $ab,a,ba,b,a,b$ (left). Introducing a claw if $ab,a$ appears in $\sigma$ (right).}
    \label{fig:spokes}
    \end{figure}

    Suppose $G$ is a counterexample to the statement of the lemma. So $G$ has an \iocro-embedding, but it does not have a weakly compatible embedding. Under this condition, let $n$ be maximal, and the number of bracers minimal. Then $n \leq 6$, since $K_{3,7}$ does not have an \iocro-embedding by Lemma~\ref{lem:No7Stars}.

    Consider the cyclic list $\omega$, constructed as above, for $c$. That list cannot contain the patterns $a,a$ or $b,b$: if it did, we could simplify to just $a$ or $b$, removing one of those bracers. The
    resulting graph has a weakly compatible embedding if and only if the original one does, since we can duplicate a bracer/delete a bracer. Hence, this would contradict the minimality of the number of bracers.

    Next we show that two consecutive items in $\omega$ cannot have different letters before and after the comma separating them. In other words, the following eight patterns do not occur: $a,b$; $b,a$; $a,ba$; $b,ab$; $ab,a$; $ba,b$; $ab,ab$ and $ba,ba$. Suppose that $\omega$
    contains one of those patterns. Without loss of generality, we can assume that in the pattern an item ending in $b$ is followed by an item starting with $a$; let $cvb$ and $cwa$ be the corresponding paths, belonging to a bracer or a claw. (See Figure~\ref{fig:spokes}, center image, for an example with pattern $ab,a$.) In the rotation at $c$ between $cv$ and $cw$, draw a small crossing-free edge $cz$, then draw edges $zb$ and $za$ alongside the paths $cvb$ and $cwa$. Since $cv$ and $vb$ cross all non-adjacent edges evenly and $v$ is even, $zb$ also crosses all non-adjacent edges evenly, and likewise for $za$. We have obtained an \iocro-drawing $D'$ be of a new graph $G'$. Note that $G'$ contains $K_{3,n+1}$. Since $K_{3,7}$ does not have an \iocro-drawing, $n\le 5$. An embedding of $G$ that is weakly compatible with its given drawing can easily be extended to an embedding of $G'$ (see Figure~\ref{fig:spokes}, right side) which is weakly compatible with its drawing, since 
    $a$ and $b$ are odd and the remaining vertices are even. But this contradicts the choice of $n$.

    It follows that any $a$ in $\omega$ must appear as part of the pattern $ba,a,ab$ and any $b$ must appear as part of $ab,b,ba$ and consecutive claws are either $ab,ba$ or $ba,ab$.
    In other words, claws alternate between $ab$- and $ba$-claws. We can then take the standard embedding of $K_{3,6}$ shown in Figure~\ref{fig:K6uC3}, for which the claws alternate as well, and add the $a$- and $b$-bracers into the embedding as specified by $\omega$; note that this correctly reproduces the rotations at all even vertices. Finally, we can add any number of bracers between $a$ and $b$ into the embedding, since $a$ and $b$ lie on the same face (and both are odd, so their rotation does not matter).
\end{proof}

\begin{remark}\label{rem:K3nspokesHTWC}
 We will apply Lemma~\ref{lem:K3nspokesHTWC} in a situation where $a$ and $b$ are part of an essential cycle. The construction in Lemma~\ref{lem:K3nspokesHTWConeeven} makes this cycle non-essential, but it is possible to keep the cycle essential. This requires a slightly sharper analysis, using that $n$ can be at most $4$, and not $6$ in this case.
\end{remark}


We also need a unified Hanani-Tutte theorem for $K_5$.

\begin{lemma}\label{lem:K5HTWC}
  For every \iocro-drawing of a subgraph $G$ of a subdivision of $K_5$ on the torus, there exists an embedding on the torus that is compatible with the \iocro-drawing.
\end{lemma}
\begin{proof}
 Suppressing a subdivision vertex does not change the \iocro-ness of a drawing, and a suppressed vertex can always be reintroduced in an embedding, so we can assume that there are no subdivision vertices. We can also remove any vertices of degree at most $1$, since they do not affect embeddability. It follows that $G$ is a subgraph of $K_n$, for some $n \leq 5$. For $n \leq 4$, the result is easy, so we assume that $G$ is a subgraph of $K_5$.

 Let $S$ be the set of even vertices in the \iocro-drawing of $G$. We distinguish cases by the cardinality of $S$.\\

 {\bfseries\noindent Case $|S| \in \{0,1, 5\}$.} If $S$ is empty we can embed $G$ as part of a toroidal embedding of $K_5$; if $S$ contains all vertices we are done using
 Theorem~\ref{thm:WHTS}; if $S$ contains a single even vertex we take a toroidal embedding of $G$ and map the vertex and its neighbors to the embedding to get a compatible embedding.\\

 {\bfseries\noindent Case $|S| = 4$.} We distinguish two subcases: if there is an essential cycle consisting of vertices in $S$, we are done by Corollary~\ref{cor:esseven}. Otherwise, all essential cycles in the drawing must contain the fifth vertex not belonging to $S$. In that case, we are done by Lemma~\ref{lem:G-x}.\\

 {\bfseries\noindent Case $|S| = 2$.} In this case we construct an embedding of $G$ that is compatible with the rotation at the two even vertices. Since we are not using the given \iocro-drawing, we can assume that $G=K_5$, and then delete the edges not present in $G$ in the end.

 Let the two even vertices be $u$,$v$ and label the remaining vertices $a$, $b$, and $c$, so that the rotation at $u$ is $vabc$. Draw $uv$ in the plane, and to each of $u$ and $v$ attach half-edges to $a$, $b$, and $c$ (this will look like $+\!\!\!-\!\!\!+$).

 We first consider the case that the rotation at $v$ is not $uabc$. Then there are at least two half-edges at $v$ which can be connected to the half-edges emanating from $u$ without crossings. The (at most one) remaining pair of half-edges can be connected
 using a handle. This leaves us with the three edges between $a$, $b$, and $c$. Edges $ab$ and $bc$ can be drawn close to $u$, since
 half-edges for $ua$ and $ub$ as well as $ub$ and $uc$ are consecutive at $u$. Edge $ac$ can be added at $v$, if $va$ and $vc$ are consecutive
 in the rotation at $v$. If not, then the rotation at $v$ must be $ucba$, since we excluded the rotation $uabc$ at $v$. Then half-edges could be connected in the plane, and we use the handle for the edge $ac$.
 See Figure~\ref{fig:K5ucba} (left) for an illustration of this case.
  (The square with the cross inside represents the handle in the figures.)

\begin{figure}[htp]
\centering
\includegraphics[scale=0.7]{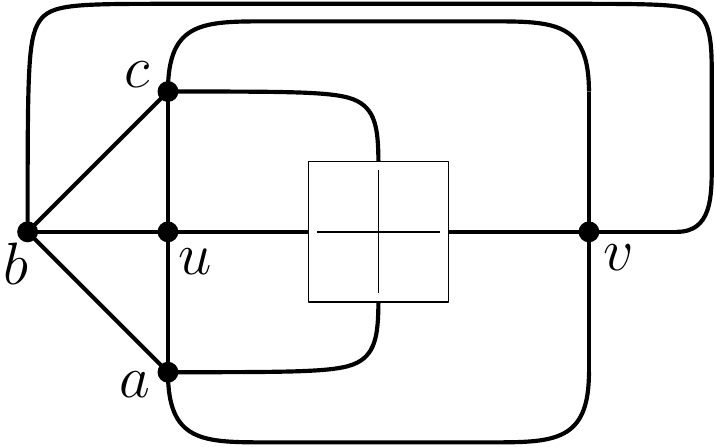} \hspace{20pt}
\includegraphics[scale=0.7]{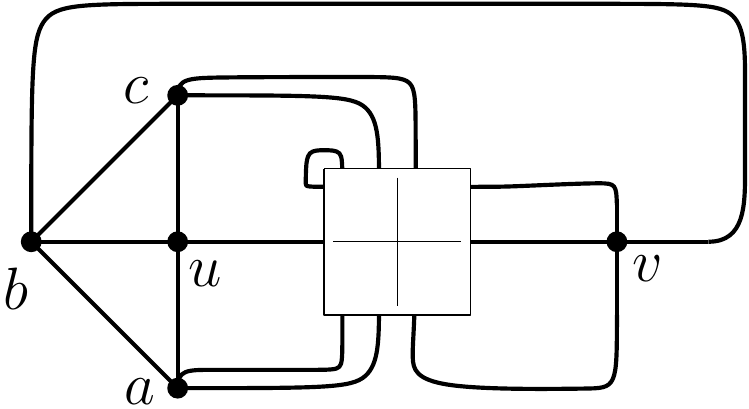}
 \caption{{\em (Two even vertices)} {\em Left:} Compatible embedding of $K_5$ in the case that $v$ has rotation $ucba$.{\em Right} Compatible embedding of $K_5$ in the case that $v$ has rotation $uabc$. We use $\boxplus$ for the (orientable) handle. }
\label{fig:K5ucba}
\end{figure}

 This leaves us with the case that the rotation at $v$ is $uabc$. Figure~\ref{fig:K5ucba} (right) shows a compatible embedding of the $K_5$ for this case.\\

 {\bfseries\noindent Case $|S| = 3$.} Let $u$, $v$, $w$ be the even vertices, and label the two remaining vertices $a$ and $b$.
 Let us first assume that $G$ contains the cycle $uvw$. If the cycle on $uvw$ is an essential cycle, we are done by Corollary~\ref{cor:esseven}. We can therefore assume that the cycle is non-essential or one of its edges is missing in $G$. If the cycle $uvw$ exists in $G$ we use Lemma~\ref{lem:treeclear} to clear $uv$ and $vw$ of crossings. It follows that $uw$ does not cross either $uv$ and $vw$, so the cycle on $uvw$ is a simple closed curve which, since it is not essential, bounds a plane region.

 \begin{figure}[htp]
\centering
\includegraphics[scale=0.7]{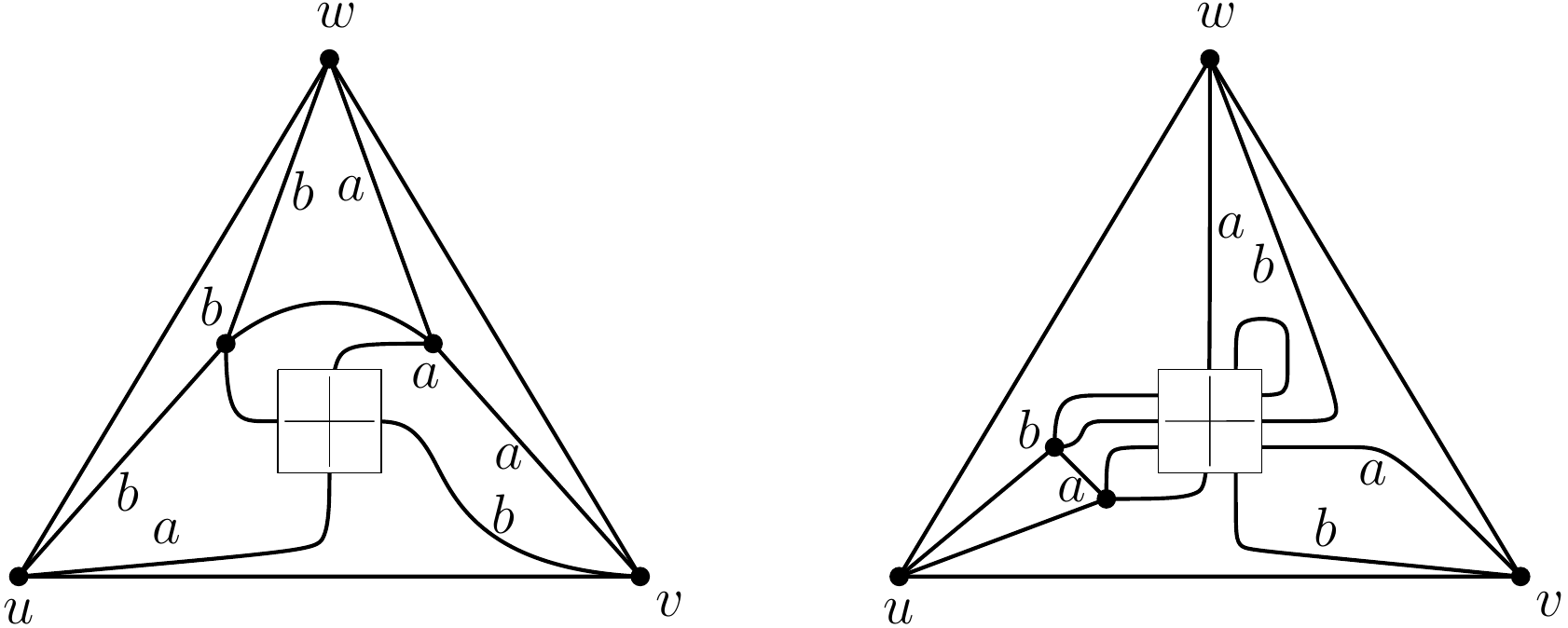}
 \caption{{\em (Three even vertices)} Compatible embeddings of $K_5$ in the case that half-edges leaving the triangle are $abbaab$ and $ababab$.}
\label{fig:K5abs}
\end{figure}

 If the  edge $ab$ exists it crosses the cycle an even number of times,
 so both $a$ and $b$ must lie on the same side of the cycle.  Since each edge incident to $a$ and $b$ also crosses the cycle an even number of times, all edges attaching to the cycle attach on the same side of the cycle (all inside the plane region, or all outside). This allows us to build a compatible embedding of $G$ as follows. In the plane draw a crossing-free cycle $uvw$, and add half-edges, as determined by the rotation, starting at the cycle towards $a$ and $b$. As we argued above, we can assume that $a$ and $b$, and the half-edges lie on the same side, without loss of generality, inside the cycle. If we write down the half-edges in the order in which they occur along the cycle $uvw$, there are only two possible results (up to renaming of $a$ and $b$ and shifting the list cyclically), namely $ababab$, and $abbaab$. Figure~\ref{fig:K5abs} shows a compatible embedding for each case, completing the proof.

 If the edge $ab$ does not exist and $a$ and $b$ are not on the same side of the cycle, the previous argument does not apply only when the two edges attaching at a vertex of the cycle are on its opposite sides, in which case we can easily construct a compatible embedding of $G$.

 \begin{figure}[htp]
\centering
\includegraphics[scale=0.7]{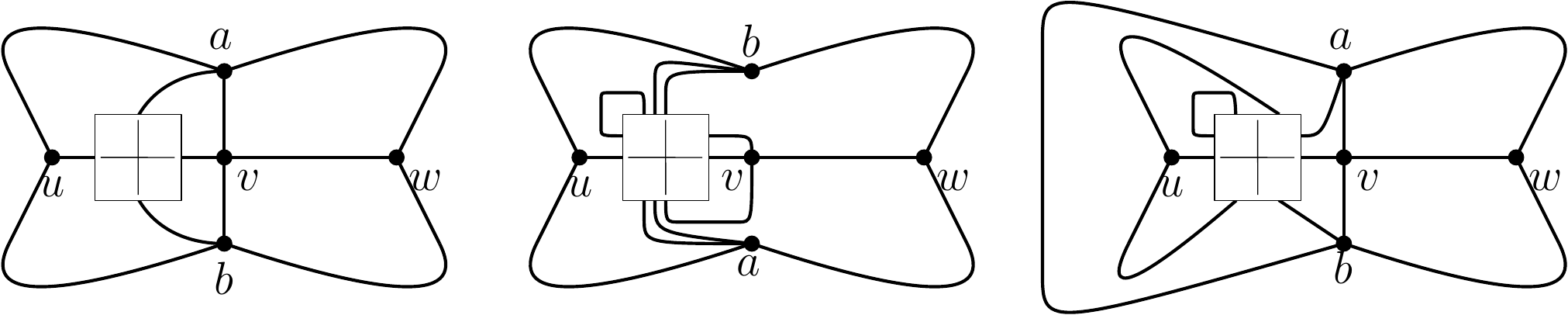}
 \caption{{\em (Three even vertices)} Compatible embeddings of $G$ in the case $uw\not\in E(G)$.}
\label{fig:K5uw}
\end{figure}

 Finally, if $uvw$ is not a cycle in $G$, the above argument applies (by an edge deletion in the end) unless $uvw$ is a path in $G$, $uw\not\in E(G)$, $v$ is a degree-4 vertex, and its two remaining edges attach on the opposite sides of the path $uvw$ (otherwise, we can add an even edge $uv$ into the drawing). In this case, up to symmetry there are three cases to consider depending on the rotations at $u$ and $w$, each of which can be easily completed into a compatible embedding of $G$, see Figure~\ref{fig:K5uw}.
\end{proof}

\section{Minimal Counterexamples}
\label{sec:MC}

In this section we collect properties of a minimal counterexample to the Hanani-Tutte theorem on the torus.
We order graphs first by their number of isolated vertices ({\bf in}creasing), then by the number of edges ({\bf in}creasing), and finally by the number of vertices ({\bf de}creasing).
We denote by $\prec$ the strict partial ordering based on these three numbers, and refer to a $\prec$-minimal graph (with certain properties). In a later section, we
will refine $\prec$ by an ordering $\prec'$; that is, whenever $G \prec H$, then $G \prec' H$.
In particular, any $\prec'$-minimal graph with a certain property is also $\prec$-minimal with that property. So the properties
proved for $\prec$-minimal graphs in this section are also true for $\prec'$-minimal graphs, which is why we
write minimal, rather than $\prec$- or $\prec'$-minimal in this section.

Since isolated vertices do not affect embeddability on a surface, or the Hanani-Tutte
criterion, a minimal counterexample contains no isolated vertices. Graphs without isolated vertices
are then ordered by number of edges (increasing), and number of vertices (decreasing); in other words, we use the (strict, partial) lexicographic order on $(|E(G)|, 2|E(G)|-|V(G)|)$; since graphs without isolated vertices satisfy $|V(G)| \leq 2|E(G)|$, this order is well-founded.
Note that if $H$ is a proper subgraph of $G$, then $H \prec G$, simply because it has fewer edges (since there are no isolated vertices).

Section~\ref{sec:WDC} presents some basic properties of minimal counterexamples with respect to cycles and cuts. In Section~\ref{sec:XC} we identify what we call an $X$-configuration, which must occur in a minimal counterexample to Theorem~\ref{thm:HTtorus}, the Hanani-Tutte theorem on the torus.
One issue we will face is that our proof of Theorem~\ref{thm:HTtorus} requires a strengthened assumption on (weakly) compatible embeddings, for which we do not know how to show that an $X$-configuration occurs. We start addressing this issue in Section~\ref{sec:PX}, which extends the proofs of Section~\ref{sec:WDC} to show that
they (mostly) still hold if the presence of an $X$-configuration is required.

\subsection{Nearly Disjoint Cycles and Cuts}\label{sec:WDC}

The following lemma is true when read with the bracketed weakly, or without it.\footnote{It is tempting to assume that the weakly compatible version of the lemma implies the compatible version, but this is not necessarily the case, since the minimal counterexample in the two cases may be different, and therefore have different properties.}

\begin{lemma}\label{lem:reduction333}
If $G$ is a minimal graph that has an \iocro-drawing $D$ on the torus, but {\em does not} have a [weakly] compatible embedding on the torus, then $D$ does not contain a pair of nearly disjoint cycles at least one of which is essential.
\end{lemma}

For the weakly compatible version, every bracketed weakly in the proof needs to be read. For the compatible version, all bracketed occurrences
of weakly have to be dropped.

\begin{proof}
Let $C$ and $C'$ be nearly disjoint cycles in $D$. By Lemma~\ref{lem:2disjoint} not both of them can be essential, so let us assume that
$C$ is essential, and $C'$ is non-essential.

Using Lemma~\ref{lem:crossingFreeCycle}, we can clear both $C$ and $C'$ of crossings.
Let $H$ be the part of $G$ drawn within the disk bounded by $C'$, and including $C'$.
Then $H$ has a compatible planar embedding, by Theorem~\ref{thm:HTS},
in which $C'$ bounds the outer face (to see this, add a new vertex outside $C'$ and connect it to every vertex on $C'$ before applying  Theorem~\ref{thm:HTS}, the resulting drawing is still \iocro\ in the plane).

 \begin{figure}[h]
\centering
\includegraphics[scale=0.7]{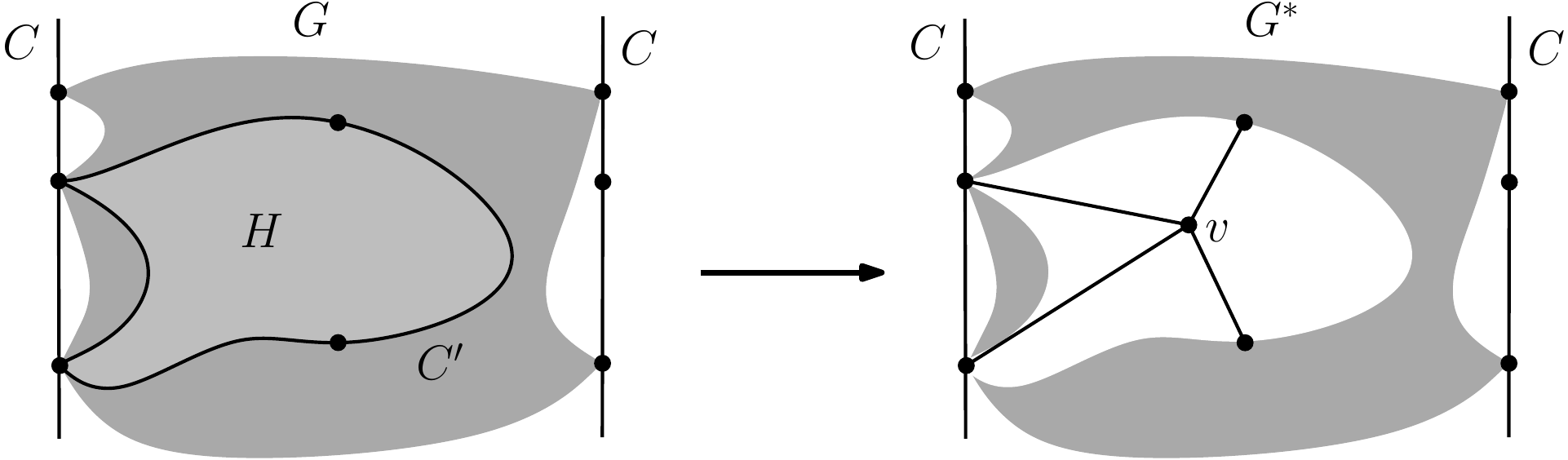}
 \caption{Construction of the drawing of the reduced graph $G^*$ in the proof of Lemma~\ref{lem:reduction333}.}
\label{fig:reduction333}
\end{figure}

Refer to Figure~\ref{fig:reduction333}.
Construct $G^*$ from $G$ as follows: remove the inner vertices of $H$, that is $V(H)-V(C')$, as well as $E(C')$, the edges of $C'$; in other words,
we clear the disk bounded by $C'$, including the edges of $C'$, but leaving its vertices. Add a new vertex $v$ in the disk, and connect
it to all vertices of $C'$. Then $v$ is even, since all its incident edges are free of crossings.

The drawing of $G^*$ is \iocro, and $G^* \prec G$ in our graph
ordering: if $H$ contains a vertex apart from $V(C')$, because the number of edges decreased; otherwise, $H = C'$, and the number of edges remained the same, but the number of vertices increased.

By minimality, $G^*$ has a [weakly] compatible embedding on the torus.

All of our redrawing operations and reductions in Lemma~\ref{lem:G-C}, Lemma~\ref{lem:2disjoint},  and Lemma~\ref{lem:crossingFreeCycle}
yield compatible drawings. We can therefore remove $v$ and edges incident to it from the drawing, and reinsert $H$ (in the weakly compatible case we may have to flip all of $H$) to obtain a [weakly] compatible embedding of the original graph $G$ on the torus, completing the proof.
\end{proof}

We use the previous lemma repeatedly to thin out $G$ in order to show that no minimal counter-example to the Hanani-Tutte theorem on the torus exists.

The following lemma allows us to focus on $2$-connected (for compatibility and weak compatibility) or $3$-connected counterexamples (in general).

\begin{lemma}
\label{lem:nearly-3-connected}
Suppose $G$ is a minimal graph that has an \iocro-drawing $D$ on the torus, but {\em does not} have a [weakly] compatible embedding on the torus. Then $G$ is $2$-connected and no 2-cut consists of two odd vertices in $D$.
\end{lemma}

If we drop the requirement that the embedding be compatible or weakly compatible, we can even conclude that the counterexample has to be $3$-connected.

\begin{corollary}
\label{cor:3-connected}
Suppose $G$ is a minimal graph that has an \iocro-drawing on the torus, but {\em does not} have an embedding on the torus. Then $G$ is $3$-connected.
\end{corollary}

The two cases in the lemma, and the third case in the corollary essentially can be established using parallel arguments, so we will present them together in a single proof. To capture the third case (where we do not assume compatibility or weak compatibility, we write ``a [[weakly] compatible embedding'', which can be read in three ways: ``a compatible embedding'', ``a weakly compatible embedding'', and ``an embedding'' (covering the case of the corollary).

\begin{proof}[Proof of Lemma~\ref{lem:nearly-3-connected} and Corollary~\ref{cor:3-connected}]
If $G$ is disconnected, and some component $K$ of $G$ does not contain an essential cycle, then
$K$ is planar and has a compatible embedding by Corollary~\ref{cor:nonessH}. By minimality of $G$, there also
is a [[weakly] compatible] embedding of $G-V(K)$ on the torus, and the embeddings of
$K$ and $G-V(K)$ can be combined into a single [[weakly] compatible] embedding of $G$.
If every component of $G$ contains an essential cycle, we can apply Lemma~\ref{lem:2disjoint} to obtain a compatible embedding. Therefore, $G$ is connected.

Suppose $G$ has a cut-vertex $x$. Then $G-x$ consists of at least two components; if each component of $G-x$
contains an essential cycle, then Lemma~\ref{lem:2disjoint} implies that $G$ has a compatible embedding, which is a contradiction.
Hence, some component $K$ of $G-x$ does not contain an essential cycle. If $x$ is odd, we use Lemma~\ref{lem:G-x}
to argue that $L=G[V(K)\cup x]$ has a planar embedding in which the embedding of $K$ is compatible with its original drawing. The graph $L$ can be embedded in a disk with $x$ on its boundary, and $G-V(K)$ has, by minimality of $G$, a [[weakly] compatible] embedding on the torus. Since $x$ is odd, we are allowed to change the rotation at $x$, so
we can combine the embeddings of $L$ and $G-V(K)$ into a single [[weakly] compatible] embedding of $G$ on the torus.

If $x$ is even, we need to argue more carefully, since we are not allowed to change the rotation at $x$. This means we cannot make use of Lemma~\ref{lem:G-x}. By Lemma~\ref{lem:reduction333} we know that $G$ does not contain two nearly disjoint cycles at least
one of which is essential. There must be at least one component $K$ of $G-x$ for which
$L=G[V(K)\cup x]$ contains an essential cycle $C$ (otherwise $G$ contains no essential cycle, and therefore has a compatible planar embedding by Corollary~\ref{cor:nonessH}).

If $C$ is contained in $K$, then $T = G[V(G)-V(K)]$ has to be a tree (by Lemma~\ref{lem:reduction333}). This is not possible, since the leaves of $T$ are degree-$1$ vertices in $G$, but a minimal $G$ does not contain vertices of degree at most $1$.

Therefore, the essential cycle $C$ must pass through $x$ in $L$. There can be no components
$K'$ of $G-x$ for which $L' = G[V(K') \cup x]$ does not contain a cycle. (Otherwise, $L'$ is a tree, which, as above, leads to degree-$1$ vertices in $G$, which do not exist in a minimal counterexample.) Therefore, if $K'$ is a component of $G-x$ other than $K$, then
$L' = G[V(K') \cup x]$ must contain a cycle $C'$. Because of Lemma~\ref{lem:reduction333} again, this cycle
cannot be nearly disjoint from $C$, so it must pass through $x$ and cross $C'$ in $x$. This forces $C'$ to be essential (the underlying curves of $C$ and $C'$ cross oddly, since $x$ is even, so both curves must be essential).
We clear $C$ of crossings using Lemma~\ref{lem:G-C}. We then use the same procedure described in Lemma~\ref{lem:G-C} to clear
$C'$ of crossings. This is possible in this particular case, since $C$ and $C'$ intersect in an even vertex $x$ (so the rotation of $x$ does not need to be modified for applying Lemma~\ref{lem:G-C}), and $C$ does not separate the torus, and neither does $C \cup C'$ once the crossings along $C'$ have been removed. We thus obtain a compatible drawing in which both $C$ and $C'$ are free of crossings. We know that $K'$ is a tree (if $K'$ contained a cycle, it would be vertex-disjoint from $C$, contradicting Lemma~\ref{lem:reduction333}. The leaves of $K'$ must be adjacent to $x$, otherwise they are degree-$1$ vertices in $G$ and would not occur in a minimal counterexample. We claim that $K'$ has to be $C'-x$. If that is not the case, then there must be a path $P$ from $x$ to a vertex $y$ on $C'-x$ otherwise avoiding $C'-x$. Then $P$ together with the path from $y$ to $x$ on $C'$ that ends on the same side of $C$ as $P$ yields a cycle nearly disjoint from $C$ and touching it (at $x$) which is a contradiction. Therefore $K'$ is just $C'-x$, so $L' = C'$.
Since $x$ is even, and we can make all other vertices of $C'$ even (they have degree $2$), we can apply Corollary~\ref{cor:esseven} to obtain a compatible embedding of $G$ in the torus.

Hence, we may assume that $G$ is two-connected.  Let $\{x,y\}$ be a cut-set.
If $K$ and $K'$
are components of $G-\{x,y\}$, by Lemma~\ref{lem:2disjoint} we may assume that either
$G[V(K)\cup \{x\}]$ or $G[V(K')\cup \{y\}]$ contains no essential cycle: say, the former.
Let $L(K)=G[V(K)\cup\{x,y\}]+xy$ where $xy$ is an edge from $x$ to $y$,
drawn along an $x,y$-path in $G-V(K)$, added only if $G[V(K)\cup\{x,y\}]$ doesn't already contain the edge $xy$.
This gives us an \iocro-drawing of $L(K)$ in which every essential cycle in it has to pass through $x$.
By Lemma~\ref{lem:G-x}, $L(K)$ is planar, and has an embedding in which the embedding of $K$ is compatible with its original drawing.  We can draw $L(K)$ in a small
disk, with $xy$ along the boundary of the disk.  Let $G^*=G-V(K)+xy$,
with the edge $xy$ 
drawn
along an $x,y$-path in $G[V(K) \cup \{x,y\}]$.
$G^*$ has an \iocro-drawing, so by minimality of $G$, we
have a [[weakly] compatible] embedding of $G^*$ in the torus.
We remove a small disk from the torus minus $G^*$, such that the
boundary of disk intersects $G^*$ at $xy$, and replace it with the disk containing $L(K)$ so that
the two $xy$'s match up. We then delete $xy$ from the drawing if it does not belong to $G$.
This gives us an embedding of $G$. If $x$ and $y$ are odd, rotations at $x$ and $y$ can be changed, and
the embedding we obtained is [weakly] compatible to the original drawing.
\end{proof}


\subsection{The $X$-Configuration}\label{sec:XC}

The Hanani-Tutte theorem holds for cubic graphs on arbitrary surfaces; this is because vertices of degree $3$ can be made even by edge-flips, at which point Theorem~\ref{thm:WHTS}, the weak Hanani-Tutte theorem for surfaces can be applied.

Edge-flips are no longer sufficient to deal with vertices of degree $4$ (or higher): Consider a vertex $v$ incident to four edges $e_1, e_2, e_3, e_4$ which occur in (cyclic) order $e_1, e_2, e_3, e_4$ in the rotation at $v$. Suppose that $e_1$ and $e_3$ cross oddly, and every other pair of (these four) edges crosses evenly. No matter what edge-flips are performed at $v$, there will always remain at least one pair of edges crossing oddly. This configuration is the unique obstacle to a vertex being even, up to edge-flips: suppose $v$ is incident to four edges. Using edge-flips we can make three consecutive edges at $v$ cross evenly with each other. Say the edges are $e_1, e_2, e_3, e_4$, in this order, and every two of $e_1, e_2$ and $e_3$ cross evenly. If $e_4$ crosses exactly $e_2$ oddly, we are done. If it crosses $e_3$ and $e_1$ oddly, we move the end of $e_4$ once around $v$, so that $e_4$ crosses exactly $e_2$ oddly. In all other cases, $e_4$ can be made to cross all of $e_1, e_2$ and $e_3$ evenly using edge-flips.

Hence, four edges at a vertex cannot always be made even by edge-flips. The next lemma shows that four edges are always the obstacle to making a vertex even by flips, independent of its degree.

\begin{lemma}\label{lem:4obs}
  If a vertex in a drawing cannot be made even by flips, then it is incident to four edges which cannot be made to cross each other evenly by flips.
\end{lemma}

This result can also be found in the full version of~\cite[Claim 8]{FK18_approx}.


\begin{proof}
 Let $v$ be a vertex in a drawing so that $v$ cannot be made even by edge-flips. Then $v$ has degree at least $4$. If it has degree exactly four, we are done by the argument preceding the lemma. So we can assume that the degree of $v$ is larger than $4$. Looking at all possible sequences of edge-flips, find a drawing which contains the longest block $B$ of consecutive ends so that all edges in the block cross each other evenly. Fix the rotation of that drawing. We can assume that $B$ contains at least four ends (otherwise, we can pick any four edges at $v$, since no four edges can be made to cross pairwise evenly). Let the block $B$ start with edge $e$ and end with edge $f$, and let the next edge after $f$ be $g$. Then $g \neq e$, since otherwise $v$ is even. If $g$ crosses
 both $e$ and $f$ oddly, we rotate $g$ once around $v$, flipping the crossing parity of $g$ with all other edges incident to $v$. In particular, we
 can assume that $g$ does not cross both $e$ and $f$ oddly. If $g$ crosses both $e$ and $f$ evenly, it has to cross some edge $h$ in $B-\{e,f\}$ oddly, since otherwise we could have added $g$ to the block $B$ of edges crossing pairwise evenly. In this case, $\{e,f,g,h\}$ are the
 four edges we were looking for. We can therefore assume that $g$ crosses exactly one of $e$ and $f$ evenly, without loss of generality let us say $g$ crosses $e$ evenly, and $f$ oddly (in the other case, we move $g$ in the rotation next to $e$).
 Then $B$ can be written as $e B_0 g' B_1 f$, where $g$ crosses all
 edges in $B_1 f$ oddly, and $g'$ evenly. Move $g$ past $f$ and $B_1$ so the new rotation is $e B_0 g' g B_1 f$. It is not possible that
 $g' = e$ or that $g$ crosses every edge in $B_0$ evenly, since in those cases $e B_0 g' g B_1 f$ would have been a longer block than $B$.
 Hence, there must be an edge $h \in B_0$ which crosses $g$ oddly. Then $e, h, g', g$ are the four edges we are looking for.
\end{proof}

The core of the inductive proof will be working with a specific configuration in \iocro-drawings: Two essential cycles $C_1$ and $C_2$ with $V(C_1) \cap V(C_2) = \{v\}$ and so that the edges of $C_1$ and $C_2$ incident to $v$ cannot be made to cross each other evenly by edge-flips at $v$. We call such a pair $(C_1,C_2)$ an {\em $X$-configuration}, see Figure~\ref{fig:pairOfCycles}.

\begin{figure}[htp]
\centering
\includegraphics[scale=1]{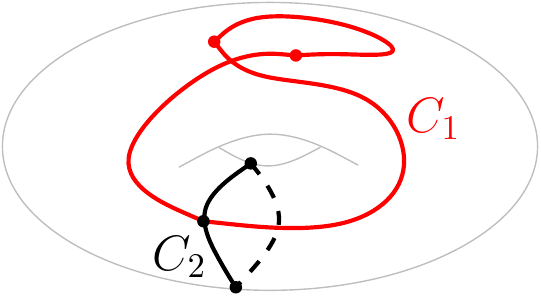}
 \caption{A pair of essential cycles $C_1$ and $C_2$ meeting in a single vertex $v$ so that
 the four edges of $C_1$ and $C_2$ incident to $v$ cannot be flipped so as to make them cross each other evenly.}
\label{fig:pairOfCycles}
\end{figure}

\begin{claim}\label{clm:Xmod}
 If $(C_1,C_2)$ is an $X$-configuration with $V(C_1) \cap V(C_2) = \{u\}$, and $C_3$ and $C_4$ are two cycles so that $V(C_3) \cap V(C_4) = \{u\}$ as well, and $u$ is incident to the same four edges in $C_3 \cup C_4$ as it is in $C_1 \cup C_2$, then $(C_3,C_4)$ is an $X$-configuration.
\end{claim}
\begin{proof}
 We only have to show that $C_3$ and $C_4$ are essential. Using edge-flips, we first make all edges of $C_3$ even, and then ensure that the two edges of $C_4$ incident to $u$ cross the edges of $C_3$ incident to $u$ evenly. If the edges of $C_3$ and $C_4$ do not alternate in the rotation at $u$, we can make the two edges of $C_4$ cross each other evenly by edge-flips, without affecting their crossing parity with the edges of $C_3$, contradicting the assumption that $(C_1,C_2)$ is an $X$-configuration.  So the edges of $C_3$ and $C_4$ do alternate in the rotation at $u$, implying that $C_3$ and $C_4$ cross each other oddly as curves, and therefore are both essential.
\end{proof}

\begin{lemma}\label{lem:Rotations}
If $G$ is a minimal $3$-connected graph that has an \iocro-drawing $D$ on the torus, but {\em does not} have an embedding on the torus, then $D$ contains an $X$-configuration $(C_1,C_2)$.
\end{lemma}

Figure~\ref{fig:pairOfCycles} serves as an illustration. We emphasize that Lemma~\ref{lem:Rotations} does not require the embedding to be compatible, or weakly compatible. We do not know whether an $X$-configuration is forced to exist in that case.

\begin{proof}
If every vertex of $G$ in an \iocro-drawing of $G$ on the torus can be made even by flips, we obtain an embedding of $G$. This follows from Theorem~\ref{thm:WHTS}, the weak version of the Hanani-Tutte theorem.
Therefore, there exists a vertex $v$ for which this is not the case.
By Lemma~\ref{lem:4obs} there are four edges $e_i = vu_i$, $1 \leq i \leq 4$ incident to $v$ which cannot all be made to cross each other evenly by edge-flips.

\begin{figure}[htp]
\centering
\includegraphics[scale=1]{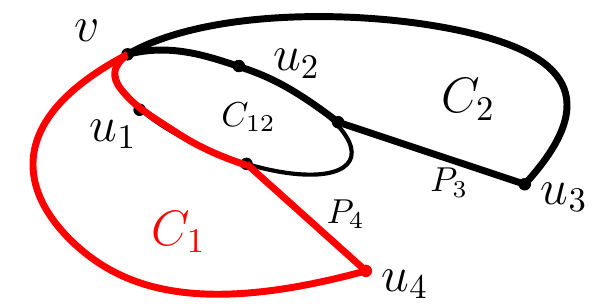}
 \caption{Finding an $X$-configuration $(C_1,C_2)$ in a minimal $3$-connected counterexample.}
\label{fig:Xconcounter}
\end{figure}

See Figure~\ref{fig:Xconcounter}. Since $G$ is $3$-connected, we can find edge-disjoint cycles $C_1,C_2$ intersecting in $v$ as follows:
Consider a minimal path in $G$ connecting two $u_i$ vertices that avoid $v$; without loss of generality this path is $P_{12}$ between $u_1$ and $u_2$
in $G- \{v,u_3,u_4\}$. Let $C_{12}$ denote the cycle obtained as the edge union of $E(P_{12}),\{e_1\}$ and $\{e_2\}$.
Let $P_3$ denote a path in $G- \{v,u_4\}$ between $u_3$ and $V(C_{12})- \{v\}$. Let $w$ denote the end vertex of $P_3$ on
$C_{12} - \{v\}$. Finally, let $P_4$ denote a path in $G- \{v,w\}$
between $u_4$ and $(V(P_{3})\cup V(C_{12})) - \{v,w\}$.
No matter where $P_4$ ends, the edge union of $E(C_{12}), E(P_3),E(P_4),\{e_3\}$ and $\{e_4\}$ contains a pair of cycles $C_1$ and $C_2$ meeting exactly in $v$.

As explained before Lemma~\ref{lem:4obs}, we can assume that the $e_i$ cross each other evenly, with the exception of one pair of edges which are not consecutive at $v$. If the ends of $C_1$ and $C_2$ do not alternate at $v$, then the two edges crossing oddly cannot both belong to the same $C_i$, since they would then be consecutive; therefore one belongs to $C_1$ and the other to $C_2$, which implies that $C_1$ and $C_2$ cross each other oddly (all edges of the cycles not incident to $v$ cross evenly, and the two cycles do not cross at $v$), hence both cycles are
essential. If the ends of $C_1$ and $C_2$ do alternate at $v$, then the two edges crossing oddly must belong to the same cycle (since otherwise they
would be consecutive), so again $C_1$ and $C_2$ cross each other oddly (at $v$, no other odd crossings), and both are essential.
\end{proof}

\subsection{The Persistence of the $X$-Configuration}\label{sec:PX}

In this section we show that Lemma~\ref{lem:reduction333} still works
when restricted to graphs with \iocro-drawings containing an $X$-configuration. With some restrictions, the same is true for Lemma~\ref{lem:nearly-3-connected}. We need these variants for our main induction.

\begin{lemma}\label{lem:Xreduction333}
If $G$ is a minimal graph that has an \iocro-drawing $D$ on the torus containing an $X$-configuration $(C, C_0)$, but {\em does not} have a [weakly] compatible embedding on the torus, then $D$ does not contain a cycle nearly-disjoint from $C$ (or $C_0$).
\end{lemma}

\begin{proof}
We show that the proof of Lemma~\ref{lem:reduction333} still applies. To that end, it will be sufficient to verify that the \iocro-drawing of the graph $G^*$ constructed in that proof contains an $X$-configuration. Let $(C, C_0)$ be the $X$-configuration in $G$, and $C'$ be a cycle in $G$ nearly disjoint from $C$. By definition, $C$ and $C'$ only touch in even vertices. It follows that $C$ cannot intersect the interior of $H$ and
is therefore present, unchanged, in $G^*$. The vertex $v \in V(C) \cap V(C_0)$ is odd, by the definition of $X$-configuration, so it cannot lie on $C'$. If $C_0$ contains no edges of $H \cup C'$, then $C_0$ exists unchanged in $G^*$, and $(C, C_0)$ is the $X$-configuration in $G^*$ we are looking for.

Hence, $C'$ must contain a first
and a last vertex of $C_0$ as we traverse it from $v$ to $v$. We modify $C_0$ to directly connect these two vertices via the newly added vertex in $G^*$. This gives us a new cycle $C'_0$ which intersects $C$ exactly in $v$. Therefore, $C'_0$ must still be essential, and we have found
the $X$-configuration $(C, C'_0)$ in $G^*$ we needed.
\end{proof}

Lemma~\ref{lem:nearly-3-connected} showed that a minimal counterexample (with compatible embedding)
has to be $2$-connected, and placed restrictions on a $2$-cut. If the presence of $X$-configurations
is required, the $2$-cut part needs further restrictions, so we state the $2$-cut part separately.

\begin{lemma}
\label{lem:X2-connected}
If $G$ is a minimal graph that has an \iocro-drawing on the torus containing an $X$-configuration, but {\em does not} have a [weakly] compatible embedding on the torus, then $G$ is $2$-connected.
\end{lemma}

\begin{proof}
Let $G$ contain an $X$-configuration $(C_1,C_2)$ with $C_1$ and $C_2$ intersecting in $v$. The proof of Lemma~\ref{lem:nearly-3-connected}
still shows that $G$ is connected.

Suppose that $G$ has a cut-vertex $x$. If $x \neq v$, then $C_1 \cup C_2$ is contained in $G-K$
for some component $K$ of $G-x$. If $L = G[V(K)\cup x]$ contains a cycle, we are done by Lemma~\ref{lem:Xreduction333},
since either $C_1$ or $C_2$ is vertex-disjoint from that cycle (not both of them can pass through $x \neq v$).

We conclude that $x = v$, so $x$ is odd. In that case, the original proof in Lemma~\ref{lem:nearly-3-connected}
works, as long as $C_1 \cup C_2$ lies in $G-K$ for some component $K$ of $G - v$. If that is not the case, then
$C_1-v$ and $C_2-v$ lie in different components of $G-v$. Let $K$ be the component of $G-v$ containing $C_2-v$
and let $L=G[V(K)\cup v]$. By Lemma~\ref{lem:G-C}, we can assume that $C_1$ is free of crossings.
By Lemma~\ref{lem:Xreduction333},
$K$ cannot contain a cycle (it would be vertex-disjoint from $C_1$), all essential cycles in $L$ pass
through $v$, and so, by Lemma~\ref{lem:G-x}, $L$ has a plane embedding, which is compatible, except for, possibly, at $v$ (in $L$, $v$ may be even; we can ignore this, however, since in the original drawing it is odd).
Let $G^*$ be the subgraph of $G$ obtained by removing all edges of $E(K)-E(C_2)$,  and all (resulting) isolated vertices.  Lemma~\ref{lem:HT1S} then applies to $G^*$ and $C_2$: Odd degree-2 vertices of $C_2$ can be easily made even, and hence, $v$ is the only odd vertex of $C_2$ in $G^*$. Therefore, we
know that $G^*$ has a [weakly] compatible embedding (except for at $v$) on the torus. We can then reinsert the compatible embedding of $L$ along $C_2$ to obtain a [weakly] compatible embedding of $G$: Every connected component of $K-C_2$ is connected by an edge with at most one other vertex of $C_2$  besides $v$.  This contradicts $G$ being a counterexample. (Note that we did not apply Lemma~\ref{lem:X2-connected} inductively in this case; we could not have, because removing $v$ destroyed the only $X$-configuration we knew about.)

In particular, $G$ is $2$-connected.
\end{proof}

\begin{lemma}
\label{lem:X2-cutodd}
If $G$ is a minimal graph that has an \iocro-drawing on the torus containing an $X$-configuration, but {\em does not} have a [weakly] compatible embedding on the torus, then $G$ has no 2-cut consisting of two odd vertices, unless both vertices belong to the $X$-configuration, with one of them belonging to both cycles.
\end{lemma}

The lemma does not cover the case in which the 2-cut consists of the vertex in which the two cycles of the $X$-configuration intersect together
with another odd vertex belonging to one of the cycles. We will deal with this situation later when investigating bridges.

\begin{proof}
By Lemma~\ref{lem:X2-connected}, we know that $G$ is two-connected.  Let $\{x,y\}$ be a cut-set
for which both $x$ and $y$ are odd. Let $(C,C_0)$ be the $X$-configuration, with $v$ being the vertex in which $C$ and $C_0$ intersect.
If $C \cup C_0$ contains at most one of $x$ and $y$, say $y$, and that vertex is not $v$,
then the proof in Lemma~\ref{lem:nearly-3-connected} can be used:
As $K'$  we choose the connected component of $G- \{x,y\}$ containing $C \cup C_0 - \{y\}$, so that $G[V(K')\cup \{y\}]$ contains $C \cup C_0$
which forces $G[V(K)\cup \{x\}]$ to not have an essential cycle. The proof can then be completed as described, since we induct on a graph $G^*=G-V(K)+xy$ which contains the $X$-configuration $(C,C_0)$.

Next, let us consider the case that both $x$ and $y$ belong to $C \cup C_0$, both of them different
from $v$. If $x$ and $y$ lie on the same cycle, $C$ or $C_0$, then we can choose $K$ and $K'$ so that
the other cycle belongs to $G[V(K')\cup \{y\}]$. The rest of the proof is as in Lemma~\ref{lem:nearly-3-connected},
but we observe that $G^*=G-V(K)+xy$ contains an $X$-configuration based on $(C, C_0)$ in which one of the cycles may
have been shortened by replacing the path between $x$ and $y$ on that cycle by $xy$.
If $x$ and $y$ lie on different cycles, $C\cup C_0 - \{x,y\}$ is connected, so belongs to the same component, and we let $K'$ be that component.
The rest of the proof is then as before.

We are left with the case that the $2$-cut has the form $\{v, y\}$ for some vertex $y \not\in V(C \cup C_0)$.
If $C-\{v\}$ and $C_0-\{v\}$ both belong to $G[V(K)]$ for the same component $K$ of $G- \{v,y\}$, we can argue as in Lemma~\ref{lem:nearly-3-connected}, since the $G^*$ constructed in that proof contains the $X$-configuration $(C,C_0)$. Otherwise, $C-\{v\}$ and $C_0-\{v\}$ belong to different components of $G - \{v,y\}$. In that case,
we can argue that if $G[V(K) \cup \{y\}]$ contains an essential cycle for any component $K$ of $G - \{v,y\}$, we are done by Lemma~\ref{lem:2disjoint}, since that cycle is vertex-disjoint from either $C$ or $C_0$. Therefore, any essential cycle in a $G[V(K) \cup \{v, y\}]$ must pass through $v$, and this remains true if we add an edge $vy$ drawn along a path, call it $vy$, connecting $v$ and $y$ in $G - K$. Hence,  $G[V(K) \cup \{v, y\}] + vy$ has a compatible embedding in the plane,by Lemma~\ref{lem:G-x}, for any component $K$ of $G- \{v,y\}$. We can then combine all of these plane embeddings to obtain a [weakly] compatible embedding of $G$ in the plane.
\end{proof}

\section{Hanani-Tutte on the Torus}\label{sec:HTT}

In this section we prove Theorem~\ref{thm:HTtorus}. We explain how to view \iocro-drawings of graphs on a torus as cylinder-drawings of an essential cycle with bridges in Section~\ref{sec:CV}, and then introduce some tools to simplify bridges in Section~\ref{sec:RB}. The main proof is in Section~\ref{sec:PHTT}, with the exception of a lemma which we prove separately in Section~\ref{sec:PLCB}, after presenting a small set of customized tools in Section~\ref{sec:TRL}.

\subsection{The Cylinder-View}\label{sec:CV}

In this section we introduce a different way of looking at an \iocro-drawing on the torus. We assume that the drawing contains an essential cycle $C$.

\paragraph{Cylinder.}
We can assume that $C$ is crossing-free (by Lemma~\ref{lem:G-C}).
If we cut the torus along $C$ we obtain a
drawing $D'$ of a graph $G'$ on a cylinder $\mathcal{C}$,
in which $C$ is replaced by two cycles of the same length, $C^L$ and $C^R$.
Indeed, $(\mathbb{S}^1\times\mathbb{S}^1)- C \cong I \times \mathbb{S}^1$,
where $I$ is an open interval and $\mathbb{S}^1$ is the $1$-sphere.
In the current section $G'$ always denotes such a cylindrically drawn graph.
 The boundary of $\mathcal{C}$ consists of the disjoint union of $C^L$ and $C^R$.
$D'$ can have independent odd pairs of edges touching the boundary.
We fix a drawing $D'$ of $G'$ on the cylinder as described above.


\paragraph{Bridge.}
Given a subgraph $F$ of $G$ (possibly with no edges), an
{\em $F$-bridge} is a subgraph of $G$
that consists of either $(i)$ an edge of $E(G)- E(F)$
with both endpoints in $V(F)$, or $(ii)$ a component of $G- V(F)$
together with all edges between the component and $V(F)$ and their endpoints in $V(F)$.
We call $F$-bridges of type $(ii)$ {\em non-trivial}.
A subgraph $F$ and its $F$-bridges partition $E(G)$.
The vertices of $F$ in an $F$-bridge are its
\emph{feet}, the edges of the component incident to its feet are its \emph{legs},
and, for type $(ii)$, its component of $G- V(F)$ is its \emph{core}.
In $G'$, $C$-bridges correspond to $(C^L\cup C^R)$-bridges, identical except for their feet.

\paragraph{Upper Indices $L$ and $R$.}
Let $v\in V(G)$ be a vertex on $C$. We denote by $v^L$ and $v^R$ its
copy on $C^L$ and $C^R$, respectively, in $G'$.
Let $e\in E(G)$ be an edge intersecting $C$ (graph-theoretically, $C$ is free of crossings).
We denote by $e^L,e^{R}$ and $e^{LR}$,
respectively, its corresponding edge in $G'$ intersecting
only $C^L, C^R$ and both $C^L$ and $C^R$, respectively, in $G'$.
Only edges of $C$ can have more than one corresponding
edge in $G'$. Then $e^{LR}$ must be an edge joining $C^L$ with $C^R$,
and $e^L$ an edge with at least one end vertex on $C^L$ and none on $C^R$.
Similarly, for $e^R$. (Note that $e^L$, $e^R$, and $e^{LR}$ may be undefined.)
If $P$ is a path in $C$ we denote by $P^L$ and $P^R$ its copy on $C^L$ and $C^R$, respectively, in $G'$.


Let $H\subseteq G$ be a subgraph of a $C$-bridge.
We denote by either $H^L,H^{R}$ or $H^{LR}$ the subgraph of  $G'$ corresponding to $H$, depending on whether the subgraph corresponding to $H$ in $G'$  intersects (in its feet) either only  $C^L$, only $C^R$, or both $C^L$ and $C^R$, respectively.

For other subgraphs
of $G'$ the upper indices cannot always be used, for example, a subgraph containing parts of $C^L$ and $C^R$.
In these remaining cases, we write the subgraph of $G'$ without using the upper index.

\paragraph{Supporting Independent Odd Pairs}

If the two edges of an odd pair in $G$ attach to the same vertex $u \in C$ on opposite sides of (the crossing-free) $C$, then
the two edges turn into an independent odd pair of edges
in $G'$, attaching to $u^L$ and $u^R$. We say that $u$ (or $u^L$ and $u^R$) {\em support(s) an independent odd pair (relative to $C$)}.
If both edges are legs of the same $C$-bridge $H$, we say that {\em $H$ supports an independent odd pair at $u$} (or $u^L$ and $u^R$).

\paragraph{Cycles and Paths on a Cylinder.}
The complement of a closed curve drawn on a cylinder is partitioned into \emph{interior} and \emph{exterior} according to the two-coloring of the connected components.
Whenever we talk about an interior or exterior of a curve (or a cycle)
in the drawing of $G'$ we specify which components are understood
to be the interior and exterior.

An \emph{$L$-diagonal} is a path $P^L$  connecting a pair of distinct vertices on $C^L$
internally disjoint from both $C^L$ and $C^R$.
Similarly we define an $R$-diagonal.
Let $P^L$ be an $L$-diagonal.
The \emph{$L$-foundation of $P^L$} (similarly we define a $R$-foundation) is a path $P'^L$
contained in $C^L$ connecting $u^L$ with $v^L$ such that the counterpart in $G$ of the cycle $C'=P^L\cup P'^L$ is non-essential.
By viewing $C^R$ from the perspective
of $C^L$ as being ``at infinity'',
for non-essential cycles $C'$ obtained in this way the interior is defined
as the union of regions in the two-coloring of the complement of $C'$
having the same color as the component bounded by $P'^L$, and similarly if the roles of $C^L$ and $C^R$ are exchanged.
If we concatenate $P^L$ with the path that is complementary to $P'^L$ on $C^L$, we get an essential cycle
 by Lemma~\ref{lem:3PC}.

An \emph{$LR$-diagonal} is a path $P^{LR}$  in $G'$ joining a vertex  on $C^R$ with a vertex on  $C^L$  internally disjoint from $C^R \cup C^L$.

\paragraph{Three-Stars.}

A {\em three-star} is a $K_{1,3}$ which occurs as a bridge of three vertices $\{u,v,z\}$. Our goal will
be to show that a counterexample to Hanani-Tutte can be reduced to a set of pairwise edge-disjoint three-stars
that are  $\{u,v,z\}$-bridges, where none of $u,v,z$ are even. (For this we need some additional tools developed in the following sections.)
The underlying graph then must be a $K_{3,t}$, which is not toroidal for $t \geq 7$.
We can then complete the proof by using the results on Kuratowski minors from Section~\ref{sec:KM}.

\subsection{Reducing Bridges}\label{sec:RB}

The easiest case occurs when all independent odd pairs in a cylindrical drawing are incident on the same vertex of the $X$-configuration. Lemma~\ref{lem:HT1S} deals with this case, but we restate it here from the cylindrical point of view. This is one of the few results which does not require minimality.

\begin{lemma}\label{lem:Xuiop}
    Suppose $G$ contains an essential cycle $C$. If there is at most one vertex on $C$ which supports independent odd pairs of edges in $G'$, then $G$ has a compatible embedding in the torus.
\end{lemma}

\begin{proof}
If it exists, let $v$ be the vertex on $C$ supporting all independent odd pairs. Suppose $u$ is an arbitrary vertex on $C$, different from $v$ if $v$ exists. We can then apply (at most) two vertex-split operations to $u$, one for each (non-empty) side of $C$ with edges incident to $u$ on that side attached, turning $u$ into an even vertex of degree at most $4$; the drawing remains \iocro. After performing these splits for all vertices on $C$, other than $v$ if it exists, we can apply Lemma~\ref{lem:HT1S} to obtain a compatible embedding of the modified $G$. Contracting the edges which resulted from the splits, we obtain a compatible embedding of $G$.
\end{proof}

The following lemma will be used to deal with a situation not covered by Lemma~\ref{lem:X2-cutodd}, where one of the feet is the shared vertex of $C$ and $C_0$.
Recall that a $C$-bridge is {\em non-trivial} if it not an edge, so it must contain a vertex in its core.


\begin{lemma}\label{lem:X2ftCBridge}
    Suppose that $G$ is a minimal graph that has an \iocro-drawing on the torus containing an $X$-configuration $(C,C_0)$, but {\em does not} have a weakly compatible embedding on the torus. The following are true:
    \begin{itemize}
        \itemi If there is a non-trivial $C$-bridge $H$ which has exactly two feet on $C$, both of which are odd, then $H$ contains $C_0$. (In particular, there exists at most one such bridge $H$.)
        \itemii If the $C$-bridge containing $C_0$ has exactly two feet on $C$, it is not possible that it supports independent odd crossings at both its feet.
    \end{itemize}
\end{lemma}

\begin{proof}
   Let $u$ be the intersection of $C$ and $C_0$. We start by proving $(i)$.

   \begin{figure}[htp]
    \centering
    \includegraphics[scale=0.7]{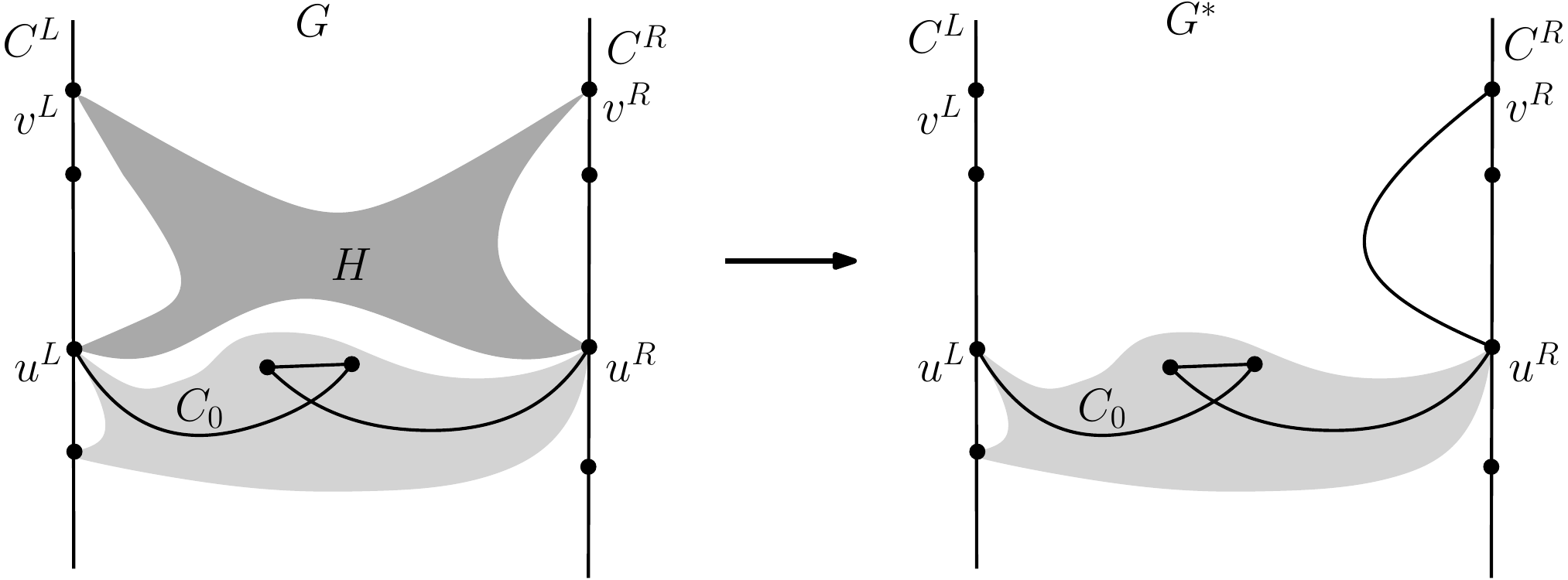}
    \caption{{\em Left}: A $C$-bridge $H$ with exactly two feet $u$ and $v$, and not containing $C_0$. {\em Right:} The reduced graph $G^*$ in the proof of Lemma~\ref{lem:X2ftCBridge}(i), in which the $C$-bridge $H$ is replaced with the edge $uv$ drawn along a path in $H$.}
    \label{fig:lemma62}
\end{figure}
   Let $H$ be a non-trivial $C$-bridge with exactly two feet on $C$ both of which are odd.  If $u$ is not one of the feet of $H$ we can apply Lemma~\ref{lem:Xreduction333}, since in this case $H \cup (C - \{u\}$) contains a cycle vertex-disjoint from $C_0$. So $u$ must be one of the feet of $H$. Let the other foot of $H$ be $v$. Suppose $H$ does not contain $C_0$. This situation is illustrated on the left side of Figure~\ref{fig:lemma62}; note that $H$ may not attach to both $v^L$ and $v^R$.

   Then $H-u$ cannot contain an essential cycle, since the cycle would be disjoint from $C_0$, contradicting Lemma~\ref{lem:2disjoint}.
   Therefore all essential cycles (if any) in $H$ pass through $u$. The same is true for $H + uv$, where we draw the edge $uv$ along a $uv$-path in $C$.
   By Lemma~\ref{lem:G-x}, $H+uv$ has a plane embedding, in which the embedding of $H-\{u\}$ is compatible with the original drawing. Let
   $G^* = G - (V(H)- \{u,v\}) + uv$. To obtain a drawing of $G^*$ we follow $G$ except that the edge $uv$ is drawn along a path in $H$. See the right side of Figure~\ref{fig:lemma62} for an illustration. We have that $G^* \prec G$, since $H$ is non-trivial,
   and $G^*$ contains $(C,C_0)$. By minimality, $G^*$ has a weakly compatible embedding. We can then embed $H+uv$ alongside $uv$ in the embedding of $G^*$; since both $u$ and $v$ are odd, changing their rotation is acceptable. This gives us a weakly compatible embedding of $G$, a contradiction. We conclude that $H$ must contain $C_0$; in other words it is the unique non-trivial $C$-bridge with two odd feet.

   To establish $(ii)$, we assume for the sake of contradiction, that $H$ supports independent odd crossings at both its feet $u$ and $v$. Let $P_v^{LR}$ be the path in $H$ from $v^L$ to $v^R$
   that starts and ends with the pair of edges involved in the independent odd crossing at $v$. Then $P_v^{LR}$ corresponds to an essential cycle $C_1$ in $G$. See the left side of Figure~\ref{fig:lemma62-2} for an illustration.
   We have two $X$-configurations then: $(C,C_0)$ and $(C,C_1)$.
    By Lemma~\ref{lem:Xreduction333},\\

    (*) there exists no cycle in $G$ nearly-disjoint from $C_1$. \\

   We claim that there is no non-trivial $C$-bridge besides $H$, and there can only be one trivial $C$-bridge, namely $uv$. Consider a $C$-bridge $H'$ other than $H$, trivial or not.
   Then $H'$ must have at least two feet on $C$ if $H'$ is trivial, and the same is true if $H'$ is non-trivial, by Lemma~\ref{lem:X2-connected}. The bridge $H'$ cannot
   have a foot different from both $u$ and $v$: if it did, there would be a cycle in $H' \cup C - \{u\}$ vertex-disjoint from $C_0$, or  in $H' \cup C - \{v\}$ vertex-disjoint from $C_1$, contradicting Lemma~\ref{lem:Xreduction333}. Hence, $H'$ must have exactly two feet, which are $u$ and $v$, on $C$. Since $H'$ does not contain $C_0$, part $(i)$, which we already proved, implies that $H'$ is trivial and consists of the single edge $uv$. It follows that $H$ is the only {\em non-trivial} $C$-bridge. and the only possible trivial $C$-bridge is $uv$.

   \begin{figure}[htp]
    \centering
    \includegraphics[scale=0.7]{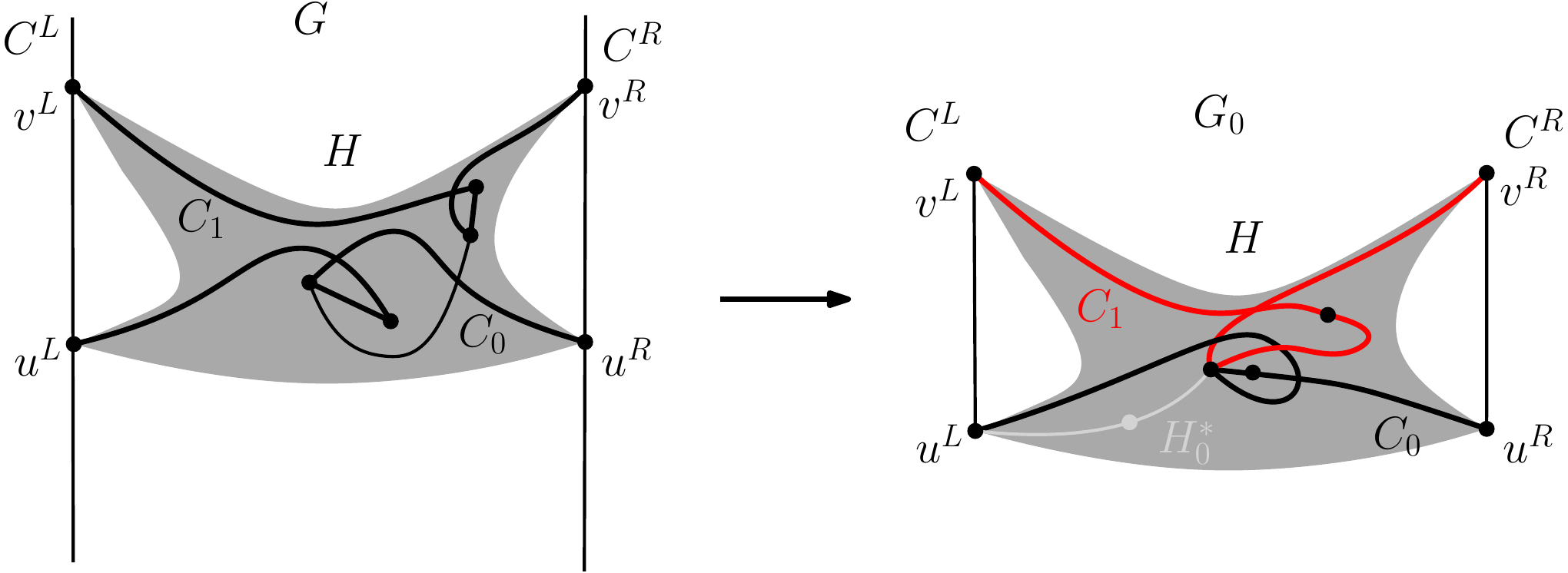}
    \caption{The cylinder view of the reduced graph $G_0$ in the proof of Lemma~\ref{lem:X2ftCBridge}(ii), in which the path $uwv$ (on the left) is removed  from $G$ (on the right). Even though the illustrations suggest so, $C_0$ and $C_1$ need not be disjoint.}
    \label{fig:lemma62-2}
   \end{figure}

   In particular, $u$ and $v$ are the only feet on $C$. Let $G_0$ be obtained from $G$ by removing $V(C)-\{u,v\}$ and adding the edge $uv$ (if $uv$ does not belong to $G$ already) drawn along $C$.  Since $u$ and $v$ are odd, a weakly compatible embedding of $G_0$ could be extended to a weakly compatible embedding of $G$ by embedding the one or two missing $u,v$-paths of $C$ close to $uv$. By minimality of $G$, we have \\

    (**) $G_0$ does not admit a weakly compatible embedding on the torus, and hence, does not contain an $X$-configuration by the minimality of $G$.  \\

    We claim that every $C_0$-bridge in $G_0$ except the one containing $v$, which we denote by $H_0$, is a path of length at most $2$ ending in $u$. 
    To see this, let $H^*_0$ be a $C_0$-bridge other than $H_0$; see the right side of Figure~\ref{fig:lemma62-2}. Since $H^*_0$ cannot contain a cycle nearly-disjoint
    from $C$, it has at most two feet on $C_0$, one of which is $u$, since $C_0 \cap C = \{u\}$. By the same token,
    every cycle in $H^*_0$ must pass through $u$. By Lemma~\ref{lem:X2-connected}, $u$ is not a cut-vertex of $G$, so $H^*_0$ must have a second foot on $C_0$, which by $(*)$ must be a vertex on $C_1$.  The core of $H^*_0$ is acyclic and it has only one edge to each of its feet, or else $H^*_0$ would have a cycle disjoint from $C$ or $C_1$.  Then $H^*_0$ must be a subdivided edge, since otherwise it would have a cut-vertex of $G$, contradicting Lemma~\ref{lem:X2-connected}. By the minimality of $G$, $H^*_0$ is a path of length at most $2$.

    In summary, $C_0$ is a cycle in $G_0$ that has a non-trivial $C_0$-bridge containing $v$, and all other $C_0$-bridges are paths of length at most $2$ with one end at the odd vertex $u$.

    We can therefore choose $C^+_0$ as the shortest cycle in $G_0$ that has a non-trivial $C^+_0$-bridge $H^+_0$ with an odd foot $u^+$ so that
    all other $C^+_0$-bridges are (subdivided) edges of length at most $2$ ending in $u^+$, and so that the number of edges in $H^+_0$ is maximized.


    We show by induction on the number of $C^+_0$-bridges other than $H^+_0$ that $G_0$  admits a weakly compatible embedding on the torus, which contradicts  (**), and thus concludes the proof.

    In the base case there is only one $C^+_0$-bridge, $H^+_0$. By~(**)  there is no $X$-configuration, so none of the vertices on $C^+_0$ support an independent odd pair (since that independent odd pair would have to belong to the unique $C^+_0$-bridge, which would force an $X$-configuration with $C^+_0$). Lemma~\ref{lem:Xuiop} then gives us a compatible embedding of $G_0$ on the torus.

    In the inductive case, $G_0$ contains a $C^+_0$-bridge $P^+_0$ different from $H^+_0$ that is a (subdivided) edge
    incident to $u^+$, which is odd, and a second vertex $v^+$ on $C^+_0$. Suppose there is such a $C^+_0$-bridge $P^+_0$ for which $v^+$ is a neighbor of $u^+$ on $C^+_0$, so the edge $u^+v^+$ is part of $C$. Choose $P^+_0$ so
    its end at $v^+$ is next to $u^+v^+$ in the rotation at $v^+$. We can then remove $P^+_0$, apply induction, and
    insert $P^+_0$ by following $u^+v^+$ closely. (This reestablishes the rotation at $v^+$, which matters if $v^+$ is even.)

    We can therefore assume that there are no $C^+_0$-bridges between $u^+$ and its neighbors on $C^+_0$ that are subdivided edges. Next, suppose there is a $C^+_0$-bridge $P^+_0$ which is a subdivided edges between $u^+$ and an odd vertex $v^+$, so we can modify the rotation at both $u^+$ and $v^+$. We already dealt with the case that $P^+_0$ joins two consecutive vertices along $C^+_0$, so we can assume that $u^+$ and $v^+$ are not consecutive along $C^+_0$. Then, by the choice of $C^+_0$, a part of $C^+_0$ between the end-vertices of $P^+_0$ must have the same length as $P^+_0$, which is $2$. Call the middle vertex of that part $w^+$. We know that $w^+$ is not incident on any subdivided-edge bridges (since $w^+$ is consecutive to $u^+$, and we already dealt with this case). Also, $w^+$ is not a foot of $H^+_0$, since otherwise, by maximality of $H^+_0$. Therefore, $w^+$ has degree $2$, and we can again remove $P^+_0$, apply induction, and insert $P^+_0$ back in the embedding along the edges of $C^+_0$.

    It remains to deal with the case that all the $C^+_0$-bridges different from $H^+_0$ are paths whose second foot,
    the foot different from $u^+$, is an even vertex.
    Suppose there were a foot $x \neq u^+$ on $C^+_0$ that supports an independent odd pair. Then $x$ can only be a foot of $H^+_0$, since all other bridges have $u^+$ or an even vertex as a foot. It follows that the independent odd pair at $x$ must belong to $H^+_0$, which forces an $X$-configuration in $C^+_0 \cup H^+_0$, contradicting (**). We
    can therefore assume that any two edges attaching at $x \in C^+_0 - \{u^+\}$ on opposite sides of $C^+_0$ cross evenly. Lemma~\ref{lem:Xuiop} then gives us a compatible embedding of $G_0$.
\end{proof}

The following lemma is the heart of the reduction; we will state it and prove it in Section~\ref{sec:PLCB}. Before that, we will see how to use it to complete the proof of the main result.

\begin{lemma}\label{lem:XCbridges}
Suppose $G$ is a minimal graph that has an \iocro-drawing on the torus containing an $X$-configuration $(C,C_0)$,
so that there is a $C$-bridge $H$ which attaches to $C$ in at least two vertices $u$ and $v$ that support independent odd pairs of legs in $H$, but $G$ {\em does not} have a weakly compatible embedding on the torus. Then
 \begin{itemize}
  \itemi if there is more than one $C$-bridge, then there is an $X$-configuration $(C^{+},C^{+}_0)$ in $G$
  for which there is only one $C^{+}$-bridge, and
  \itemii if there is only one $C$-bridge, then $G$ is a subgraph of a subdivision of a $K_5$ or a $K_{3,t}$ with bracers.
 \end{itemize}
\end{lemma}

\subsection{Proof of Theorem~\ref{thm:HTtorus}}\label{sec:PHTT}
\label{sec:mainproof}

Assume, for a contradiction, that Theorem~\ref{thm:HTtorus} fails. Then there is a graph $G$ with an \iocro-drawing $D$ on the torus, such that $G$ cannot be embedded on the torus. Let $G$ be a minimal such counterexample, with \iocro-drawing $D$.
By Corollary~\ref{cor:3-connected} we know that $G$ is 3-connected, and by Lemma~\ref{lem:Rotations} that
$D$ contains an $X$-configuration $(C, C_0)$. So there is a counterexample to the following statement, which strengthens Theorem~\ref{thm:HTtorus}:

\begin{quote}
   If $G$ has an \iocro-drawing $D$ containing an $X$-configuration $(C,C_0)$, then $G$ has a weakly compatible embedding on the torus.
\end{quote}

When the statement is true, we say that $G$---or more precisely, $(G,D,C,C_0)$---satisfies {\em HT-XWC}, an acronym for Hanani-Tutte--X-Weakly-Compatibility. So we know that
there is a $(G,D,C,C_0)$ which does not satisfy HT-XWC, and therefore a minimal such counterexample.

Fix a $(G,D,C,C_0)$ violating HT-XWC for which $G$ is minimal
with respect to first, $\prec$ (defined at the beginning of Section~\ref{sec:MC}), second, the number of $C$-bridges, and third, $|E(C)| + |E(C_0)|$. We write $\prec'$ for the strict partial ordering defined in this way; by definition $\prec'$ refines $\prec$, so a $\prec'$-minimal counterexample is also a $\prec$-minimal counterexample.

We first consider the case that there exists a $C$-bridge that supports independent odd pairs at two (or more) vertices of $C$. Then, by Lemma~\ref{lem:XCbridges}$(i)$, there is an $X$-configuration  $(C^{+},C^{+}_0)$ so that there is only one $C^{+}$-bridge. By Lemma~\ref{lem:Xuiop} there cannot be a single vertex on  $C^{+}$ supporting all independent odd pairs in $G'$, with respect to $X$-configuration  $(C^{+},C^{+}_0)$, so there must be at least two such vertices. Since there is only one $C^{+}$-bridge, the independent odd pairs are supported by that single $C^{+}$-bridge. Hence, Lemma~\ref{lem:XCbridges} applies, and, by part $(ii)$, $G$ is a subgraph of a subdivision of $K_5$ or a $K_{3,t}$ with bracers. We get a contradiction by Lemma~\ref{lem:K5HTWC} in the first case and by Lemma~\ref{lem:K3nspokesHTWC} in the second case .


We conclude that for every $C$-bridge $H$ there is at most one vertex on $C$ so that $H$ supports an independent odd pair at that vertex. Let $u$ be the vertex in the intersection of $C$ and $C_0$.  By definition, the two $C_0$-edges incident to $u$ form an independent odd pair at $u$. Hence the path $P^{LR}$ obtained by following $C_0$ from $u^L$ to $u^R$ in $G'$ contains an independent odd pair at $u$. Let $H^{LR}$
be the $C$-bridge containing $P^{LR}$.

There must be a $v \in V(C)-\{u\}$ so that $v^L$ and $v^R$  are incident to edges forming an independent odd pair $l_1$ and $l_2$ (in $G'$),
since otherwise we are done by Lemma~\ref{lem:Xuiop}. Since $H^{LR}$ only supports independent odd pairs at $u$, it is not possible
that both $l_1$ and $l_2$ belong to $H^{LR}$, so at least one of them must belong to a different $C$-bridge $H'$.

First, suppose that $H'$ is not just a single edge.
Since $v$ and $u$ are odd, by Lemma~\ref{lem:X2ftCBridge}
$H'$ must have a foot $w'$ on $C^L$ or $C^R$ different from $u$ and $v$.
Then $H'$  contains a path, avoiding $u$, from $w'$ to $v$ (which lies on $C^L$ or $C^R$). This path is disjoint from $P^{LR}$. Adding to that a subpath of $C$ connecting $v$ to $w'$ while avoiding $u$ gives us a cycle $C'$ in $G$ which is vertex-disjoint from $P^{LR}$, and, therefore, $C_0$. This contradicts the choice of $G$ by  Lemma~\ref{lem:Xreduction333} (note that $C'$ does not have to be essential for the lemma to apply).

Therefore, $H'$ consists of a single edge $e$. We distinguish two cases:
$e$ is incident to one side of $C$ only, or $e$ is an $LR$-diagonal.

Suppose first that $e$ is incident to $C^L$ only. This contradicts the choice of $C$: Let $C'$ be the result of shortening $C$ by using $e$ to replace the $L$-foundation of $e$ on $C^L$. Then $C'$ is essential, and $C'$ cannot have more bridges than $C$, since $V(C') \subseteq V(C)$: we gain at most one new bridge (containing
the $L$-foundation of $e$ on $C^L$), but we lose the bridge $e$, so the total number of bridges does not increase. The
path $P^{LR}$ remains an $LR$-diagonal for $C'$, since its ends remain
on the opposite sides of $C'$. (This is not necessarily true if $e$ is an $LR$-diagonal.) The two new cycles still intersect in $u$
which can still not be made even by flips. Since $C'$ is shorter than $C$, and $C_0$ did not change, this contradicts the minimality of the $X$-configuration $(C, C_0)$.


 \begin{figure}[htp]
\centering
\includegraphics[scale=0.7]{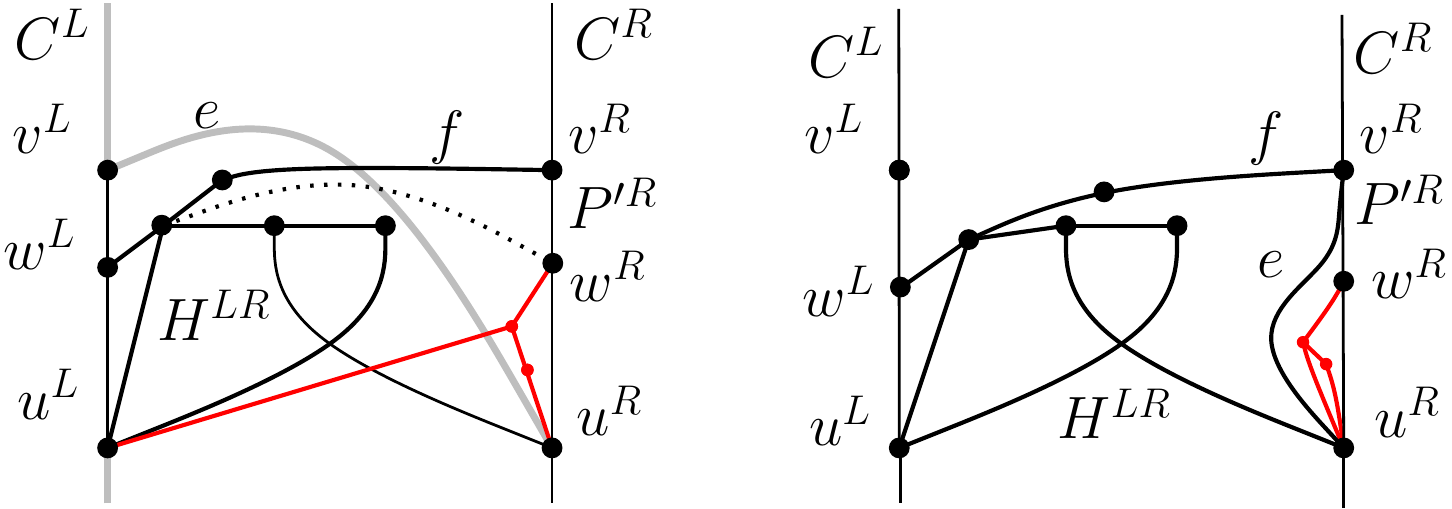}
 \caption{A $C$-bridge $H'$ is a single edge $e$ (left) that can be embedded along the path $P'^R$ after all the other $C$-bridges $H^\#$ different from $H$ are embedded along $C^R$ (right).}
\label{fig:orthogonalPair}
\end{figure}

In the remaining case, $e$ is an $LR$-diagonal.
We want to redraw $C$-bridges other than $H^{LR}$
so that $e$ can be drawn along $C^R$ or $C^L$, see Figure~\ref{fig:orthogonalPair}.

By symmetry, we can assume that $e$ is incident to $v^L$. The other endpoint of $e$ must be $u^R$. Otherwise we have an essential cycle disjoint from $C_0$, which passes through $u$: concatenate $e$ with the path on $C^R$ connecting the endpoint of $e$ on $C^R$ to $v^R$ while avoiding $u$. This contradicts Lemma~\ref{lem:2disjoint}.

Let $P'$ be a path between $u$ and $v$ on $C$. This path corresponds to two paths $P'^L$ and $P'^R$  on $C^R$ and  $C^L$, respectively, in $G'$.
We claim that it is not possible that $H^{LR}$ has feet in the interior of both both of these paths. If it did, we could use
a path (for example, a path passing along the dashed curve in the left illustration of Figure~\ref{fig:orthogonalPair}) in $H^{LR}$ connecting the two feet (and otherwise avoiding $C$) combined with a subpath of $P'^L$ (or $P'^R$) to
obtain an essential cycle which is vertex-disjoint from the essential cycle formed by $e$ and the path connecting $u$ and $v$ on $C$
which is not $P'$. This contradicts Lemma~\ref{lem:2disjoint}.

Without loss of generality then, we can assume that $P'^R$ contains no feet of $H^{LR}$.
Consider a $C$-bridge $H^{\#}$ other than $H$ and $H' = e$. Then $H^{\#}$ must have at least two feet in $C$ by Lemma~\ref{lem:X2-connected}.
On the other hand, $H^{\#}$ cannot have two feet on $C^L$ (or $C^R$) which are different from $u$, since this would contradict
Lemma~\ref{lem:Xreduction333} by giving us two vertex-disjoint cycles, the essential cycle $C_0$ and a cycle in $H^{\#} \cup C$ avoiding $u$.
Similarly, if $H^{\#}$ has a foot each on $C^L$ and $C^R$ both different from $u^L$ and $u^R$,
then $H^{\#} \cup C$ contains an essential cycle in $G$ avoiding $u$, and thereby vertex-disjoint from $C_0$, contradicting Lemma~\ref{lem:2disjoint}. Therefore, each such bridge $H^{\#}$ has exactly one foot on $C^L$ or $C^R$ other than $u^L$ or $u^R$ (and at least one of these).

We now prove that $e$ can be redrawn so that it is not an $LR$-diagonal.
Let ${\cal H}$ be the set of $C$-bridges (other than $H$ and $H'$) which have a foot in $P'^R$ different from $u^L$ and $u^R$ (and at least one of these).

We will show, by induction on the size of ${\cal H}$, that there is a weakly compatible redrawing
where all $C$-bridges in ${\cal H}$ as well as $e$ are free of crossings and $e$ is not an $LR$-diagonal. This completes the proof, since then we are back in an earlier case, where $e$ can be used to shortcut $C$, implying that $(G, D, C, C_0)$ was not a minimal counterexample.

In the base case, ${\cal H}$ is empty, and we can redraw $H' = e$ close to $P'^R$, crossing-free and no longer an $LR$-diagonal.
The rotation at $u$ and $v$ changes, but both
are odd vertices, so that is fine. Moreover, $(C,C_0)$ is still an $X$-configuration (since the rotation of the edges in $H^{LR}$ did not change with respect to $C$, and those edges could not be made even by flips).

In the inductive case we pick a bridge in ${\cal H}$ with a foot $w^R$ as close to $u^R$ (on $P'^R$) as possible. If $w$ is odd, we pick
any bridge $H^{\#}$ incident to $w^R$. If $w$ is even, we pick the bridge $H^{\#}$ first in the rotation at $w^R$ after $w^Ru^R$ (anti-clockwise). By Lemma~\ref{lem:Xreduction333} there is a single edge between $w^R$ and the core of $H^{\#}$ (if there was more than one edge, their endpoints are connected in the core of $H^{\#}$ resulting
in a cycle nearly disjoint from $C$).
Similarly, there cannot be a cycle in $H^{\#}$ avoiding $u$, since it would be vertex-disjoint from $C_0$, contradicting Lemma~\ref{lem:Xreduction333} due to the $X$-configuration $(C,C_0)$.
Therefore all cycles in $H^{\#}$ pass through $u$, and $H^{\#} \cup uw$ has a compatible plane embedding (except for, possibly, at $u$),
by Lemma~\ref{lem:G-x}.

By induction, we can assume that $e$ and all the bridges in ${\cal H}$ other than $H^{\#}$ have been (weakly compatibly) redrawn without crossings, and $e$ is no longer an $LR$-diagonal. We can then insert the embedding of $H^{\#}$ close to $w^Ru^R$, see the right illustration in Figure~\ref{fig:orthogonalPair}; in the case that $w$ is even, this reestablishes the rotation
at $w$, since we picked the bridge with the closest leg to $w^Ru^R$ in the rotation at $w$, and the bridge connects to $w^R$ via a single edge. This gives us a weakly compatible drawing in which $e$ as well as all bridges in ${\cal H}$ are free of crossings, and $e$ is no longer an $LR$-diagonal.

Note that in the induction argument on the size of ${\cal H}$ we do not use the minimality of $G'$, we only incrementally redraw the bridges in ${\cal H}$ so that $e$ can be redrawn without crossings attaching to one side of $C$ only.


\subsection{Three Reduction Lemmas}\label{sec:TRL}

In preparation for the proof of Lemma~\ref{lem:XCbridges} in the next section, we present three reduction lemmas.
For each lemma, we assume that $G$ is a minimal graph that has an \iocro-drawing on the torus containing an $X$-configuration $(C,C_0)$, but {\em does not} have a weakly compatible embedding on the torus,

\begin{lemma}
\label{lem:X3deg}
Let $v\in V(G)$ be a vertex of degree $3$, and $u$ and $w$ two of its neighbors.
If both $u$ and $w$ are odd, then $uw\not\in E(G)$ unless $G- uw$ does not contain an $X$-configuration.
\end{lemma}


\begin{proof}
Otherwise we would  violate the minimality of $G$.
Indeed, removing $uw$ results in a graph that is not a counterexample, and $uw$ can be inserted into a toroidal embedding of $G- uw$ without introducing an edge crossings and while maintaining the rotations at even vertices in the given independently even drawing.
\end{proof}

\begin{lemma}
\label{lem:X2-cuttwo}
Suppose that $G$ contains a 2-cut $\{w',x\}$ such that $w'$ belongs to $C$, but not to $C_0$, and $x$ does not belong to $C$.
If the edges incident to $w'^L$ (or $w'^R$) in $G' - E(C)$ connect $w'$ with at least two connected components of $G[V- \{w',x\}]$ or at least one such component if $w'^Lx$ (or $w'^Rx$) is an edge of $G'$, then the vertex $x$ is even.
\end{lemma}
\begin{proof}
If $w'$ is even, then the union of two connected components of $G[V- \{w',x\}]$ or the union of one such components together with the edge $w'^Lx$ or $w'^Rx$ contains a cycle $C'$ nearly disjoint from $C$, contradicting Lemma~\ref{lem:Xreduction333}. We conclude that $w'$ must be odd. Since $w'$ does not belong to $C_0$ and $x$ does not belong to $C$, Lemma~\ref{lem:X2-cutodd} applies, and $x$ must be even.
\end{proof}

\begin{lemma}
\label{lem:X2-cut2}
Suppose that $G$ has a 2-cut $\{u,v\}$ such that $v\not\in V(C)$ is even and $u\in V(C)$. Let $H=G[V(G^*)\cup\{u,v\}]$, where $G^*$ is the  union  of all components of $G- \{u,v\}$ excluding the component containing $C- \{u\}$. If $H$ does not contain an essential cycle, then $H$ is a path. In particular, $H$ has only a single edge adjacent to $v$.
\end{lemma}
\begin{proof}
If $H$ is an edge or path, we are done. Hence $H$ must contain a cycle $C'$: all internal vertices of $H$, that is, vertices other than $\{u,v\}$, have degree at least two, otherwise there would be a cut-vertex. In particular, $H$ contains at least three edges. We claim that $u$ is odd. Suppose $u$ were even. The cycle $C'$ in $H$ cannot be essential, by assumption, so $C$ and $C'$ cross an even number of times. Since $u$ is even this means that $C$ and $C'$ are either vertex-disjoint, or touch in $u$, contradicting Lemma~\ref{lem:Xreduction333}. This shows that $u$ is an odd vertex.

The edges of $H$ at $v$ have to be consecutive in the rotation at $v$:
If not, there has to be a pair of paths $P_1$ and $P_2$ between $v$ and $C- u$  internally disjoint from $H$, whose ends at $v$ alternate with edges $e$, $f$ of $H$. Since $H$ and $C$ only share $u$, the cycle in $P_1 \cup P_2 \cup (C - u)$ and the cycle
in $H$ containing $e$ and $f$ only have the even vertex $v$ in common, and, since their ends alternate at $v$ and the drawing is \iocro, the two cycles cross an odd number of times, which implies both of them are essential, contradicting the assumption that $H$ does not contain an essential cycle.

In $G$, we now replace $H$ with a path of length two between $u$ and $v$ (following a $uv$-path in $H$ and suppressing all but one interior vertex). Then the resulting graph $G^*$ satisfies $G^* \prec G$ (since $H$ contains at least three edges). Moreover, $G^*$ still contains an $X$-configuration, with $C$ unchanged, and $C_0$ possible shortened: if $C_0 \cap H$ is non-empty, then that piece of $C_0$ is replaced by the new $uv$-path of length two. By minimality of $G$, there is a weakly compatible embedding of $G^*$ in the torus.

Let $H^*$ be $H$, together with a path $P$ between $v$ and $u$
internally disjoint from $H$ (such a $P$ exists, since otherwise $u$ is a cut-vertex, contradicting Lemma~\ref{lem:nearly-3-connected}. Then any essential cycle in $H^*$ must use $P$ (since $H$ does not contain any essential cycles by assumption), and, therefore, some interior vertex $x$ of $P$. We can then apply Lemma~\ref{lem:G-x} to show that $H^*$ has a plane embedding, in which the embedding of $H$ is compatible. Removing the internal vertices of
$P$ from the drawing of $H^*$, yields a compatible plane embedding of $H$ in which $u$ and $v$ lie in the same (outer) face. We can the insert this embedding in place of the $uv$-path in $G^*$ which replaced $H$, to obtain a weakly compatible embedding of $G$. Since $u$ is odd, we do not have to recover the rotation at $u$.

Again, we have reached a contradiction, so we can conclude that $H$ is a path (possibly a single edge).
\end{proof}

\subsection{Proof of Lemma~\ref{lem:XCbridges}}\label{sec:PLCB}


We restate the core of the lemma for reference:
\begin{quote}
  Suppose there is a $C$-bridge $H$ which attaches to $C$ in at least two vertices $u$ and $v$ that support independent odd pairs in $H$.
 \begin{itemize}
  \itemi If there is more than one $C$-bridge, then there is an $X$-configuration $(C^{+},C^{+}_0)$ in $G$
  for which there is only one $C^{+}$-bridge.
  \itemii If there is only one $C$-bridge, then $G$ is a subgraph of a subdivision of a $K_5$ or a $K_{3,t}$ with bracers.
 \end{itemize}
\end{quote}

Let the two independent odd pairs of $H$ at $u$ and $v$ be $(l_1^L$, $l_2^R)$ and $(l_3^L$, $l_4^R)$
with $l_1^L$, $l_2^R$ incident to $u^L$, $u^R$, and $l_3^L$, $l_4^R$ incident to $v^L$, $v^R$.\marginnote{$(l_1^L$, $l_2^R)$, $(l_3^L$, $l_4^R)$} See Figure~\ref{fig:P1P2DoNotInt1}.

We denote by $P_1^L$ a path in $H^{LR}$ from $u$ to $v$ starting with $l_1^L$ and ending with $l_3^L$. Similarly,
$P_2^R$ is a path in $H^{LR}$ from $u$ to $v$ starting with $l_2^R$ and ending with $l_4^R$. \marginnote{$P_1^L$, $P_2^R$}
Both paths $P_1^L$ and $P_2^R$ can be assumed to avoid vertices of $C$ in their interior.
We choose $P_1$ and $P_2$ so that the sum of the lengths of $P_1^L$ and $P_2^R$ is minimal; we will use this assumption in {\bf Case~(B2)} below at the very end of the proof.
Then, $P_1^L$ and $P_2^R$, are $L$ and $R$-diagonals.
Let $C_1$ and $C_2$, denote the cycles obtained by concatenating $P_1^L$ and $P_2^R$ with their $L$ and $R$-foundations on $C^L$ and $C^R$.\marginnote{ $C_1$, $C_2$}

A brief outline of the remainder of the proof: We distinguish two cases depending on whether $P_1^L$ and $P_2^R$ intersect.
In both cases we find, by the choice of $G'$, a vertex
$z$ that together with $u$ and $v$ forms
a set of vertices that due to the minimality of $G$
yields in $G$ a collection of $\{z,u,v\}$-bridges
that must be ``almost like'' three-stars.
The heart of the matter then is to reduce the instance so that Lemma~\ref{lem:K3nspokesHTWC}, or Lemma~\ref{lem:K5HTWC} apply. \\


For goal $(i)$, replacing $C$ with an essential cycle having only a single bridge, we will use a specific way to short-cut a path. Suppose $P$ is a path in $G'$.
A {\em $P$-shortcut} is a path\marginnote{$P$-shortcut, $\sct(P)$}
that is obtained from $P$ by the following iterative procedure. Start with $P_0=P$. If there exists a $P_i$-bridge $B$ (which is also a $P$-bridge) with a pair of feet $x$ and $y$, that is an edge $xy$ or a subdivided edge, that is, a path with all the internal vertices of degree 2, and the subpath of $P_i$ between $x$ and $y$ is neither an edge nor a subdivided edge  then we construct $P_{i+1}$ from $P_i$ by replacing the subpath of $P_i$ between $x$ and $y$ by $B$. If no such $B$ exists then $P_i$ is a shortcut of $P$.
We write $\sct(P)$ for a $P$-shortcut; it may not be unique, but it always exists. Moreover,  $\sct(P)$ is a path in $G'$ connecting the same endpoints as $P$.

When $P$ or $\sct(P)$ correspond to a cycle in $G$, we will write $P$ in $G$ and $\sct(P)$ in $G$ to refer to those cycles.

\begin{claim}\label{clm:Xshort} Suppose $P$ is a path in $G'$ which corresponds to a cycle in $G$ and both $P \cap C^L$ and $P \cap C^R$ are paths, with at least one of them being a single vertex.
\begin{itemize}
\itemi $P$ and $\sct(P)$ are essential cycles in $G$.
\itemii If $(C,P)$ is an $X$-configuration, and $\sct(P)$ has a single bridge in $G$, then $\sct(P)$ is part of an $X$-configuration.
\itemiii Any bridge of $P$ in $G$ is part of a bridge of $\sct(P)$ in $G$, or has become part of $\sct(P)$.
\end{itemize}
\end{claim}
\begin{proof}

$(i)$.  Since $P$ is a path in $G'$ and a cycle in $G$, it must connect  $x^L$ on $C^L$ to $x^R$ on $C^R$ for some $x \in V(G)$.  By symmetry we can assume that $x^L$ is the only vertex of $P$ on $C^L$. The cycle in $G$ corresponding to $P$  must be essential, since a  generic perturbation of the cycle crosses $C$ an odd number of times. The shortcut $\sct(P)$ also is a path between $x^L$ and $x^R$ in $G'$. By the construction of  $\sct(P)$, $x$ is also the only vertex  such that $x^L$  is on $\sct(P)$. Indeed, only degree-2 can be on $\sct(P)$ and not on $P$. However, degree-2 vertices on $C$ are not feet of any $C$-bridge, and therefore every degree-2 vertex $y^L$ on $C^L$ is not adjacent to any vertex of $P$ besides possibly $x^L$.
It follows that  $y^L$ cannot be contained in $\sct(P)$, since there does not exist a $P$-bridge with two feet containing $y^L$.
Since $\sct(P)$ intersect $C^L$ only in its end vertex, a  generic perturbation of the cycle corresponding to   $\sct(P)$ crosses $C$ an odd number of times. Thus,~$(i)$ follows.

$(ii)$.  We can then make the edges of $\sct(P)$ even (in $G$), using Lemma~\ref{lem:G-C}. Let $P \cap C = \{u\}$. We know that $u$ cannot be made even by edge-flips; since it still lies on $\sct(P)$, that means there are two edges incident to $u$ which cross oddly. Since $\sct(P)$ only has a single bridge, there is a cycle containing both edges, and no other vertices of $\sct(P)$. Then that cycle, together with $\sct(P)$, forms an $X$-configuration.

$(iii)$. The shortcut $\sct(P)$ consists of (a subset of the) vertices of $P$ and degree-$2$ vertices. Therefore, any bridge of $\sct(P)$ must attach to $\sct(P)$ in vertices of $P$, so it is, or contains, a bridge of $P$.
\end{proof}

We begin with the case that $P_1^L$ and $P_2^R$ do not intersect, illustrated in Figure~\ref{fig:P1P2DoNotInt1}.

 \begin{figure}[htp]
\centering
\includegraphics[scale=0.7]{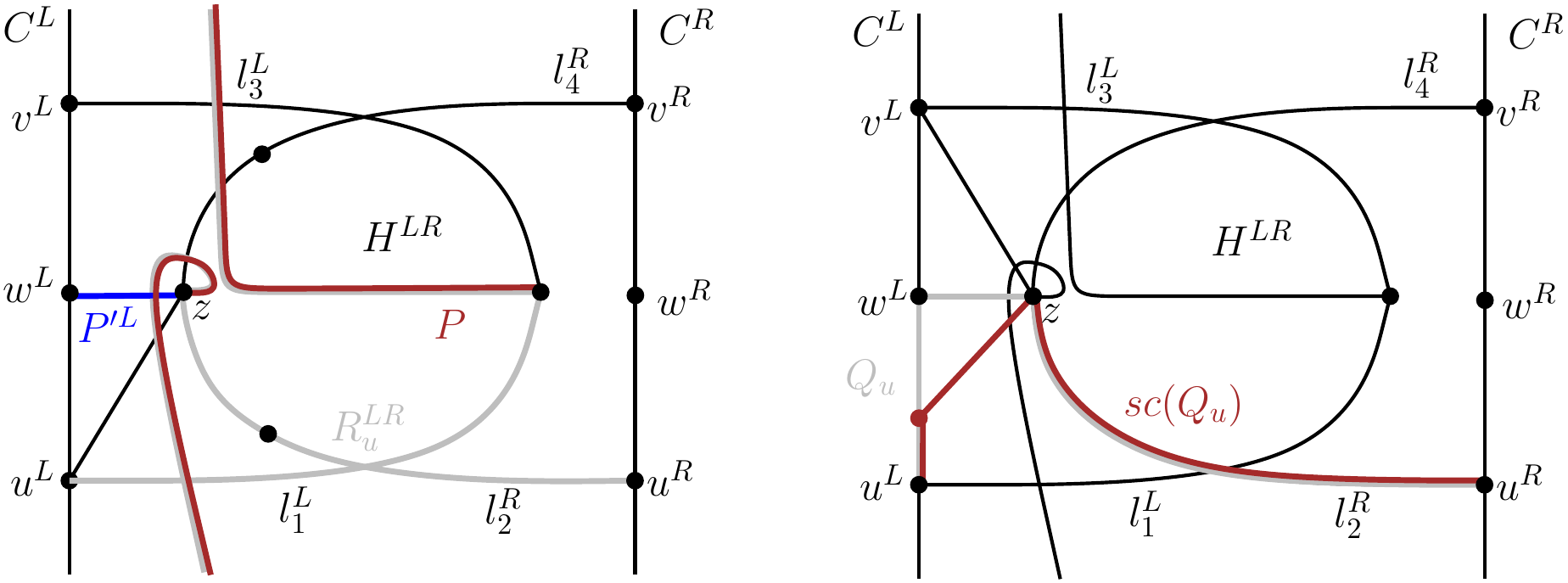}
 \caption{The case when $P_1^L$ and $P_2^R$ connecting $u^L,v^L$ and $u^R,v^R$, respectively, in $H^{LR}$ do not intersect in $G'$. The paths $P,P'^L, R_u^{LR}$  highlighted on the left, $Q_u$ and $sc(Q_u)$ highlighted on the right.}
\label{fig:P1P2DoNotInt1}
\end{figure}

\subsubsection*{CASE (A): $P_1^L$ and $P_2^R$ do not intersect.}

There exists a path $P$ in $H$ that joins a vertex of $P_1^L$ with a vertex $z$ of $P_2^R$.\marginnote{$P$ \bfseries{(A)}}
Since all the vertices of $P_1^L$ are contained in the interior of $C_2$ and all the vertices of $P_2^R$ are contained in the interior of $C_1$ all the feet of $H^{LR}$ are on the cycles $C_1$ and $C_2$.
By Lemma~\ref{lem:X2ftCBridge}, $H^{LR}$ has another foot on $C^L$ or $C^R$.
Let us say that foot is $w^L \neq u^L, v^L$ on $C^L$.\marginnote{$w^L$, $P'^L$ \bfseries{(A)}}
Let $P'^L$
be a path between $w^L$ and $P_2^R$ in $H^{LR}$.

For future reference, we give names to some of the $LR$-diagonals.
Let $Q_u$\marginnote{$Q_u$ \bfseries{(A)}}
denote the path in $G'$ consisting of the subpath of $P_2^R$ between $u^R$ and $z$; $P'^L$; and the subpath of $C^L$ between $w^L$ and $u^L$ not containing $v^L$.
Symmetrically, let $Q_v$\marginnote{$Q_v$ \bfseries{(A)}}
denote the path in $G'$ consisting of the subpath of $P_2^R$ between $v^R$ and $z$; $P'^L$; and the subpath of $C^L$ between $w^L$ and $v^L$ not containing $u^L$.
Let $R_u^{LR}$ \marginnote{$R_u^{LR}$ \bfseries{(A)}}
denote the $LR$-diagonal in $G'$ obtained by concatenating the part of $P_2^R$ between $u^R$ and $z; P$; and the part of $P_1^L$ between the end vertex of $P$ and $u^L$.
Symmetrically, we define $R_v^{LR}$. \marginnote{$R_v^{LR}$ \bfseries{(A)}}
Refer to Figure~\ref{fig:P1P2DoNotInt1}.
By Claim~\ref{clm:Xshort} all of $Q_u$, $Q_v, R_u^{LR}$ and $R_v^{LR}$ are essential cycles in $G$.
Then each of the pairs $(C \oplus C_1, Q_u)$, $(C \oplus C_2, Q_v)$, $(C, R_u^{LR})$, and $(C, R_v^{LR})$ contains the four edges at $u$ (or $v$) which cannot be made even, so each pair is an $X$-configuration, implying that
all of $Q_u, Q_v, R_u^{LR}$, and $R_v^{LR}$ are part of $X$-configurations.

\medskip

We collect some crucial properties of $z, w^L$ and  $P'^L$ which are needed in the remainder of the argument.

\begin{claim}
\label{clm:3rdfoot}
We have the following.
\begin{enumerate}[(a)]
\item \label{it:3rdfoota} Vertex $z$ is the end vertex of $P'^L$.
\item \label{it:3rdfootb} Vertices $\{w^L, z\}$ disconnects any interior vertex of $P'^L$ from the rest of $G'$.
\item \label{it:3rdfootc} Vertex $z$ is odd.
\item \label{it:3rdfootd}  Vertex $w^L$ has degree 1 in $H^{LR}$.
\end{enumerate}
\end{claim}

A statement analogous to the claim with the third foot being $w^R$ on $C^R$ holds by symmetry (we will use that variant of the claim below).

\begin{proof}
We start with the proof of~\eqref{it:3rdfoota}
Note that $P'^L$ is disjoint from $P_1^L$ (since its vertices lie outside $C_2$ which contains all vertices of $P_1^L$ in its interior).
The end vertex of $P'^L$ on $P_2^R$ must be $z$, the end vertex of $P$ on $P_2^R$: If $P'^L$ had an end vertex $z' \neq z$ on $P_2^R$,
there would be two vertex-disjoint cycles: consider the cycle starting with $w^L$, following $P'^L$ up to $z'$, then the part of $P_2^R$ either to $u^R$ or  $v^R$ whichever avoids $z$, and finally, the part of $C^R$ back to $w^R$ avoiding either $u^R$ or  $v^R$; this cycle is vertex-disjoint from either $R_u$ or $R_v$, depending on which of these was avoided by the first cycle. Hence, $P'^L$ ends in $z$ on $P_2^R$, which concludes the proof of~\eqref{it:3rdfoota}.

We continue with~\eqref{it:3rdfootb}. Claim~\eqref{it:3rdfootb} is true for the vertices in the core of $H$, since the core of $H$ is a tree, leaving only vertices of $C^L$, $u^R$ and $v^R$ due to the fact that the interior vertices of $P'^L$ are in the interior of $C_1$. If there were a path from an interior vertex
of $P'^L$ to a vertex of $C^L$, or $u^R$ or $v^R$, avoiding both $w$ and $z$, then such a path, together with a subpath of $C^L$, forms a path giving rise to a cycle in $G$ vertex-disjoint from an essential cycle corresponding to $R_u^{LR}$ or $R_v^{LR}$, since the subpath of $C^L$ can be chosen to avoid at least one of $u$ and $v$. This would contradict Lemma~\ref{lem:Xreduction333}.
By Lemma~\ref{lem:Xreduction333}, it also follows that the  $\{w^L, z\}$-bridge $B$ containing $P'^L$ is just a (subdivided) edge, i.e., all of its internal vertices are of degree two.
Indeed, if the bridge $B$ is not just a path it must contain a cycle disjoint from $z$ or $w^L$. The cycle is not contained in the core of $H$, which is a tree. The cycle must then pass through $w^L$ and not through $z$. Hence,  Lemma~\ref{lem:Xreduction333} applies
if we use $R_u^{LR}$ or $R_v^{LR}$.

Next, we prove~\eqref{it:3rdfootc}.
If $z$ is even, or can be made even by flips, then $Q_u$ and $R_v^{LR}$ would correspond to nearly disjoint essential cycles in $G$, contradicting Lemma~\ref{lem:2disjoint}.  Hence, $z$ is odd and cannot be made even by flips.

Finally, we prove~\eqref{it:3rdfootd}.
 The parts~\eqref{it:3rdfoota} and~\eqref{it:3rdfootb}, imply that any edge of $H^{LR}$ incident to $w^L$ must be contained in a $\{w^L, z\}$-bridge that is a path (possibly just an edge). Hence, if there is another such  edge incident to $P'^L$ besides the one in $P'^L$ then $\{w,z\}$ forms a 2-cut in $G$. Combining this with the oddness of $z$, proved in part~\eqref{it:3rdfootc}, Lemma~\ref{lem:X2-cuttwo} then implies that $w^L$ has degree $1$ in $H$, that is, the core of $H$ has only one leg ending in $w^L$, which concludes the proof.
\end{proof}

We are now in a position to prove $(i)$ for Case $(A)$.

\subsubsection*{CASE (A) - Establishing Property $(i)$.}

We would like to argue that $Q_u$ has only a single bridge; this need not be true, however. For example, if we subdivide the edge $u^Lz$ in $G'$ in the left picture of Figure~\ref{fig:P1P2DoNotInt1}, the resulting path is a bridge of $Q_u$. We therefore work with a $Q_u$-shortcut, $\sct(Q_u)$. Claim~\ref{clm:Xshort} tells us that $\sct(Q_u)$ is an essential cycle in $G$ and
that any bridge of $Q_u$ (in $G$) is contained in some bridge of $\sct(Q_u)$, or has become part of $\sct(Q_u)$.

Suppose then that $\sct(Q_u)$ contains a bridge $H_2$ other than the bridge containing $R_v^{LR}$. If $H_2$
has two feet on $\sct(Q_u)$ different from $z$, then $\sct(Q_u) \cup H_2$ contains an essential cycle avoiding $z$ (passing through the two feet),
which contradicts Lemma~\ref{lem:2disjoint}, since $R_v^{LR}$ corresponds to a disjoint essential cycle, so this does not happen.
Since $G$ is $2$-connected, by Lemma~\ref{lem:X2-connected}, $H_2$ must have exactly
two feet on $\sct(Q_u)$, one of which is $z$. The other foot of $H_2$ cannot be one
of the degree-two vertices short-cutting $Q_u$ (being a foot), and it cannot belong to the core of $H$, since that is a tree.
Therefore, the other foot must lie on $C$. Suppose the foot lies on $C^L$, and call it $w'^L$.
We allow the case $w'^L = u^L$, but note
that $w'^L \neq w^L$, since $w^L$ has degree $1$ in $H$ as we argued earlier.

By Claim~\ref{clm:3rdfoot}~\eqref{it:3rdfootb} and~\eqref{it:3rdfootd},  $H_2$ is a subdivided edge (in case $w'^L = u^L$ it {\em is} an edge, by Lemma~\ref{lem:X2-cutodd}, since both $z$ and $u$ are odd; we work with $X$-configuration $(C,Q_v)$ to apply the lemma). Now $\sct(Q_u)$ contains a path from $z$ to $w'^L$;
if this path contained any vertex of $Q_u$, it would have been replaced by $H_2$ in $\sct(P)$; since, by assumption, that
did not happen, the path from $z$ to $w'^L$ in $\sct(Q_u)$ contains no vertices of $Q_u$, which implies that it a (subdivided) edge.
This contradicts Lemma~\ref{lem:X2-cuttwo},
because $z$ is not even; we work with $X$-configuration  $(C, Q_v)$. The lemma applies, since  $w'$ belong to $C$ and not to $Q_v$, and $z$ does not belong to $C$.

Hence $\sct(Q_u)$ has only one bridge, the bridge containing $R_v^{LR}$. Moreover, the cycles corresponding to $\sct(Q_u)$ and $R_v^{LR}$ are essential and intersect in a single vertex $z$, which, as we saw, cannot be made even by flips. This concludes the argument of part $(i)$ of Case $(A)$.

\subsubsection*{CASE (A)- Establishing Property $(ii)$.}

For part $(ii)$ we can assume that $C$ has only a single bridge.


Suppose $C$ has length at least $4$. Then there must be a fourth vertex $w'$ on $C$.\marginnote{$w'$ \bfseries{(A)}}
If $w'$ is not incident to a leg of $H^{LR}$ we can suppress $w'$ thereby violating the minimality of the counterexample.
So $w'$ is incident to a leg of $H^{LR}$, and this leg must be contained in a cycle vertex-disjoint from an essential one,
contradicting Lemma~\ref{lem:2disjoint}: If $w'^L$ is incident to such a leg $l$, a cycle passing through $l$, and  containing $z$ and $w^L$ is vertex-disjoint from an essential cycle whose edge set is the symmetric difference of $E(C_1)$ and $E(C^L)$.
Otherwise, if $l$ is incident to $w'^R$; in this case, a cycle passing through $l$ and a subpath of $C_1$, and containing  $v^R$, say, is vertex-disjoint from $Q_u$. In both cases, we contradict Lemma~\ref{lem:2disjoint}.

Therefore, $C$ has length $3$. Consider the case that $w$ has degree $3$ in $G$. In that case we are done, by Lemma~\ref{lem:X3deg}, since the edge $uv$ violates the lemma: $u$ and $v$ are odd, and the essential cycle $\sct(Q_u)$ (or $\sct(Q_v)$), as we argued in part $(i)$, has a single $\sct(Q_u)$-bridge, and removing the edge $uv$ does not change that. If all vertices of $\sct(Q_u)$ can be made even by flips, then, by Corollary~\ref{cor:esseven}, $G - uv$ has a weakly compatible embedding in the torus, and to that embedding we can add back $uv$ to get a weakly compatible embedding of $G$ in the torus. Hence $\sct(Q_u)$ contains an odd vertex which cannot be made even by edge-flips, and, since $\sct(Q_u)$ has a single bridge, there must be a cycle $C''$ in $G - uv$ so that $(\sct(Q_u), C'')$ is an $X$-configuration, contradicting Lemma~\ref{lem:X3deg}.

We conclude that $w$ has degree at least $4$ and $C$ consists of three vertices only, $u,v$ and $w$.
We argued earlier that $w^L$ has degree at most one in a $C$-bridge.
 Claim~\ref{clm:3rdfoot}~\eqref{it:3rdfootb} and~\eqref{it:3rdfootd}
applies also to $w^R$.
Hence, $w$ has degree exactly $4$, with both $w^L$ and $w^R$
being of degree $3$ in $G'$. We can then define $Q_u$ and $Q_v$ with $w^R$ playing the role of $w^L$.
It follows (by symmetry) that in $G$ there exists an essential cycle avoiding $z$ and $v$, and an essential cycle avoiding $z$ and $u$ that can play the role of $C$ similarly as the cycles corresponding to $\sct(Q_u)$ or $\sct(Q_v)$. Moreover, there is a path $P'^R$ between $w^R$ and $z'$, the endpoint of $P$ on $P_1^L$; we have $z \neq z'$ since $P_1^L$ and $P_2^R$ do not intersect by assumption. As we did with $P'^L$, we can show that $P'^R$ is a subdivided edge, and $z'$ is odd.

Let us summarize what we know at this point: $C$ consists of vertices $u$, $v$, and $w$; $w$ has degree $4$ in $G$ and is connected to $u$ and $v$ by an edge, and to $z$ and $z'$ by subdivided edges. We claim that $G$ has no other $\{u,v,z\}$-bridge (attached at all three vertices) besides the one containing $w$. If it did,
we could replace $zu^L$ in $Q_u$ (for $w^L$) by a path from $z$ to $u$ through that bridge (avoiding $v$), creating a vertex-disjoint cycle to the $Q_v$ we constructed for $w^R$ (which passes through $w_R$ so belongs to the bridge containing $w$). This would contradict Lemma~\ref{lem:2disjoint}. Since $z$ and $z'$ are odd and do not lie on $C$, Lemma~\ref{lem:X2-cutodd} implies that $\{z, z'\}$ is not a cut, and similarly $\{u, z\}$ and $\{v, z\}$.
Since there are no cycles disjoint from $C$, and we have accounted for all edges attaching to $C$, this means that $z$ and $z'$ if they are connected,  are connected by an edge.
We conclude that $G$ is a (subgraph of a) subdivision of $K_5$ with all vertices, $u$, $v$, $w$, $z$, and $z'$ odd. By Lemma~\ref{lem:K5HTWC} such a $G$ is not a counterexample.

\smallskip

\subsubsection*{ CASE (B) $P_1^L$ and $P_2^R$ intersect.}

The intersection of $P_1^L$ and $P_2^R$ must be a path, possibly consisting of a single vertex: If $x$ and $y$ are interior vertices of $P_1^L$,
then the unique path between $x$ and $y$ in the core of $H^{LR}$ (which is a tree) must belong to $P_1^L$. The same is true
for $P_2^{R}$, showing that $P_1^L$ and $P_2^R$ intersect in a path.

First, we consider the case that $C_1$ and $C_2$ are (edge)-complements of each other on $C$.

\subsubsection*{ CASE (B1) $C_1$ and $C_2$ are complementary on $C$.}

See Figure~\ref{fig:rotationAtZ}~(left) for an illustration; in the figure the red subpath of $C^L$ is the complement of the grey subpath of $C^R$.

We correct the rotation at the vertices of $P_1^L$ so that the edges of $P_1^L$ cross every other {\em adjacent} edge in $G'$ an even number of times. Let $H'^{LR}$ denote the union of $P_1^L$ and $P_2^R$.\marginnote{$H'^{LR}$ \bfseries{(B1)}}

Let $R_u^{LR}$ and $R_v^{LR}$\marginnote{$R_u^{LR}$, $R_v^{LR}$ \bfseries{(B1)}}
denote the paths joining $u^L$ and $u^R$, and $v^L$ and $v^R$, respectively, in $H'^{LR}$. Since $C$ is essential, and $C_1$ and  $C_2$ are non-essential,
$C_1 \oplus (C_2 \oplus C)=R_u \oplus R_v$ has non vanishing homology over $\ZN_2$. Hence, $R_u$ and $R_v$ as curves cross an odd number of times. Indeed, $R_u$ and $R_v$ are essential cycles in $G$, and therefore belong to different non-vanishing homology class over $\ZN_2$.
It follows that the order of end vertices of $P_1^L\cap P_2^R$ along $P_1^L$ and $P_2^R$ when traversing  $P_1^L$ and $P_2^R$, respectively, from $u^L$ to $v^L$, and from $u^R$ to $v^R$, is reversed.

Let $w^R$\marginnote{$w^R$ \bfseries{(B1)}}
denote a third foot of $H^{LR}$ (a third foot must exist by Lemma~\ref{lem:X2ftCBridge}, we arbitrarily assume it lies on $C^R$, the case $C^L$ is symmetric). 

Let $z$\marginnote{$z$, $P^R$ \bfseries{(B1)}}
denote the vertex in  $H'$, but not on $C$, joined by a path $P^R$ connecting $H'^{LR}$ with $w^R$ (internally disjoint from $C^R,C^L$ and $H'^{LR}$).

\begin{claim}
\label{clm:z}
The vertex $z$ lies in the intersection of $P_1^{L}$ and $P_2^{R}$, and there exists no pair of vertex-disjoint paths in $G'$ internally disjoint from $H'^{LR}$ joining the intersection of $P_1^{L}$ and $P_2^{R}$ with $C^R$, and the same applies to $C^L$ by symmetry.
\end{claim}
\begin{proof}
Refer to  Figure~\ref{fig:P1P2DoInt0}.
If the first part of the claim does not hold we contradict  Lemma~\ref{lem:2disjoint}, see the top-left part of the figure. If the second part of the claim does not hold we are done by the same token, see the top-right and bottom part of the figure.
\end{proof}

\begin{figure}[htp]
\centering
\includegraphics[scale=0.7]{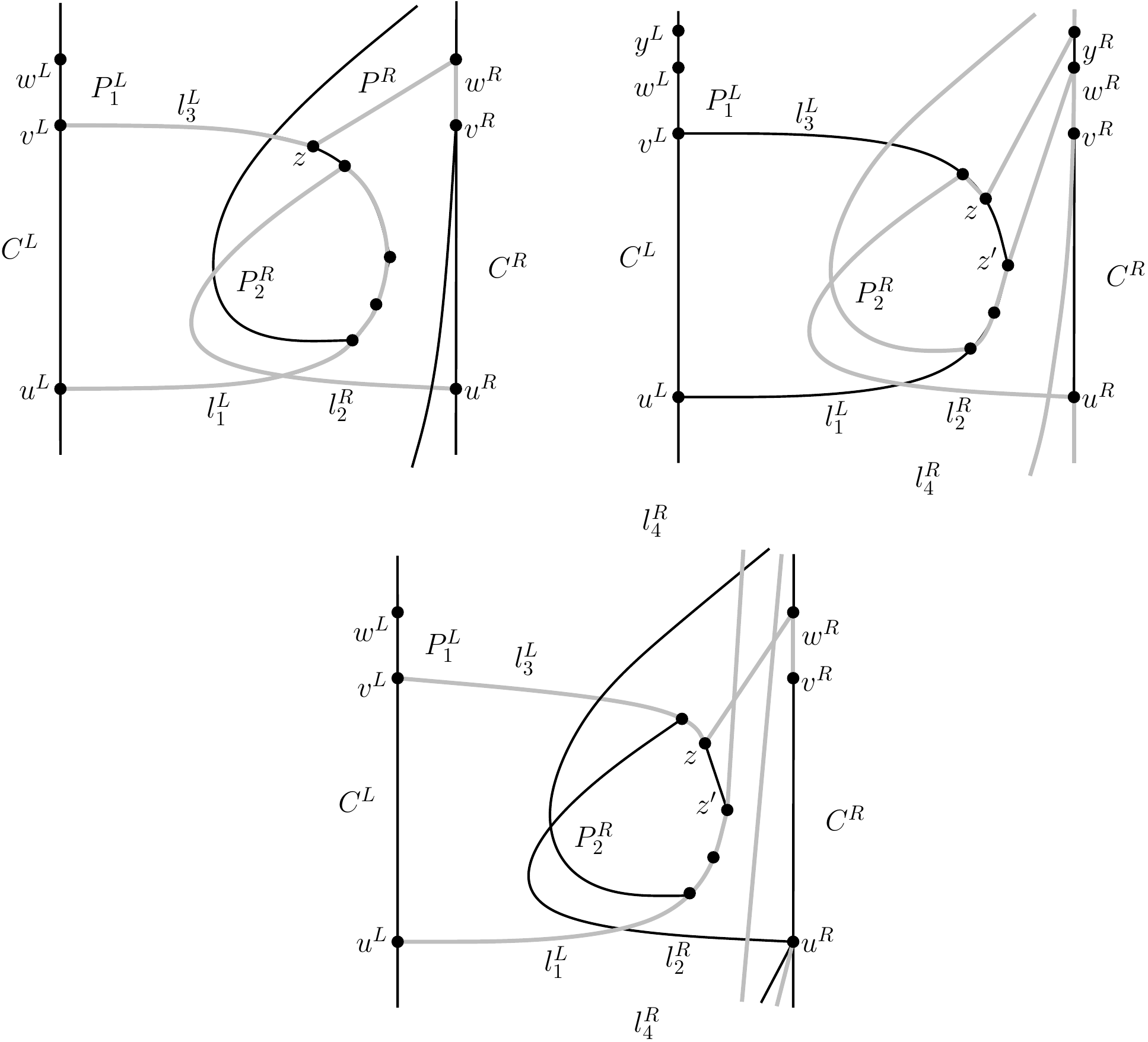}
 \caption{$P_1^L$ and $P_2^R$  intersect. A pair of gray essential cycles violating Lemma~\ref{lem:2disjoint} if $z$ is not in the intersection of $P_1^L$ and $P_2^R$ (left). A pair of disjoint paths joining the intersection of $P_1^L$ and $P_2^R$ with $C^R$ yielding a pair of gray essential cycles violating Lemma~\ref{lem:2disjoint} (right,bottom). }
\label{fig:P1P2DoInt0}
\end{figure}

 \begin{figure}[htp]
\centering
\includegraphics[scale=0.7]{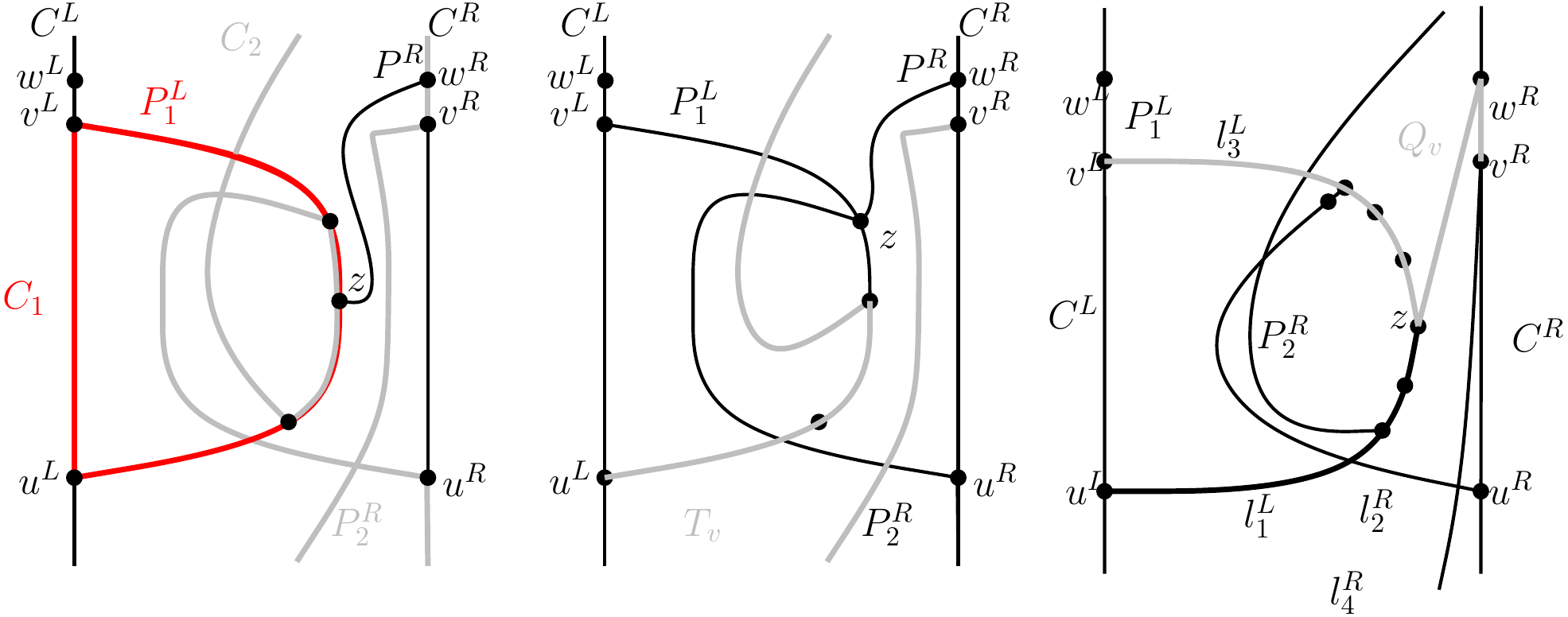}
 \caption{The interiors of $C_1$ (red), and $C_2$ (grey) in the regions incident to their vertices (left). The dashed edges (paths) are not possible (right). The one between $z$ and $w^L$ due to Lemma~\ref{lem:2disjoint}. }
\label{fig:rotationAtZ}
\end{figure}

Recall that $R_u^{LR}$ is the path in $H'^{LR}$ joining $u^L$ with $u^R$, and similarly for $R_v^{LR}$. By choice of $u$ and $v$
both $(C, R_u^{LR})$ and $(C,R_v^{LR})$ are $X$-configurations, but $R_u^{LR}$, and $R_v^{LR}$, may have more than one bridge.
That is, why we work with $\sct(R_u^{LR})$, and $\sct(R_v^{LR})$, the $R_u^{LR}$- and $R_v^{LR}$-shortcuts.

\subsubsection*{CASE (B1)- Establishing Property $(i)$.}

We prove property $(i)$ by showing that $\sct(R_u^{LR})$ is an essential cycle with a single bridge (and similarly for $\sct(R_v^{LR})$). By Claim~\ref{clm:Xshort}, $\sct(R_u^{LR})$ is an essential cycle, so we have to show that $\sct(R_u^{LR})$ only has a single bridge, which then implies that $\sct(R_u^{LR})$ is part of an $X$-configuration by Claim~\ref{clm:Xshort}.


We will work with an essential cycle $C_3$\marginnote{$C_3$ \bfseries{(B1)}}
in $G$ defined as follows: The cycle $C_3$ combines $P^R$, the subpath of $C^R$ connecting $w^R$ with $v^R$ avoiding $u^R$, and the part of $P_2^R$ between $v^R$ and $z$. See Figure~\ref{fig:P1P2DoInt1a}. The cycle $C_3$ is essential due to the fact that its small generic perturbation crosses $R_u^{LR}$ an odd number of times.

 \begin{figure}[htp]
\centering
\includegraphics[scale=0.7]{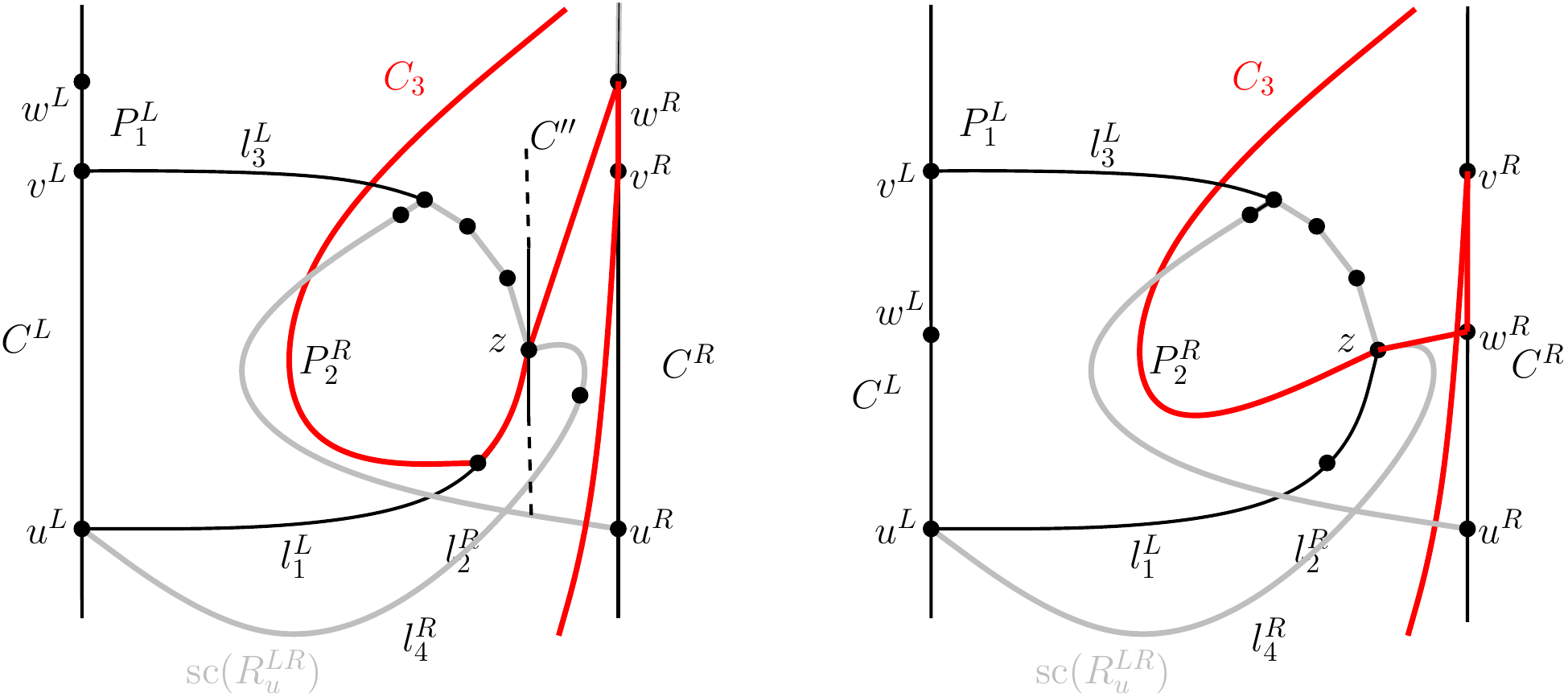}
 \caption{Path $\sct(R_u^{LR})$ and the cycle $C_3$. The dashed path cannot exist if $z$ is even since
 it yields a pair of nearly disjoint essential cycles, which in $G'$ form a pair of paths between $u^L$ and $u^R$, and $v^L$ and $v^R$, respectively. We depict the two cases depending on the position of $w^R$ on $C^R$ with respect to $C_2$.}
\label{fig:P1P2DoInt1a}
\end{figure}

Suppose there were a $\sct(R_u^{LR})$-bridge $H_2$ different from the bridge containing $C- u$. Then $H_2$ cannot have two
feet in the interior of $\sct(R_u^{LR})$, since this would imply that $H_2 \cup \sct(R_u^{LR}) - u$ contains a cycle vertex-disjoint from $C$, which contradicts Lemma~\ref{lem:Xreduction333}. Since $H_2$ must have at least two feet (by Lemma~\ref{lem:X2-connected}), those feet must be $\{u,u'\}$, where $u'$ is an interior vertex of $\sct(R_u^{LR})$. Since
vertices of $\sct(R_u^{LR})$ have either degree $2$ or belong to $R_u^{LR}$, we know that $u'$, as a foot, must belong
to $R_u^{LR}$.

Since there are two $\sct(R_u^{LR})$-bridges, $\{u,u'\}$ must be a cut-set. We can then apply Lemma~\ref{lem:X2-cutodd}, with $X$-configuration
$(C, R_v^{LR})$, to conclude that $u'$ is even.

Let $F$ denote the union of connected components of $G- \{u',u\}$  not containing $v$ (this also includes an edge $u'u$ if it exists).
Let $H''=G[V(F)\cup \{u',u\}]$. By Lemma~\ref{lem:X2-cut2}, either $H''$ has only a single edge adjacent to $u'$ or there is an essential cycle $C''$  in $H''$. In the first case, $H''$ is a (subdivided) edge---by Lemma~\ref{lem:Xreduction333}---between $u$ and $u'$, which cannot be the case due to the existence of $H_2$.

Hence, $H''$ contains an essential cycle $C''$. $C''$ must pass through $u$ (otherwise $C''$ and $C$ are vertex-disjoint,
contradicting Lemma~\ref{lem:2disjoint}).
If $u' \neq z$, then $u'$ lies either before or after $z$ on the path $R_u^{LR}$ from $u^L$ to $u^R$. If $u'$ lies on
the $u^L$-to-$z$ part of $R_u^{LR}$, then $C''$ is vertex-disjoint from the
essential cycle $Q_v$\marginnote{$Q_v$ \bfseries{(B1)}}
through $v$ and $z$ consisting of the part of $P_1^L$ from $v^L$ to $z$ followed by $P^R$ and the part of $C^R$ from $w^R$ to $v^R$ avoiding $u^R$; if $u'$ lies on the $z$-to-$u^R$  part of $R_u^{LR}$, then $C''$ is vertex-disjoint from the essential cycle $C_3$. In both cases, Lemma~\ref{lem:2disjoint} applies, establishing a contradiction.

We conclude that $u' = z$. Recall that $u'$ is even. Therefore, if $C''$ touches either $Q_v$ or $C_3$ in $u' = z$, we are done by
Lemma~\ref{lem:2disjoint}. Hence $C''$ crosses both of those cycles in $z$. Then $C''$ must be nearly disjoint from $R_v^{LR}$, see Figure~\ref{fig:P1P2DoInt1a}. This contradicts Lemma~\ref{lem:2disjoint}.

This completes the proof of $(i)$ in case (B1).
\medskip

Before turning to $(ii)$, we show that there are two more paths in $G'$ that can (potentially) play the role of $C$, in the sense of being essential with a single bridge, and therefore part of an $X$-configuration.

Refer to Figure~\ref{fig:rotationAtZ}~(middle-right). Let  $Q_u$\marginnote{$Q_u$ \bfseries{(B1)}}
be the path in $G'$ obtained by following $P_1^L$ from $u^L$ to $z$, then following $P^R$ from $z$ to $w^R$, and $C^R$ from $w^R$ to $u^R$ avoiding $v^R$. We already defined $Q_v$ earlier ($P_1^L$ from $v^L$ to $z$, then $P^R$ from $z$ to $w^R$ and $C^R$ from $w^R$ to $v^R$ avoiding $u^R$). By Claim~\ref{clm:Xshort}, both $Q_u$ and $Q_v$ are essential cycles in $G$.
Let  $T_v$\marginnote{$T_v$ \bfseries{(B1)}}  denote the cycle
starting at $v^R$ following $P_2^R$ until it intersects $P_1^L$; from that intersection continue along $P_1^L$ to $u^L$ and, along $C^L$ from $u^L$ to $v^L$, avoiding $w^L$.
Similarly, $T_u$ \marginnote{$T_u$ \bfseries{(B1)}}
starts at $v^R$ and follows $P_2^R$ until it intersects $P_1^L$, which it follows to $u^L$; then from $u^L$ to $v^L$ along $C^L$ avoiding $w^L$.

If $z$ does not lie on $T_u$, then $(Q_u, T_u)$ forms an $X$-configuration; if $z$ does not lie on $T_v$, then $(Q_v, T_v)$ is an $X$-configuration. If $P_1^L$ and $P_2^R$ intersect in at least one edge, then $z$ cannot belong to both $T_u$ and $T_v$, so one
of the two cases holds.

\subsubsection*{CASE (B1)- Establishing Property $(ii)$. $P_1^L$ and $P_2^R$ contains at least one edge.}

We can now complete the proof of $(ii)$ in case that the intersection of $P_1^L$ and $P_2^R$ contains at least one edge.
 We can use the fact that either $(Q_u, T_u)$ or  $(Q_v, T_v)$ is an $X$-configuration (though we may not know which of the two it is). First, we show that  \\

$(\circ) \ $ $w^L$ cannot be a foot of $H^{LR}$. \\

Suppose for the sake of contradiction that $w^L$ is a foot of $H^{LR}$.
 Let $P_w^L$ denote the shortest path  between $w^L$ and $H'^{LR}$, that is completely contained in $H^{LR}$. We distinguish two cases depending on whether  $P_w^L$ ends in $z$.
If $P_w^L$  ends in $z$, then the essential cycle consisting of $P_w^L$, $P^R$ would be vertex-disjoint from either $T_u$ or $T_v$ (here we use that $P_1^L$ and $P_2^R$ have at least one edge in common), which is a contradiction.
Therefore, $P_w^L$ ends in a vertex $y\not = z,w^L$ on $H'^{LR}$. Then $y$ is
in the interior of the cycle $C_u^R$ formed by $P_u^R$ (or $P_v^R$ in which case a symmetric argument applies) and its $R$-foundation.
Indeed, the part of $P_1^L$ between $u^L$ and $y$ crosses $P_2^R$, and hence, $P_u^R$,  an odd number of times. If a pair of legs incident to $w^L$ and $w^R$ form an independent odd pair, $w$ can play the role of $v$ and we end up in {\bf CASE~(A)}.
Otherwise, since $y$ is in the interior of  $C_u^R$,
$P_u^R$ must intersect $P_w^L$ an odd number of times, which is not possible (contradiction). This completes the proof of $(\circ)$.\\

If $C$ consists of three vertices only, $u,v$ and $w$ then similarly as in {\bf  CASE (A)} we are done by Lemma~\ref{lem:X3deg}:
By~$(\circ)$, $w^L$ is only incident to the two $C$-edges (since $H$ is the only $C$-bridge now). In $G$, we can then split $w$ to make it of degree three and remove the edge $uv$ (this contradicts the minimality of $G$, the number of edges remains the same, but the number of vertices increased).

Hence, there must be a fourth foot of $H^{LR}$. We want to show that this implies that $z$ is odd. Suppose the fourth foot lies on $C^R$, call it $w'^R$.\marginnote{$w'^R$ \bfseries{(B1)}}
By Claim~\ref{clm:z}, the shortest path $S^R$ in $H^{LR}$ between $w'^R$ and $P_1^L \cup P_2^R$ must attach in $z$. If $z$ is even, then
$S$ attaches outside of $C_1$, and inside $C_2$, which means that $w'^R$ lies on $C_2$, like $w^R$.
Then the cycle formed by $P^R$, $S^R$, and the path from $w^R$ to $w'^R$ along $C^R$ avoiding $u^R$ is nearly disjoint from
from $\sct(R_u^{LR})$, contradicting Lemma~\ref{lem:Xreduction333}. We conclude that $z$ is odd in this case. We proved the following. \\

(*) If $w'^R$ is a foot of $H^{LR}$ different from $u^R,v^R$ and $w^R$ then the shortest path $S^R$ (in $H^{LR}$) between $w'^R$ and $H'^{LR}=P_1^L \cup P_2^R$ ends in $z$, which must be odd in this case. \\

Suppose the fourth foot of $H^{LR}$ lies on $C^L$, call it $w'^L$.\marginnote{$w'^L$ \bfseries{(B1)}}
Let $S^L$ be a shortest path between $w'^L$ and $H'^{LR}=P_1^L \cup P_2^R$. Similarly as $S^R$ from the above, path $S^L$ must end in the intersection $P_1^L \cap P_2^R$. Indeed, otherwise its concatenation with a part of  $P_1^L$ or $P_2^R$, and a part $C^L$ or $C^R$ is a cycle in $G$ that is vertex-disjoint from $R_u^{LR}$ or $R_v^{LR}$ (contradiction by Lemma~\ref{lem:Xreduction333}), in particular, it attaches to $P_1^L$.

We need the following analog of (*). \\

(**)  If $w'^L$ is a foot of $H^{LR}$ different from $u^L,v^L$ and $w^L$ then the shortest path $S^L$ (in $H^{LR}$) between $w'^L$ and $H'^{LR}=P_1^L \cup P_2^R$ ends in $z$, which must be odd in this case. \\

 Suppose $S^L$ instead attaches at a vertex $z' \neq z$. We distinguish four cases, based on whether $z'$ occurs before or after $z$ on $P_1^L$  on the way from $u^L$, and whether $w$ and $w'$ belong to the same, or different parts of $C - \{u,v\}$. Let us first assume that $w$ and $w'$ belong to different parts of $C - \{u,v\}$. Let us also assume that $z'$ occurs before $z$ on $P_1^L$. As we argued in part $(i)$, in this case $z$ does not lie on $T_v$, and so $Q_v$ is part of a $X$-configuration $(Q_v, T_v)$. But $Q_v$ is vertex-disjoint from $W_u'$ starting at $w'^L$, following $S^L$ to $P_1^L$, then following $P_1^L$ to $u^L$, and $C^L$ from $u^L$ to $w'^L$ avoiding $v^L$. This contradicts Lemma~\ref{lem:Xreduction333}.  If, instead, $z'$ occurs after $z$ on $P_1^L$, we know that $z$ does not lie on $T_u$, so $(Q_u, T_u)$ is an $X$-configuration. In this case, consider the cycle $W_v'$ starting at $w'^L$, following $S^L$ to $P_1^L$, then following $P_1^L$ to $v^L$, and $C^L$ from $v^L$ to $w'^L$ avoiding $u^L$. $W_v'$ is disjoint from $Q_v$, so we have a contradiction by  Lemma~\ref{lem:Xreduction333}.

We can therefore assume that $w$ and $w'$ belong to the same part of $C - \{u,v\}$. Suppose the ordering along $C$ is $u, w, w', v$. If $z'$ occurs before $z$ on $P_1^L$, consider $W_v$ starting at $v_R$, following $P_2^R$ until it intersects $P_1^L$, then following $P_1^L$ to $z'$, then $S$ to $w'^L$ and $C^L$ to $v^L$ avoiding $w^L$ and $u^L$. The second cycle is $W_u$ starting at $u^L$ following $P_1^L$ to $z$, then $P$ to $w^R$ and $C^R$ to $u^R$ avoiding $w'^R$ and $v^R$. Cycles $W_v$ and $W_u$ are vertex-disjoint, and, both are essential. Therefore Lemma~\ref{lem:2disjoint} applies and we are done in this case.
If the ordering along $C$ is $u, w', w, v$, we can swap sides and apply the same arguments to obtain a contradiction.

We conclude that $S^L$ attaches in $z$. Suppose $z$ were even. If the foot of $S^L$ on $C^L$ belongs to $C^L-C_1$, then $S^L$ attaches to $z$ from outside $C_1$ (like $P^R$), and the cycle
consisting of $P^R$, $S^L$, and the path from $w^R$ to $w'^R$ along $C^R$ avoiding $u^R$ is nearly disjoint from
from $\sct(R_u^{LR})$, contradicting Lemma~\ref{lem:Xreduction333}. Hence, the foot of $S^L$ must belong to
$C_1 \cap C^L$. Then the cycle consisting of $S$ (which must attach to $z$ from inside $C_1$) together
with the subpaths of $C_1$ connecting it to $u^L$ touches $Q_v$ and the cycle consisting of $S$ together
with the subpaths of $C_1$ to $v^L$ touches $Q_u$. Since at least one of $Q_u$ or $Q_v$ belongs to an $X$-configuration,
we can invoke Lemma~\ref{lem:Xreduction333} to conclude that this is a contradiction. Again, we conclude that $z$ is odd thereby completing the proof of~(**).

Together, (*) and (**) tell us that in the presence of a fourth foot, $z$ must be an odd vertex.
By Lemma~\ref{lem:X2-cutodd}, neither $\{z,v\}$,
nor $\{z,u\}$ are $2$-cuts (work with $(C,R_u)$ and $(C,R_v)$); and neither is $\{u,v\}$ (since there is only one $C$-bridge, and an additional foot). It follows that there are no $\{z,v\}$-,
$\{z,u\}$- or $\{u,v\}$-bridges, except, possibly, the edges $zu$, $zv$, and (certainly) $uv$. Since there are no cut-vertices,
this means that except for these three edges, $G$ is the union of $\{u,v,z\}$-bridges, with each bridge attaching to all three vertices.
We would like to show that every $\{u,v,z\}$-bridge is (a subdivision of) $K_{1,3}$, since this proves $(ii)$.
We distinguish two cases:

\subsubsection*{Subcase 1. The vertex $w^R$ is incident to exactly one leg of $H^{LR}$.}
Without loss of generality, we assume that $z$ does not lie on $T_v$, so $(Q_v,T_v)$ forms an $X$-configuration.
Let $H_w$\marginnote{$H_w$ \bfseries{(B1)}}
denote the $\{u,v,z\}$-bridge containing $w^R$, and thus $P^R$.

We will show that that $H_w$  is a (subdivision) of $K_{1,3}$, but first,
using that  $w^R$ is not incident to more than one leg of $H^{LR}$, we prove that \\

(***) the two edges of $P_1^L$ incident to $z$ belong to different $\{u,v,z\}$-bridges (or $\{u,z\}$, $\{v,z\}$-bridges) and both bridges are different
from $H_w$.  \\

For a contradiction, assume two of those bridges are the same; then there is a cycle $C_z$ passing through the two $P_1^L$-edges incident to $z$, or through one of those edges, and an edge belonging to $H_w$. Since $C_z$ belongs to a $\{u,v,z\}$-bridge (or $\{u,z\}$, $\{v,z\}$-bridge), it avoids both $u$ and $v$. Then the cycle $C_z$ must still intersect $C$ due to Lemma~\ref{lem:Xreduction333}, and therefore we either violate~(*) or~(**).

Next, we show that $H_w$ can have only one leg at $z$. If there were two legs, then there would be a cycle $C_z$
through $z$ in $H_w$. By~(***), $H_w$ does not contain any of the two edges of $P_1^L$ incident to $z$, neither does $C_z$, so $C_z$
is disjoint from  $T_v$, contradicting  Lemma~\ref{lem:Xreduction333}; here we use that $z$ does not belong to $T_v$. Furthermore, a cycle avoiding $z$ and $u$ (respectively, $v$) in $H_w$ is disjoint from $R_u^{LR}$ (respectively, $R_v^{LR}$), which is again a contradiction with Lemma~\ref{lem:Xreduction333}.
Hence, $H_w$ cannot contain any cycle. Finally, there are no cut-vertices by Lemma~\ref{lem:X2-connected} in $G$, and therefore $H_w$ is a (subdivision of) $K_{1,3}$.

By~(***), $Q_v$ passes through a trivial $\{v,z\}$-bridge or $\{u,v,z\}$-bridge $H_{Q_v}$ and $H_w$; and $T_v$ is contained in a single $\{u,v,z\}$-bridge $H_{T_v}$ which $Q_v$ meets only in $v$. Both $H_{T_v}$ and $H_{Q_v}$ have a very simple structure which we delve further into next.

\begin{figure}[htp]
\centering
\includegraphics[scale=0.7]{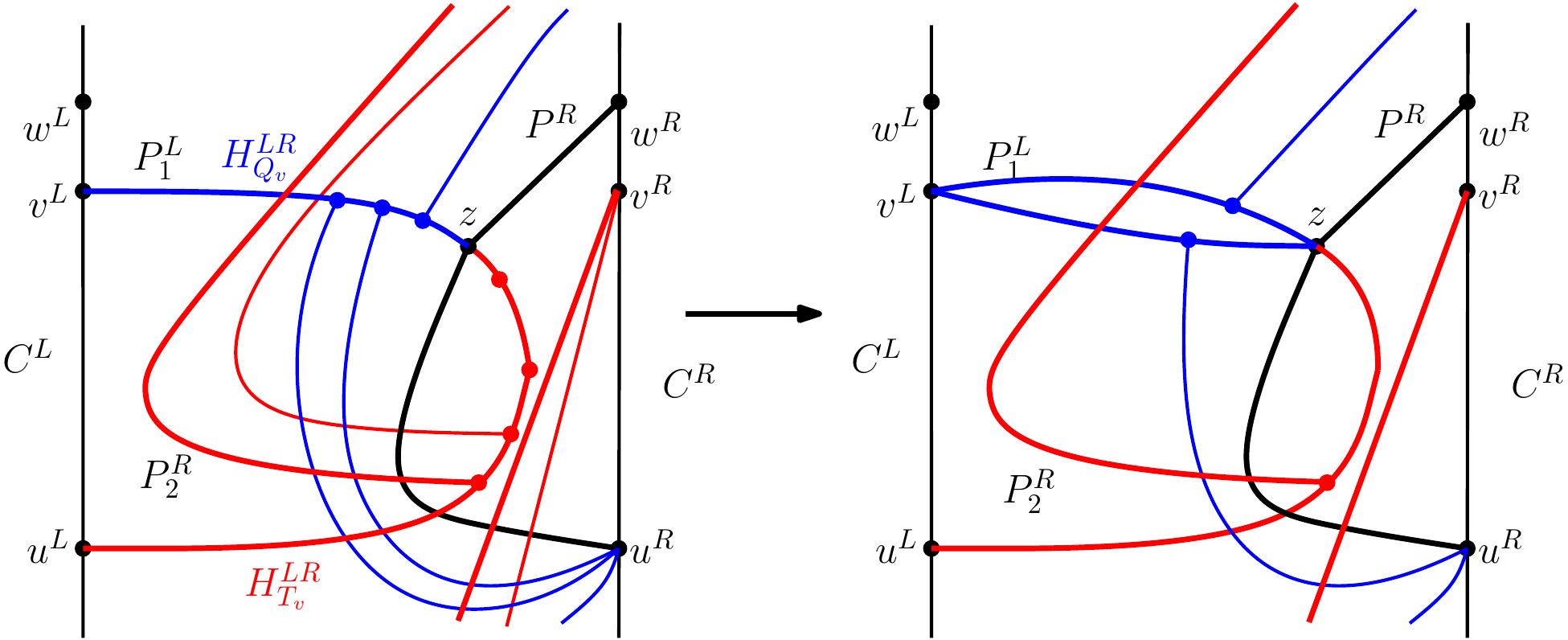}
 \caption{The bridges $H_{T_v}^{LR}$ and $H_{Q_v}^{LR}$ (left). The doubling of the path in $H_{Q_v}^{LR}$ between $z$ and $v^L$ (right).}
\label{fig:doubling}
\end{figure}

\medskip

Refer to Figure~\ref{fig:doubling}.
 We  show that   \\

 (****)~The vertex $v$  is  incident to a single edge in $H_{Q_v}$, and that   $u$ is incident to a single edge in $H_{T_v}$. \\

If $H_{T_v}$ contains a cycle passing through $u$ and disjoint from $v$ and $z$ then this cycle is disjoint from $Q_v$ (contradiction with  Lemma~\ref{lem:Xreduction333}).
If $H_{Q_v}$ contains a cycle passing through $v$ and disjoint from $u$ and $z$ then this cycle is disjoint either from $R_u^{LR}$ or $Q_u$ depending on whether $z$ is the end point of $P_1^L \cap P_2^R$. In particular, if $z$ is  the end point of $P_1^L \cap P_2^R$ then it is disjoint from $R_u^{LR}$,  and otherwise it is disjoint from $Q_u$ and also  $(Q_u,T_u)$ forms an $X$-configuration in this case, and hence, we can apply  Lemma~\ref{lem:Xreduction333} (to derive a contradiction).

By~(***) and~(****), we see that both $v$ and $z$ are  incident to a single edge in $H_{Q_v}$. Therefore if $H_{Q_v}$ contains a cycle passing through $u$ and disjoint from $v$ and $z$ then this cycle must intersect $P_1^L$. Since no cycle is disjoint from $C$,
$H_{Q_v}$ is  a path with a single bridge $B_u$ which is a (subdivided) star with the center $u$. Unless $H_{Q_v}$ contains two edges attached to the opposite sides of $P_1^L$ at an even vertex or each at a different even vertex
we can remove  edges from  $B_u$ in order to convert  $H_{Q_v}$ into a (subdivision of) $K_{1,3}$ (apply the minimality of the counterexample and reinsert the deleted edges into the embedding). Otherwise, we remove all the edges from  $B_u$, except for the $u$-to-$P_1^L$ paths ending in the two special  edges attached to the opposite sides of $P_1^L$. Again, this is possible (apply the minimality of the counterexample and reinsert the deleted edges into the embedding). Now, we split $H_{Q_v}$ into two (subdivisions of) $K_{1,3}$ by doubling $P_1^L$ except at its end-vertices, which can be easily achieved while constructing an \iocro-drawing of the resulting graph. Indeed, every interior vertex of  $P_1^L$  is even at this point.

Similarly, if $H_{T_v}$ contains a cycle passing through $v$ and disjoint from $u$ and $z$ then this cycle must intersect  $P_1^L$ and analogously as in the previous paragraph we reduce $H_{T_v}$ to at most 2 subdivisions of $K_{1,3}$.
Now, any other $\{u,v,z\}$-bridge is easily seen to be (almost)  a subdivided $K_{1,3}$: a cycle through $u$ in such bridge that avoids $z$ and $v$ is disjoint from $Q_v$. Now, a cycle through $v$ in any other $\{u,v,z\}$-bridge  that avoids $u$ and $z$ is either disjoint from  $R_u^{LR}$ (contradiction), or we can again split the bridge into at most two $K_{1,3}$, and we are done.

\subsubsection*{Subcase 2. The vertex $w^R$ is incident to more than one leg of $H^{LR}$.}

If there exists at least three legs of $H^{LR}$ at $w^R$ they can be extended into pairwise interior disjoint paths ending on $P_1^L$. By Lemma~\ref{lem:X2-cuttwo}, they cannot all end in $z$, since $z$ is odd.  Hence, one of them ends in $z'\not= z$. It follows that both $(Q_v,T_v)$ and $(Q_u,T_u)$ form an $X$-configuration if we appropriately redirect $Q_v$ or $Q_u$ through $z'$.
Now, we  obtain a cycle through two of the legs  disjoint from $T_v$ or $T_u$  (contradiction with Lemma~\ref{lem:Xreduction333}).
Hence, using $(\circ)$ there exists at most two legs and therefore $w^R$ can be made even, but in this case we have a cycle nearly disjoint from $C$ (contradiction with Lemma~\ref{lem:Xreduction333}).

 We conclude that all $\{u,v,z\}$-bridges (except for the edges $\{uz,vz,uv\}$) are subdivisions of $K_{1,3}$
attached to $\{u,v,z\}$.


This completes the proof of $(ii)$ in the case that $P_1^L \cap  P_2^R$ contains at least one edge. \\

\subsubsection*{CASE (B1)- Establishing Property $(ii)$. $P_1^L$ and $P_2^R$ intersect in a single vertex.}

We are left with showing $(ii)$ in the case that $P_1^L$ and $P_2^R$ intersect in a single vertex $z$. If $z$
can be made even by flips, we are done by Lemma~\ref{lem:2disjoint}, since $R_u^{LR}$ and $Q_v$ touch in the even vertex $z$.
We conclude that $z$ is odd, and cannot be made even by edge-flips.

At this point we can complete the argument as in the case that $P_1^L$ and $P_2^R$ have an edge in common.  In this case, the previous argument applies unless there exists a path $S$ in $H^{LR}$ connecting $z$ and $w^L$; or $w'^R$ and $w'^R$ such that $w'^R$ is in $C_2$ if and only if $w''^R$ is in $C_2$.

 If there exits a pair of feet $w'^L$ and $w''^R$ of $H^{LR}$ such that $w'$ and $w''$ are both in  $C- C_2$ or both in $C_2$,
 and both $w'$ and $w''$ are even, possibly $w'=w''$, we remove all the legs of $H^{LR}$ incident to $C^R- C_2$  except for the pair attached at $w^L$ and $w'^R$.
 Then we double the part of $C$ passing through $C- C_2$ or $C_2$. Here, we again use the minimality of the counterexample when removing the edges. Hence, we obtain a  pair of $\{u,v,z\}$-bridges in $G$ that are (subdivisions of)  $K_{1,3}$. Otherwise, we turn $\{u,v,z\}$-bridges containing edges of $C$ (there can be at most two of them) into (subdivisions of) $K_{1,3}$ by edge removals.   The rest of the argument goes through as in the previous case since we don't need to use the $X$-configuration $(T_u,Q_u)$ or $(T_v,Q_v)$ at this point.

\subsubsection*{ CASE (B2) $C_1$ and $C_2$ agree on $C$.}

 \begin{figure}[htp]
\centering
\includegraphics[scale=0.7]{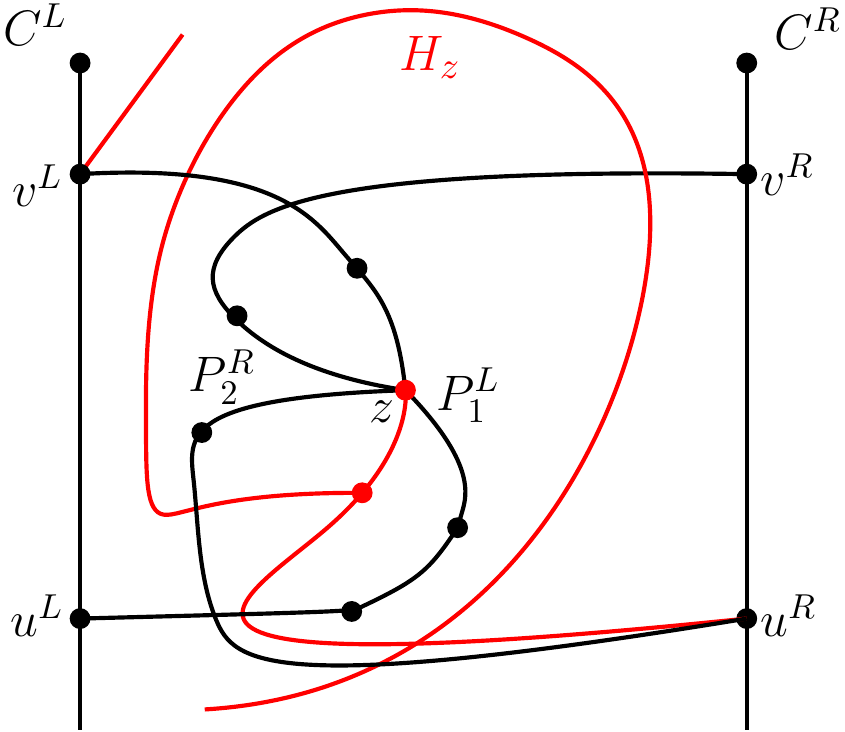}
 \caption{A $\{u,v,z\}$-bridge with both $u$ and $v$ as feet.}
\label{fig:uvzfeet}
\end{figure}

We consider the case that the subpaths of $C_1$ on $C^L$ and $C_2$ on $C^R$ are the same as subgraphs of $G$.

If the intersection of $P_1^L$ and $P_2^R$ is not a single vertex we will find a pair of vertex-disjoint paths in $C_1 \cup C_2$, contradicting Lemma~\ref{lem:2disjoint}.
Let $R_u^{LR}$ and $R_v^{LR}$\marginnote{$R_u^{LR}$, $R_v^{LR}$ \bfseries{(B2)}}
denote the paths joining $u^L$ and $u^R$, and $v^L$ and $v^R$, respectively, in $P_1^L\cup P_2^R$. Now, $C_1\oplus C_2=R_u \oplus R_v$ has vanishing homology over $\ZN_2$. Hence, $R_u$ and $R_v$ are in the same homology class over $\ZN_2$.
It follows that the order of end vertices of $P_1^L\cap P_2^R$ along $P_1^L$ and $P_2^R$ when traversing  $P_1^L$ and $P_2^R$, respectively, from $u^L$ to $v^L$, and from $u^R$ to $v^R$, is the same.
Hence, $R_u^{LR}$ and $R_v^{LR}$ are disjoint.
Since the cycles in $G$ corresponding to $R_u^{LR}$ and $R_v^{LR}$ are essential, this contradicts Lemma~\ref{lem:2disjoint}.

Therefore, we can assume that the intersection of $P_1^L$ and $P_2^R$ consists of a single vertex $z$.\marginnote{$z$ \bfseries{(B2)}}
If $z$ can be made even by flips, we do so.
In this case, $R_u^{LR}$ and $R_v^{LR}$ touch at $z$: The crossing parity
between $R_u^{LR}$ and $R_v^{LR}$ is the same as the crossing parity between $P_1^L$ and $P_2^R$. Now $P_1^L$ and $P_2^R$ (as part of two closed curves
$C_1$ and $C_2$) cross evenly, so $R_u^{LR}$ and $R_v^{LR}$ do too. The only edges of $R_u$ and $R_v$ which can cross oddly, are the edges incident to
$u$ and $v$, and the edges incident to $z$. The edges incident to $u$ and $v$ in $R_u^{LR}$ and $R_v^{LR}$ form two independent odd pairs, so together they do not affect the crossing parity  between $R_u^{LR}$ and $R_v^{LR}$. Thus,  $R_u$ and $R_v$
cross evenly overall. Now, if $R_u^{LR}$ and $R_v^{LR}$  ``cross'' in $z$, it would follow that the cycles corresponding to $R_u$ and $R_v$ are not in the same homology class over $\ZN_2$. However, we already showed the opposite in the previous paragraph. Hence, $R_u^{LR}$ and $R_v^{LR}$ touch at $z$, which contradicts Lemma~\ref{lem:2disjoint}.

We conclude that $z$ is odd and cannot be made even by flips. Let the edges of $R_u^{LR}$ incident to $z$ be $e$ and $f$. Correct
the rotation at $z$ so that $e$ and $f$ cross each other and every other edge at $z$ evenly. If $e$ and $f$
are consecutive in the rotation at $z$ (not separated by any edge in $G$), then $R_v^{LR}$ and $R_u^{LR}$ touch at $z$, and we are done, as before.
Therefore, there must be some edge $g$ separating $e$ and $f$. Since the graph is $2$-connected (by Lemma~\ref{lem:X2-connected}), there must be a
path $P'$ starting with $g$ at $z$, and connecting $z$ to $C$ ($P'$ avoids the interior vertices of $H^{LR}$, since otherwise,
we would contradict Lemma~\ref{lem:Xreduction333}). By inspecting the rotation at $z$ (note that we made edges of $P_1^L$ and $P_2^R$ cross each other evenly at $z$, see Figure~\ref{fig:uvzfeet}), the vertices of this path lie both in the interior of $C_1$ and $C_2$, so $P'$ can only attach to $C$ by crossing $C_1$ or $C_2$ oddly, which it can only do if it connects to $u^L$, $u^R$, $v^L$, or $v^R$. So $g$ is part of a
$\{u,v,z\}$-bridge $H_z$. If both $u$ and $v$ are feet of $H_z$ (which can happen, see Figure~\ref{fig:uvzfeet}), then $H_z$ contains a path
from $u$ to $v$ avoiding $z$ and disjoint from $P_2^R$. We can therefore treat this as {\bf CASE (A)} or {\bf CASE (B1)} instead depending on whether $H_z^{LR}$ attaches at the both $C^L$ and $C^R$.  Hence, any such $\{u,v,z\}$-bridge
is either a $\{u,z\}$-bridge or a $\{v,z\}$-bridge.
By the choice of $P_1^L$ and $P_2^R$ minimizing the sum of their lengths, which we assumed at the beginning of the the proof, the $\{u,z\}$-bridge or  $\{v,z\}$-bridge $H_z$  is not an edge and therefore feet of $H_z$   form a  2-cut.  Since $z$ is odd,
Lemma~\ref{lem:X2-cutodd} implies that there is no such $H_z$ (since $P_1^L \cup P_2^R$ already contains paths from $z$ to each of $u^L$, $u^R$, $v^L$, and $v^R$; we work with $X$-configurations $(C,R_u)$ and $(C,R_v)$ to make sure the lemma can be applied). Hence $e$ and $f$ are consecutive after all, and that case we already dealt with.

\section{Applications and Open Questions}\label{sec:A}

We survey results which rely on the Hanani-Tutte theorem, and discuss how Theorem~\ref{thm:HTtorus} can lead to toroidal versions of some of these results; along the way, we also encounter some open questions, and make some conjectures. This extends the discussion in~\cite{S14}).

\subsection*{Adjacent Crossings}

Theorem~\ref{thm:HTtorus} implies that if a graph can be drawn on the torus so that the only crossings are between adjacent edges, then it is toroidal. This would appear to be intuitively clear, but for the plane, for the projective plane, and now for the torus, we only know it to be true by virtue of their respective Hanani-Tutte theorems; even for the plane no simpler proof is known. Since we know that the Hanani-Tutte theorem does not hold for orientable surfaces of genus at least $4$~\cite{fulek2019counterexample}, this begs the question whether we can always remove adjacent crossings.

\begin{conjecture}
  If a graph can be drawn in a surface without any independent crossings, then the graph is embeddable in that surface.
\end{conjecture}

The counterexample from~\cite{fulek2019counterexample} requires independent edges to cross (evenly), so it does not resolve the conjecture.

\subsection*{Crossing Number Variants}

The Hanani-Tutte theorem is related to the {\em independent odd crossing number} $\iocr(G)$ of a graph $G$, the fewest number of pairs of independent edges that have to cross oddly in a drawing of $G$ (see~\cite[Section 6]{PT00}). The Hanani-Tutte theorem then states that $\iocr(G) = 0$ implies that $\cro(G) = 0$, where $\cro(G)$ is the traditional crossing number of $G$. It is even true that $\iocr(G) = \cro(G)$ for $\iocr(G) \leq 2$~\cite{PSS10}, though we know that there are graphs $G$ for which $\iocr(G) < \cro(G)$~\cite{PSS08}. The two crossing numbers cannot be arbitrarily far apart though, one can show that $\cro(G) \leq \binom{2\iocr(G)}{2}$~\cite{PSS10}. It is not known whether similar bounds, or any bounds, for that matter, hold for other surfaces, not even the projective plane. (The subscript in  crossing number variants indicates the surface we work on.)

\begin{conjecture}
  $\cro_{\Sigma}(G) \leq \binom{2\iocr_{\Sigma}(G)}{2}$, where $\Sigma$ is the projective plane, or the torus.
\end{conjecture}

We also now have a crossing lemma for $\iocr$ on the torus.

\begin{corollary}[Crossing Lemma]\label{cor:iocrcl}
  $\iocr_T(G) \geq c m^3/n^2$,  where $T$ is the torus, $c = 1/64$, $m = |E(G)|$, and $n = |V(G)|$.
\end{corollary}

The proof follows from the standard crossing lemma argument done carefully combined with Theorem~\ref{thm:HTtorus} (see the section on crossing lemma variants in~\cite{S13b}). Since $\iocr$ lower bounds several other crossing number variants, such as $\iacr$, $\pcr$, and $\pcr_-$ this also implies a crossing lemma for these variants. Also see~\cite{K20} for a sketch of an argument that shows a crossing lemma for $\iocr$ for arbitrary surfaces, but without explicit constant $c$.

\subsection*{Pseudodisks and Admissible Regions}

Smorodinsky and Sharir~\cite{SS04} showed that if $P$ is a collection of $n$ points, and $C$ a collection of $m$ pseudo-disks in the plane such that every pseudo-disk in $C$ passes through a distinct pair of points in $P$, and no pseudo-disk contains a point of $P$ in its interior, then $m \leq 3n-6$, where $3n-6$ is just the maximal number of edges in a planar graph. For pseudodisks in the projective plane, we have $m \leq 3n-3$~\cite{S14}, and, using Theorem~\ref{thm:HTtorus} we can now conclude that $m \leq 3n$ for pseudodisks in the torus.

A family of simply connected regions is {\em $k$-admissible} if each pair of region boundaries  intersect in an even number of points, not exceeding $k$.  Whitesides and Zhao~\cite{whitezhao}
showed that the union of $n \geq 3$ planar $k$-admissible sets is bounded by at most $k(3n-6)$ arcs, where, again, $3n-6$ is
the bound on the number of edges in a planar graph. This result was reproved by Pach and Sharir~\cite{PS99} using
the Hanani-Tutte theorem. Consequently, we get a bound of $k(3n-3)$ for regions in the projective plane~\cite{S14},
and $3kn$ for regions in the torus.

Keszegh~\cite{K19} gave a more uniform treatment of these types of results based on hypergraphs, which at the core again uses the Hanani-Tutte theorem. While he only states results for the plane, all of his material should lift to the projective plane and the torus based on the availability of the Hanani-Tutte theorem on those surfaces.

\subsection*{Surfaces and Pseudosurfaces}

We know that, at least for orientable surfaces, Hanani-Tutte characterizations will stop at genus $4$. This still leaves a fair number of open cases, and our general set-up may be useful in exploring those cases. For example, we may ask whether our proof can be adjusted to cover the projective plane (this would give a third proof, after~\cite{PSS09} and~\cite{CdVKP17}), or extend it to the Klein bottle (the first case where there are non-trivial surface separating cycles), or the spindle pseudosurface with $n$ pinchpoints.

The graph minor theorem for surfaces implies that embeddability in a surface can be characterized by a finite set of forbidden minors. For the plane and the projective plane, proofs of the Hanani-Tutte theorem were possible, because the list
of forbidden minors is known explicitly~\cite{PSS09}. For the torus, this is no longer the case (though progress is being made on all cubic obstructions). One may ask, whether it is possible to go backwards: from the Hanani-Tutte theorem to
the set of forbidden minors? For the plane this was done by van der Holst~\cite{vdH07}. Can this work be extended
to the $1$-spindle, the projective plane (where~\cite{CdVKP17} gave a proof of Hanani-Tutte not using forbidden minors), or even the torus?

The proof of Theorem~\ref{thm:HTtorus} is algorithmic in the sense that given an \iocro-drawing of $G$ we can
find an embedding of $G$ in the torus (all reductions in the lemmas are constructive like that, so the minimal counterexample assumption can be turned into a recursion). This does not imply that testing whether a graph
can be embedded in the torus lies in polynomial time (see~\cite[Section~1.4.2]{S14}). The planar Hanani-Tutte theorem
can be turned into a linear system of equations over $GF{2}$, which is solvable in polynomial time.
Unfortunately, modeling the handle of the torus requires quadratic equations, which
loses us polynomial-time solvability. This is similar to the situation for the projective plane,
where the projective handle also leads to quadratic equations.

\begin{question}
 Can the Hanani-Tutte criteria for the projective plane and the torus be turned into a polynomial-time test?
\end{question}

For the $1$-spindle this is easily possible, but not quite surprising.

Such tests would not be competitive with existing algorithms---the running time in the plane is $O(n^6)$, but they have the potential to be significantly simpler (the planar Hanani-Tutte test is).

We have focussed on the Hanani-Tutte theorem over $\ZN_2$, that is, we count crossings by parity. We can also count crossings algebraically, over $\ZN$, by assigning a crossing the value $+1$ or $-1$ based on the direction of the crossing (left-to-right, or right-to-left, for an arbitrary orientation of the edges). The counterexample to the Hanani-Tutte theorem on surfaces of genus $4$ is for the $\ZN_2$-version~\cite{fulek2019counterexample}; does the Hanani-Tutte theorem hold algebraically, over $\ZN$? Similarly, we saw that the unified Hanani-Tutte theorem, Theorem~\ref{thm:HTS}, due to ~\cite{FKP17} fails for the torus. Does the algebraic variant, over $\ZN$, hold on the torus?

Finally, the Hanani-Tutte theorem has some variants in the plane which have not yet been generalized. For example, if a graph can be drawn in the plane so that every cycle in the graph has an even number of independent self-crossings, then the graph is planar~\cite[Theorem 1.16]{S14}. Is this result still true for the projective plane, or the torus? The proof in the plane uses the Kuratowski minors.

\subsection*{Arf and Approximating the Hanani-Tutte Theorem}

The fact that the Hanani--Tutte theorem cannot be extended to all orientable surfaces  has some positive practical consequences. Some results about graphs embedded on a surface remain true if we only require that the graph has an independently even drawing on the surface, and this is a strictly larger class of graphs in general (starting at genus $4$).

One notable example is the Arf invariant formula for the number of perfect matchings in a graph embedded on an orientable surface~\cite[Remark 1.4, Theorem 1]{LM11}\footnote{The authors attribute the extension of the result  to independently even drawings  to Norine. The extension is even more general, since instead of independently even drawings Norine considers, so-called (perfect) matching even drawings, in which every perfect matching induces an even number of crossings. It is easy to see that independently even drawings form a proper sub-class of matching even drawings on every surface.}, see also~\cite{cimasoni2007dimers,norine2008pfaffian}.
The proof of the similar result by Tesler~\cite{tesler2000matchings}, which applies to all closed surfaces,
still works for independently even drawings instead of embeddings. In both cases, the complexity of the formula depends exponentially on the genus of the underlying surface, so for graphs which have an
independently even drawing in a surface in which they cannot be embedded, the complexity of the formula is improved; its length does not depend on the genus of the graph, but rather its {\em $\ZN_2$-genus}, the smallest genus of a surface on which the graph has an independently even drawing.

It is therefore interesting to study how large the gap between the genus and the $\ZN_2$-genus of a graph may be. It is known that the gap can be at least linear~\cite[Corollary 4.1]{fulek2019counterexample}. There also is an upper bound on this gap, but it is not effective~\cite{FK18}.

\bibliographystyle{plain}
\bibliography{HTTorus}

\end{document}